%% file: quantum.tex
\documentclass[a4paper,onecolumn,11pt,accepted=2024-04-02]{quantumarticle}
\pdfoutput=1
\usepackage[utf8]{inputenc}
\usepackage[english]{babel}
\usepackage[T1]{fontenc}
\usepackage[usenames,dvipsnames]{xcolor}
\usepackage{hyperref}
\usepackage{tikz}

\usepackage{amsmath,amsfonts,amssymb,amsthm}
\usepackage{subcaption}
\usepackage{enumitem}
\usepackage{mathtools}
\usepackage{bm}
\usepackage{mleftright}
\usepackage{booktabs}
\usepackage{color}
\hypersetup{colorlinks=true,
    linkcolor=BrickRed,
    citecolor=OliveGreen,
    urlcolor=MidnightBlue}
    
\usepackage{cleveref}

\usepackage{siunitx}
\usepackage{pgfplots}
\usepackage{tikz}
\pgfplotsset{compat=1.9}
\usepgfplotslibrary{fillbetween}

\usepackage{algorithm,algpseudocode}
\makeatletter
\def\ALG@special@indent{%
    \ifdim\ALG@thistlm=0pt\relax
        \hskip-\leftmargin
    \else
        \hskip\ALG@thistlm
    \fi
}
\newcommand{\Input}[1]{\item[]\noindent\ALG@special@indent \textbf{Input:}\ #1}
\newcommand{\Output}[1]{\item[]\noindent\ALG@special@indent \textbf{Output:}\ #1}
\newcommand{\Indent}[1]{\item[]\noindent\ALG@special@indent \hspace{2.675em} #1}
\makeatother

\DeclareMathOperator*{\minimize}{min}
\DeclareMathOperator*{\maximize}{max}
\DeclareMathOperator*{\subjto}{subj. to}

\DeclareMathOperator*{\argmax}{arg\,max}
\DeclareMathOperator*{\argmin}{arg\,min}
\DeclareMathOperator{\proj}{proj}
\DeclareMathOperator{\tr}{tr}
\DeclareMathOperator{\interior}{int}
\DeclareMathOperator{\relinterior}{relint}
\DeclareMathOperator{\domain}{dom}
\DeclareMathOperator{\diag}{diag}

\DeclareMathOperator{\vecc}{vec}

\DeclarePairedDelimiterX{\divx}[2]{(}{)}{#1\mspace{1.5mu}\delimsize\|\mspace{1.5mu}#2}
\DeclarePairedDelimiterX{\divy}[2]{(}{)}{#1\mspace{1mu}\delimsize|\mspace{1mu}#2}
\DeclarePairedDelimiterX{\inp}[2]{\langle}{\rangle}{#1, #2}
\DeclarePairedDelimiterX{\norm}[1]{\lVert}{\rVert}{#1}
\DeclarePairedDelimiterX{\abs}[1]{\lvert}{\rvert}{#1}
\DeclarePairedDelimiterX{\bk}[2]{\langle}{\rangle}{#1 \delimsize\vert #2}

\newcommand*\wc{{\mkern 2mu\cdot\mkern 2mu}}

\NewDocumentCommand{\grad}{e{_^}}{%
  \mathop{}\!
  \nabla
  \IfValueT{#1}{_{\!#1}}
  \IfValueT{#2}{^{#2}}
}

\newtheorem{thm}{Theorem}[section]
\newtheorem{lem}[thm]{Lemma}
\newtheorem{prop}[thm]{Proposition}
\newtheorem{cor}[thm]{Corollary}
\newtheorem{defn}[thm]{Definition}
\newtheorem{rem}[thm]{Remark}
\newtheorem{exmp}[thm]{Example}
\newtheorem{fact}[thm]{Fact}
\newtheorem{assump}{Assumption}

\usepackage[
    backend=bibtex,
    sorting=none,
    style=numeric
]{biblatex}
\begin{document}

\title{Efficient Computation of the Quantum Rate-Distortion Function}

\author{Kerry He}
\affiliation{Department of Electrical and Computer System Engineering, Monash University, Clayton VIC 3800, Australia}
\orcid{0000-0003-4052-969X}
\author{James Saunderson}
\affiliation{Department of Electrical and Computer System Engineering, Monash University, Clayton VIC 3800, Australia}
\orcid{0000-0002-5456-0180}
\author{Hamza Fawzi}
\affiliation{Department of Applied Mathematics and Theoretical Physics, University of Cambridge, Cambridge CB3 0WA, United Kingdom}
\maketitle

\begin{abstract}
    The quantum rate-distortion function plays a fundamental role in quantum information theory, however there is currently no practical algorithm which can efficiently compute this function to high accuracy for moderate channel dimensions. In this paper, we show how symmetry reduction can significantly simplify common instances of the entanglement-assisted quantum rate-distortion problems. This allows us to better understand the properties of the quantum channels which obtain the optimal rate-distortion trade-off, while also allowing for more efficient computation of the quantum rate-distortion function regardless of the numerical algorithm being used. Additionally, we propose an inexact variant of the mirror descent algorithm to compute the quantum rate-distortion function with provable sublinear convergence rates. We show how this mirror descent algorithm is related to Blahut-Arimoto and expectation-maximization methods previously used to solve similar problems in information theory. Using these techniques, we present the first numerical experiments to compute a multi-qubit quantum rate-distortion function, and show that our proposed algorithm solves faster and to higher accuracy when compared to existing methods.
\end{abstract}

\section{Introduction}
The rate-distortion function quantifies how much a signal can be compressed by a lossy channel while remaining under a maximum allowable distortion of the signal. This important tool in information theory was originally formulated for the classical setting in Shannon's seminal work~\cite{shannon1948mathematical}, and has recently been extended to the quantum setting~\cite{datta2012quantum,wilde2013quantum}. More concretely, in the quantum setting, the signal can be represented by an $n\times n$ positive semidefinite Hermitian matrix $\rho\in\mathbb{H}^n_+$ with unit trace, and the distortion can be quantified using a matrix $\Delta\in\mathbb{H}^{mn}_+$. The corresponding entanglement-assisted quantum rate-distortion function $R(\wc; \rho, \Delta):\mathbb{R}_+\rightarrow\mathbb{R}$ is then given by
\begin{align}\label{eqn:intro}
    R(D; \rho, \Delta) = \minimize_{\sigma \in \mathbb{H}^{mn}_+} \{I(\sigma; \rho) : \tr[\sigma\Delta] \leq D, \tr_B(\sigma) = \rho \}.
\end{align}
Here, $D\geq0$ represents the maximum allowable distortion of the signal, $\tr_B:\mathbb{H}^{mn}\rightarrow\mathbb{H}^n$ is a linear operator called the partial trace (see Section~\ref{sec:qit}), and $I(\wc; \rho)$ is a reparameterization of the quantum mutual information and defined as
\begin{equation*}
    I(\sigma; \rho) = \tr[\sigma \log(\sigma)] - \tr[\tr_R(\sigma) \log(\tr_R(\sigma))] - \tr[\rho\log(\rho)],
\end{equation*}
where $\log$ denotes the matrix logarithm and $\tr_R:\mathbb{H}^{mn}\rightarrow\mathbb{H}^m$ is another partial trace operator. We will introduce this problem in more detail in Section~\ref{sec:rd}.

In general, problem~\eqref{eqn:intro} does not admit a closed-form solution, and therefore numerical optimization algorithms are required to compute values of the quantum rate-distortion function. However, it is not immediately obvious how to solve such a problem efficiently. Functions involving logarithms often do not have Lipschitz-continuous gradients nor Hessians, which are standard assumptions required to prove global convergence rates of gradient descent-like methods. Therefore, it is not obvious how quickly gradient descent-like methods will converge to the solution, if at all. Additionally, a na\"ive projected gradient descent approach is not suitable as the algorithm typically produces iterates on the boundary of the feasible set, where the gradient of logarithms are not well-defined. 

The Blahut-Arimoto algorithm~\cite{blahut1972computation,arimoto1972algorithm} is one of the most well-known methods in information theory which has been applied to solve various information-theoretic optimization problems, including evaluating the classical rate-distortion function. Recently, we showed in~\cite{he2023mirror} how the Blahut-Arimoto algorithm is related to mirror descent, a well studied type of Bregman proximal method~\cite{nemirovskij1983problem,beck2003mirror,tseng2008accelerated,beck2017first}, and how relative smoothness analysis~\cite{bauschke2017descent, lu2018relatively, teboulle2018simplified} could be used to prove convergence rates of the algorithm when applied to various problems in quantum information theory. However, finding an efficient analog of these algorithms for the quantum rate-distortion problem is not straightforward due to the more complicated matrix-valued constraints, and the fact that non-diagonal matrices no longer necessarily commute with each other.

In this paper, we show how we can apply the mirror descent framework to compute the quantum rate-distortion function. Although the mirror descent iterates do not admit closed-form expressions in general when applied to the quantum rate-distortion problem, in this paper we explore how they can still be efficiently computed. The main novelty of our work is showing how we can exploit inherent symmetries of common quantum rate-distortion instances to better understand the structural properties of quantum states and channels obtaining the optimal rate-distortion trade-off. Moreover, we show how this knowledge can be leveraged to gain considerable improvements to the proposed optimization algorithm by restricting the problem to a smaller subspace. A secondary contribution is showing how the mirror descent iterates for solving the quantum rate-distortion problem, which are expressed as an optimization subproblem, can be efficiently computed using a combination of duality and inexact computation of the iterates while retaining convergence guarantees. We note that our resulting algorithm is similar to~\cite{hayashi2022bregman}, which proposes similar iterates to solve the quantum rate-distortion problem, but does not show how to implement certain steps of the algorithm. Altogether, we present a fully realizable method for computing the quantum rate-distortion function which outperforms all existing methods.

\paragraph{Related work}
\sloppy In~\cite{hayashi2022bregman}, it was shown how a method derived using expectation-maximization, which is related to the Blahut-Arimoto algorithm~\cite{hayashi2023minimization}, could solve classical and quantum rate-distortion problems by solving a convex subproblem at every iteration. The author shows how one can account for inexact computation of these subproblems for classical rate-distortion problems while retaining convergence guarantees. However, it is not shown how to explicitly implement these same inexact computations for the quantum rate-distortion problem. We fill in this gap between theory and application in this work by proposing an implementable inexact stepping procedure for the quantum rate-distortion problem. In~\cite{he2023mirror}, we showed how an instance of the primal-dual hybrid gradient algorithm could be used to compute the classical and quantum rate-distortion functions. However, this method converges at a sublinear rate, and therefore can be slow to converge to high-accuracy solutions. We show in Appendix~\ref{appdx:linear} that our algorithm obtains linear convergence for common instances of the quantum-rate distortion problem. Alternatively, we can formulate the quantum rate-distortion problem as an instance of a quantum relative entropy program (see, e.g.,~\cite{chandrasekaran2017relative,fawzi2018efficient}), which we can solve using existing software which supports this problem class, such as \textsc{CvxQuad}~\cite{fawzi2019semidefinite}, Hypatia~\cite{coey2023performance}, or DDS~\cite{karimi2020primal,karimi2023efficient}. However, these methods do not scale well to large problem dimensions. \textsc{CvxQuad} uses large semidefinite approximations of matrix logarithms, while Hypatia and DDS use second-order interior point algorithms. Therefore, they are limited to computing the rate-distortion function of small quantum channels. In contrast, the mirror descent algorithm is a first-order method which is better at scaling to large problem dimensions.

\paragraph{Contributions}
Our main contributions are threefold. 
\begin{itemize}
    \item First, we show that for certain common distortion measures, the corresponding rate-distortion problems possess symmetries that allow us to better understand the structure of quantum states and channels which obtain the optimal rate-distortion trade-off by using a symmetry reduction technique. For the quantum rate-distortion problem with entanglement fidelity distortion measure and maximally mixed input state, we show in Theorem~\ref{thm:qrd-maxmix-solution} that there are sufficient symmetries which allow us to directly obtain an explicit solution to the rate-distortion function. For some other common quantum-rate distortion instances, we show how symmetry can be exploited to significantly reduce the dimensionality of the optimization problem. This results in improvements to both computational complexity and memory requirements, and is independent of the optimization algorithm used.
    \item Second, we show that the expectation-maximization approach of~\cite{hayashi2022bregman} can be interpreted as a mirror descent algorithm, and that we can rederive the algorithm and its analysis in an elementary way using tools from convex optimization.
    \item Third, unlike~\cite{hayashi2022bregman}, we present several concrete strategies for how to solve the mirror descent subproblems while retaining convergence guarantees for the quantum rate-distortion problem. We do this primarily by introducing an inexact method, based on~\cite{yang2022bregman}, which improves on~\cite{hayashi2022bregman} by allowing the tolerance at which we compute each iterate at to converge to zero with each subproblem, and by using an explicitly computable error criterion. This results in a more efficient algorithm which can converge sublinearly to the solution, rather than only converging to a neighborhood around the solution.

\end{itemize}

Our main algorithm is presented in Algorithm~\ref{alg:inexact-qrd}, which we provide an implementation of with symmetry reduction techniques at
\begin{center}
    \url{https://github.com/kerry-he/efficient-qrd}.
\end{center} 
We present numerical experiments on quantum channels of up to $9$-qubits to demonstrate the scalability of our approach. To the best of our knowledge, no other works have presented numerical experiments computing the quantum rate-distortion function for more than a single-qubit channel (e.g.,~\cite{datta2013quantum}, which estimates the curve using a sampling-based method). We note that unlike~\cite{he2023mirror,hayashi2022bregman}, we do not directly account for the distortion constraint. However, it is relatively straightforward to extend most of our analysis and algorithm to account for it.

Throughout the paper, we predominantly study the quantum rate-distortion function using the entanglement fidelity distortion measure. This is generally considered as the canonical distortion measure first used by the foundational works~\cite{barnum2000quantum,datta2012quantum,wilde2013quantum} (see also~\cite{khanian2021rate,khanian2023rate}) as it represents a natural way of measuring the dissimilarity between two quantum states, and possesses nice properties which makes it useful when analyzing the quantum rate-distortion function. We note that although the symmetry reduction in Section~\ref{sec:symm} is dependent on the choice of distortion measure (we discuss how symmetry reduction can be applied to other common distortion measures in Section~\ref{sec:other-dist}), the inexact mirror descent algorithm and its convergence analysis introduced in Section~\ref{sec:alg} is independent of this choice. This means the inexact mirror descent method can be applied for any distortion measure, however each iteration may be more expensive to compute if we cannot apply symmetry reduction.

\subsection{Outline}

The remainder of the paper is structured as follows. In Section~\ref{sec:prelim}, we summarize the notation we use, as well as preliminaries about quantum information theory and mirror descent. In Section~\ref{sec:rd}, we formalize the quantum rate-distortion problem, and explain how we can reformulate the optimization problem to make it easier to solve and analyze. In Section~\ref{sec:symm}, we recall symmetry reduction and show how it can be applied to common quantum rate-distortion instances. In Section~\ref{sec:alg}, we introduce the inexact mirror descent algorithm we propose to solve the rate-distortion problem with. In Section~\ref{sec:exp}, we present experimental results comparing our proposed algorithm to existing methods. Finally, we end with concluding remarks in Section~\ref{sec:conc}. We provide some auxiliary results in the appendix, including how our results, analysis and techniques can be applied to the classical rate-distortion in Appendix~\ref{appdx:crd}, local linear convergence rates in Appendix~\ref{appdx:linear}, and derivations for bounds on the optimality gap of mirror descent in Appendix~\ref{appdx:lb}.

\section{Preliminaries} \label{sec:prelim}

\subsection{Notation}

We denote arbitrary finite-dimensional inner product spaces by $\mathbb{V}$ and $\mathbb{V}'$, the set of $n\times n$ Hermitian matrices as $\mathbb{H}^n$ with inner product $\inp{X}{Y}=\tr[XY^\dag]$, the $n$-dimensional non-negative orthant and $n\times n$ positive semidefinite cone as $\mathbb{R}_+^n$ and $\mathbb{H}^n_+$ respectively, and their respective interiors as $\mathbb{R}_{++}^n$ and $\mathbb{H}^n_{++}$. For $X,Y\in\mathbb{H}^n$, we use $X \succeq Y$ to mean $X-Y \in \mathbb{H}^n_+$, and $X \succ Y$ to mean $X-Y \in \mathbb{H}^n_{++}$. We denote the extended reals by $\bar{\mathbb{R}}=\mathbb{R}\cup \infty$, and the set of positive integers as $\mathbb{N}$. We use $\domain$ to denote the domain of a function, $\interior$ and $\relinterior$ to denote the interior and relative interior of a set, respectively, $\ker$ to denote the kernel of a linear operator, $\dim_{\mathbb{C}}$ and $\dim_{\mathbb{R}}$ to denote the dimension of a complex and real vector space respectively, $\mathbb{I}$ to denote the identity matrix, and $\otimes$ to denote the Kronecker product between two matrices. For a complex matrix $X\in\mathbb{C}^{n\times m}$, we use $\overline{X}$ to denote its entry-wise complex conjugate, and $X^\dag$ to denote its conjugate transpose. For a linear operator $\mathcal{A}:\mathbb{V}\rightarrow\mathbb{V}'$, we denote its adjoint by $\mathcal{A}^\dag:\mathbb{V}'\rightarrow\mathbb{V}$. We use $\{e_i\}$ to denote the standard basis for the vector space $\mathbb{C}^n$.

\subsection{Quantum information theory} \label{sec:qit}
Here, we introduce some preliminary concepts relating to quantum information theory. We refer the reader to e.g.,~\cite{nielsen2010quantum,wilde2017quantum,watrous2018theory} for some standard references for these topics.

\paragraph{States and channels}

We define $n\times n$ quantum states as $n \times n$ positive semidefinite Hermitian matrices with unit trace, i.e.,
\begin{equation*}
    \mathcal{D}_n \coloneqq \{\rho \in \mathbb{H}^n_+ : \tr[\rho] = 1\}.
\end{equation*}
These matrices are also commonly referred to as \emph{density matrices}. \emph{Quantum channels} are completely positive trace preserving (CPTP) (see e.g.,~\cite[Definitions 4.4.2 and 4.4.3]{wilde2017quantum}) linear operators which map $n\times n$ quantum states to $m\times m$ quantum states. We denote the set of these CPTP linear operators as $\Phi_{m,n}$. It turns out that any CPTP linear map $\mathcal{N}:\mathbb{H}^n\rightarrow\mathbb{H}^m$ in $\Phi_{m,n}$ can be written in the form
\begin{equation}\label{eqn:kraus-choi}
    \mathcal{N}(\rho) = \sum_{i=1}^k V_i \rho V_i^\dag,
\end{equation}
where $k\leq mn$ and the $V_i\in\mathbb{C}^{m\times n}$ (for $i=1,\ldots,k$) are matrices satisfying $\sum_{i=1}^k V_i^\dag V_i = \mathbb{I}$~\cite[Theorem 4.4.1]{wilde2017quantum}. See also Definition~\ref{defn:choi} for an alternate characterization of these linear maps.

\paragraph{Composite systems}
An $n$-dimensional \emph{quantum system} is represented mathematically as an $n$-dimensional complex vector space $\mathbb{C}^n$, which we can associate with the set of density matrices $\mathcal{D}_n$ which act on this space. A \emph{bipartite system} is obtained by taking the tensor product between two quantum systems, e.g., the tensor product between an $n$-dimensional system and an $m$-dimensional system produces an $nm$-dimensional system, which we can associate with the set of density matrices $\mathcal{D}_{mn}$ which act on this joint system. To aid in keeping track of these composite systems, we will denote quantum systems with letters, e.g., $A,B$, and the tensor product of these two spaces as, e.g., $AB$. We will use subscripts to clarify that a density matrix acts on a particular system. For instance, we write $\rho_A$ to indicate that the quantum state acts on system $A$. 

To make the concept of composite quantum systems more concrete, we note that the Kronecker product between any two density matrices $\rho_A\in\mathcal{D}_m$ and $\rho_B\in\mathcal{D}_n$ on systems $A$ and $B$ respectively will produce a density matrix $\rho_A\otimes\rho_B\in\mathcal{D}_{mn}$ on system $AB$. This operation is analogous to forming the joint distribution of two independent probability distributions. We also consider the ``marginalization'' operation, called the \emph{partial trace}, which is a linear operator which maps density matrices on the joint system $AB$ to density matrices on the individual systems $A$ or $B$. This is analogous to finding the marginal distribution from a joint distribution. We denote the partial trace as $\tr_A:\mathbb{H}^{mn}\rightarrow\mathbb{H}^n$ and $\tr_B:\mathbb{H}^{mn}\rightarrow\mathbb{H}^m$, and define them as the adjoint of the Kronecker product with the identity matrix, i.e., $(\tr_A)^\dag(\rho_B)=\mathbb{I}_A\otimes\rho_B$ and $(\tr_B)^\dag(\rho_A)=\rho_A\otimes\mathbb{I}_B$. More concretely, if we represent a bipartite density matrix $\rho_{AB}\in\mathcal{D}_{mn}$ on system $AB$ as a block matrix consisting of an $m\times m$ array of $n\times n$ blocks $X_{ij}$, then
\begin{equation*}
    \tr_A(\rho_{AB}) = \sum_{i=1}^m X_{ii} \qquad \textnormal{and} \qquad [\tr_B(\rho_{AB})]_{ij} = \tr[X_{ij}].
\end{equation*}
We will also use the tensor product between linear operators. In particular, for an operator $\mathcal{N}\in\Phi_{m, n}$ given by~\eqref{eqn:kraus-choi}, we define $(\mathcal{N}\otimes\mathbb{I})\in\Phi_{(m+p), (n+p)}$, where $\mathbb{I}$ is of size $p\times p$, as 
\begin{equation*}
    (\mathcal{N}\otimes\mathbb{I})(\rho) = \sum_{i=1}^k (V_i\otimes\mathbb{I}) \rho (V_i\otimes\mathbb{I})^\dag.
\end{equation*}

\paragraph{Purification}
Consider a quantum state $\rho_A\in\mathcal{D}_n$ with spectral decomposition $\rho_A=\sum_{i=1}^n \lambda_i v_iv_i^\dag$ defined on an $n$-dimensional system $A$, and a system $R$ which is isomorphic to $A$. A \emph{purification} of $\rho_A$ is a complex unit vector $\psi\in\mathbb{C}^{n^2}$ defined on the joint system $AR$ which satisfies $\tr_R(\psi\psi^\dag)=\rho_A$. The purification is not unique. However, in this paper we will consider the particular choice of purification
\begin{equation}\label{eqn:purificaiton}
    \psi = \sum_{i=1}^n \sqrt{\lambda_i}\;v_i\otimes v_i,
\end{equation}
and refer to this as the purification of $\rho_A$. With some further abuse of terminology, we will also refer to the quantum state associated with~\eqref{eqn:purificaiton}, i.e., $\rho_{AR}\coloneqq \psi\psi^\dag \in \mathcal{D}_{n^2}$, as the purification of $\rho_A$. Note that this choice of purification satisfies $\rho_R\coloneqq\tr_A(\rho_{AR}) = \tr_R(\rho_{AR}) = \rho_A$.

\paragraph{Entropy}
For a function $f: \mathbb{R}\rightarrow\bar{\mathbb{R}}$ and Hermitian matrix $X\in\mathbb{H}^{n}$ with spectral decomposition $X=\sum_{i=1}^n\lambda_i v_iv_i^\dag$, the function $f$ can be extended to matrices $X$ as
\begin{equation*}
    f(X) \coloneqq \sum_{i=1}^nf(\lambda_i) v_iv_i^\dag.
\end{equation*}
With some overloading of function notation, we define the von Neumann entropy $S:\mathbb{H}_+^n\rightarrow\mathbb{R}$ as
\begin{equation}
    S(\rho)\coloneqq-\tr[\rho \log(\rho)],
\end{equation}
and the quantum relative entropy $S:\mathbb{H}_+^n\times\mathbb{H}_{++}^n\rightarrow\mathbb{R}$ as
\begin{equation}
    S\divx{\rho}{\sigma}\coloneqq\tr[\rho (\log(\rho) - \log(\sigma)].
\end{equation}
For notational convenience, we use the natural base $e$ throughout the paper for logarithms and exponentials. However, all computational experiments are performed in base $2$ to obtain results measured in bits, i.e., the standard unit of measurement for information. This is done by rescaling relevant quantities by a factor of $\log_2(e)$.

\subsection{Mirror descent}

We will focus on constrained convex optimization problems of the form
\begin{subequations} \label{eqn:constr-min}
    \begin{align}
        \minimize_{x} \quad & f(x)  \\
        \subjto \quad & \mathcal{A}(x) = b,
    \end{align}
\end{subequations}
where $f:\mathbb{V}\rightarrow\bar{\mathbb{R}}$ is a convex function which is differentiable on $\interior\domain f$, $\mathcal{A}:\mathbb{V}\rightarrow\mathbb{V}'$ is a linear operator, and $b\in\mathbb{V}'$. Throughout the rest of the paper, we define 
\begin{equation}
    \mathcal{C}\coloneqq\{ x\in\domain f : \mathcal{A}(x) = b \},
\end{equation}
as the intersection of the implicit and explicit constraints of~\eqref{eqn:constr-min}.

Before introducing the mirror descent algorithm, we first introduce the class of Legendre functions, which are known to satisfy certain nice properties. We adopt the terminology in~\cite{rockafellar1970convex} for common terms relating to convex functions. 
\begin{defn}[{{\cite[Definition 2.1]{teboulle2018simplified}}}]
    A function $\varphi: \mathbb{V}\rightarrow\bar{\mathbb{R}}$ is \emph{Legendre} if it is proper, lower semicontinuous, strictly convex, and essentially smooth.
\end{defn}
\begin{fact}[{{\cite[Theorem 26.5]{rockafellar1970convex}}}]\label{fact:legendre}
    Let $\varphi: \mathbb{V}\rightarrow\bar{\mathbb{R}}$ be a Legendre function.
    \begin{enumerate}[label=(\roman*), ref=\ref{fact:legendre}(\roman*)]
        \item The convex conjugate $\varphi^*$ of $\varphi$ is also Legendre. \label{fact:legendre-i}
        \item \sloppy The gradient $\nabla\varphi : \interior\domain\varphi \rightarrow \interior\domain\varphi^*$ is continuous and invertible, where $(\nabla\varphi)^{-1}=\nabla\varphi^*$. \label{fact:legendre-ii}
    \end{enumerate}
\end{fact}
\sloppy Given a Legendre function $\varphi: \mathbb{V}\rightarrow\bar{\mathbb{R}}$, the Bregman divergence $D_\varphi : \domain\varphi\times\interior\domain\varphi\rightarrow\mathbb{R}$ is~\cite{bregman1967relaxation}
\begin{equation}\label{eqn:breg-div}
    D_\varphi\divx{x}{y} \coloneqq \varphi(x) - \mleft(\varphi(y) + \inp{\nabla\varphi(y)}{x - y} \mright).  
\end{equation}
Now consider using a Legendre function $\varphi$ to construct a local approximation of $f$ at $y\in\interior\domain\varphi\cap\interior\domain f$ for some scaling factor $t>0$
\begin{equation}\label{eqn:mirror-pre}
    \min_{x\in\mathcal{C}} \quad f(y) + \inp{\nabla f(y)}{x - y} + \frac{1}{t} D_\varphi\divx{x}{y}.
\end{equation}
If $y$ is close to the optimal solution $x^*$ of~\eqref{eqn:constr-min}, and if $\varphi$ is, in a sense, similar to $f$, then we expect the solution to~\eqref{eqn:mirror-pre} to be a good approximation of~\eqref{eqn:constr-min}. More formally, we recall from~\cite{bauschke2017descent, lu2018relatively, teboulle2018simplified} the following definitions for relative smoothness and strong convexity to describe structural similarities between $f$ and $\varphi$.
\begin{defn}\label{def:rel}
    Consider a function $f:\mathbb{V}\rightarrow\bar{\mathbb{R}}$, a Legendre function $\varphi:\mathbb{V}\rightarrow\bar{\mathbb{R}}$, and convex set $\mathcal{X}\subseteq\mathbb{V}$. Then
    \begin{enumerate}[label=(\roman*), ref=\ref{def:rel}(\roman*)]
        \item $f$ is $L$-smooth relative to $\varphi$ on $\mathcal{X}$ if $L\varphi - f$ is convex on $\relinterior\mathcal{X}$ for some $L>0$, and
        \item $f$ is $\mu$-strongly convex relative to $\varphi$ on $\mathcal{X}$ if $f - \mu\varphi$ is convex on $\relinterior\mathcal{X}$ for some $\mu>0$. \label{def:rel-strong-convex}
    \end{enumerate}    
\end{defn}
In particular, if $f$ is $L$-smooth relative to $\varphi$ on $\mathcal{C}$, then the local approximation used in~\eqref{eqn:mirror-pre} for $t=1/L$ is a majorization of $f$, and therefore gives an upper bound on the optimal value $f(x^*)$. To solve~\eqref{eqn:constr-min}, mirror descent uses the approximation~\eqref{eqn:mirror-pre} to construct the sequence of iterates
\begin{equation}\label{eqn:mirror-descent}
    x^{k+1} = \argmin_{x \in \mathcal{C}} \mleft\{ \inp{\nabla f(x^k)}{x} + \frac{1}{t_k} D_\varphi\divx{x}{x^k} \mright\}.
\end{equation}
This is a practical algorithm when~\eqref{eqn:mirror-pre} is easier to solve compared to the original problem~\eqref{eqn:constr-min}. When $f$ satisfies certain smoothness and strong convexity properties relative to $\varphi$, mirror descent has the following convergence properties.
\begin{fact}[{{\cite[Theorem 4.1 and Proposition 4.1]{teboulle2018simplified}}}]\label{fact:conv}
    Let $\varphi$ be Legendre, and $\{ x^k \}$ be the sequence of iterates generated by mirror descent~\eqref{eqn:mirror-descent} to solve~\eqref{eqn:constr-min}. 
    \begin{enumerate}[label=(\roman*), ref=\ref{def:rel}(\roman*)]
    \item If $f$ is $L$-smooth relative to $\varphi$, then $f(x^k)$ satisfies
    \begin{equation*}
        f(x^k) - f^* \leq \frac{L}{k}D_\varphi\divx{x^*}{x^0}, \quad \forall k\in\mathbb{N},
        \label{eqn:sublin-conv}
    \end{equation*}
    i.e., $f(x^k)$ converges sublinearly, $O(1/k)$, to the optimal value. 
    \item If $f$ is $L$-smooth relative to $\varphi$ and $\mu$-strongly convex relative to $\varphi$, then $f(x^k)$ satisfies
    \begin{equation*}
        f(x^k) - f^* \leq \mleft(1 - \frac{\mu}{L}\mright)^k LD_\varphi\divx{x^*}{x^0}, \quad \forall k\in\mathbb{N},
        \label{eqn:lin-conv}
    \end{equation*}
    i.e., $f(x^k)$ converges linearly, $O(\alpha^k)$ for $\alpha\in(0, 1)$, to the optimal value.
    \end{enumerate}    
\end{fact}

We will make the following assumptions throughout the rest of the paper to ensure the mirror descent algorithm is well defined and to simplify some of the analysis.
\begin{assump}\label{assump:general}
    Consider the convex optimization problem~\eqref{eqn:constr-min} and mirror descent rule~\eqref{eqn:mirror-descent}.
    \begin{enumerate}[label=(\roman*), ref=\ref{assump:general}(\roman*)]
        \item $f:\mathbb{V}\rightarrow\bar{\mathbb{R}}$ is convex and continuously differentiable on $\interior\domain f$,\label{assump:general-i}
        \item $\varphi:\mathbb{V}\rightarrow\bar{\mathbb{R}}$ is Legendre,\label{assump:general-ii}
        \item $\domain f = \domain \varphi$, \label{assump:general-iii}
        \item $\mathcal{C}$ is compact and non-empty, and \label{assump:general-iv}
        \item $\domain\varphi^*=\mathbb{V}$. \label{assump:general-v}
    \end{enumerate}
\end{assump}
Under Assumptions~\cref{assump:general-i,assump:general-ii,assump:general-iii,assump:general-iv}, the mirror descent iterate~\eqref{eqn:mirror-descent} always exists and is uniquely defined~\cite[Lemma 2]{bauschke2017descent}. Assumption~\ref{assump:general-v} may seem restrictive, however it is satisfied by many important examples (see~\cite[Example 2.1]{teboulle2018simplified}), including the ones we later use to compute the classical and quantum rate-distortion functions. The motivation behind Assumption~\ref{assump:general-iii} is that this allows the Bregman divergence to act as a natural barrier for the domain of the objective function, allowing these implicit constraints to be more easily handled. In particular, Legendre functions satisfy the following property.
\begin{fact}[{{\cite[Lemma 2.2]{maddison2021dual}}}]\label{fact:interior}
    Consider an essentially smooth convex function $\varphi: \mathbb{V}\rightarrow\bar{\mathbb{R}}$. If $\varphi$ is minimized at $x\in\domain\varphi$, then $x\in\interior\domain\varphi$.
\end{fact}
Therefore, as the sum of an essentially smooth convex function and an affine function is still clearly essentially smooth and convex, Fact~\ref{fact:interior} tells us that the mirror descent iterates~\eqref{eqn:mirror-descent} always lie in $\interior\domain\varphi$.

\paragraph{Dual of mirror descent subproblem}
It will also be useful to consider the Lagrangian of the problem~\eqref{eqn:mirror-pre} for some $y\in\interior\domain\varphi$
\begin{equation}
    \mathcal{L}(x, \nu; y) \coloneqq \inp{\nabla f(y)}{x} + \frac{1}{t} D_\varphi\divx{x}{y} + \inp{\nu}{\mathcal{A}(x) - b},
\end{equation}
and the corresponding dual function
\begin{equation}\label{eqn:md-dual}
    g(\nu; y) \coloneqq \inf_x \mathcal{L}(x, \nu; y) = - \frac{1}{t} \varphi^*[ \nabla\varphi(y) - t(\nabla f(y) + \mathcal{A}^\dag(\nu)) ] - \inp{b}{\nu},
\end{equation}
which under Assumption~\ref{assump:general-v} has the property $\domain g(\wc; y)=\mathbb{V}'$. When $y=x^k$, these are the Lagrangian and dual function for the mirror descent subproblem~\eqref{eqn:mirror-descent} at the $k$-th mirror descent iteration. We will use $\mathcal{L}^k(x, \nu) \coloneqq \mathcal{L}(x, \nu; x^k)$ and $g^k(\nu)\coloneqq g(\nu; x^k)$ to simplify notation. An alternative characterization of the dual is $g(\nu; y) = \mathcal{L}(\Tilde{x}, \nu; y)$ where
\begin{equation}\label{eqn:md-p-var}
    \Tilde{x} = \argmin_x \mathcal{L}(x, \nu; y) = (\nabla\varphi)^{-1}[\nabla\varphi(y) - t (\nabla f(y) + \mathcal{A}^\dag(\nu))],
\end{equation}
which can be used to recover a primal optimal solution from a dual optimal solution assuming $(\nabla\varphi)^{-1}$ is easily computable. Using Fact~\ref{fact:interior}, it follows that $\Tilde{x}\in\interior\domain\varphi$. We can show that the dual problem 
\begin{equation}\label{eqn:md-dual-prob}
    \max_\nu \; g(\nu; y),
\end{equation}
is well-posed under certain conditions using the following proposition. We begin with a standard definition from convex analysis before introducing the main result.
\begin{defn}[Coercive function]\label{defn:coercive}
    A function $f: \mathbb{V}\rightarrow\bar{\mathbb{R}}$ is coercive if for every sequence $\{x^k\}$ such that $\norm{x^k}\rightarrow\infty$, we have $f(x^k) \rightarrow \infty$.
\end{defn}
\begin{prop}\label{prop:dual-prop}
    Consider the dual function $g(\wc; y) : \mathbb{V}'\rightarrow\mathbb{R}$, as defined in~\eqref{eqn:md-dual}, for some $y\in\interior\domain\varphi$ and where $\varphi$ is Legendre. Then
    \begin{enumerate}[label=(\roman*), ref=\ref{}(\roman*)]
        \item  $-g(\wc; y)$ is strictly convex if and only if $\mathcal{A}^\dag$ is injective, and
        \item  $-g(\wc; y)$ is coercive if and only if $\mathcal{A}^\dag$ is injective and $b\in\{ \mathcal{A}(x) : x\in\interior\domain\varphi \}$.
    \end{enumerate}
\end{prop}
\begin{proof}
    See Appendix~\ref{appdx:proof-dual}.
\end{proof}
Using this result, it is well known that existence of a solution to the dual problem follows from coercivity, and uniqueness of the solution follows from strict convexity~\cite[Proposition 3.1.1 and 3.2.1]{bertsekas2009convex}.

\section{Problem definition} \label{sec:rd}

Consider an $n$-dimensional input system $A$, an $m$-dimensional output system $B$, and a quantum state $\rho_A\in\mathcal{D}_n$ with purification $\rho_{AR}\in\mathcal{D}_{n^2}$. Let $\Delta \in \mathbb{H}^{mn}_+$ be a distortion matrix and $D\geq0$ be a maximum allowable distortion. The \emph{entanglement-assisted quantum rate-distortion function}~\cite{datta2012quantum,wilde2013quantum} is
\begin{subequations}\label{eqn:pre-qrd}
    \begin{align}
        R_q(D; \rho_A, \Delta) \quad \coloneqq \quad \minimize_{\mathcal{N} \in \Phi_{m,n}} \quad & I_q(\mathcal{N}; \rho_A) \\
        \subjto \quad & \inp{\Delta}{(\mathcal{N}\otimes\mathbb{I})(\rho_{AR})} \leq D, 
    \end{align}
\end{subequations}
where
\begin{equation}
    I_q(\mathcal{N}; \rho_A) \coloneqq S\divx{ (\mathcal{N}\otimes\mathbb{I})(\rho_{AR}) }{ \mathcal{N}(\rho_A) \otimes \rho_A }.
\end{equation}
is the \emph{quantum mutual information}. It follows from joint convexity of quantum relative entropy~\cite[Corollary 11.9.2]{wilde2017quantum} that~\eqref{eqn:pre-qrd} is a convex optimization problem.

A standard distortion measure used for the quantum rate-distortion function is to use the entanglement fidelity (see, e.g.,~\cite{barnum2000quantum,datta2012quantum,wilde2013quantum,khanian2021rate,khanian2023rate}). This corresponds to the channel where $A$ and $B$ are isomorphic, and the distortion matrix is $\Delta = \mathbb{I} - \rho_{AR}$. This distortion measure possesses symmetries which we will can take advantage of in Section~\ref{sec:symm}. We also briefly discuss the symmetries of other common distortion measures in Section~\ref{sec:other-dist}.


\begin{rem}\label{rem:degenerate}
    Without loss of generality, we will always assume that the input state of the quantum channel is full-rank, i.e., $\rho_A \succ 0$. This is because it is possible to construct equivalent optimization problems of a lower input dimension to avoid these degeneracies by discarding parts of the channel and distortion measures which correspond to degenerate inputs which have no effect on the problem. See Appendix~\ref{appdx:rd-equiv-degenerate} for more details.
\end{rem}


\subsection{Problem reformulation}

In this section, we summarize standard reformulations of the quantum rate-distortion problems that are more convenient for our purposes.

\paragraph{Change of variables}
We will make a change of variables where we optimize over the bipartite state $\sigma_{BR}\in \mathcal{D}_{mn}$, where $\sigma_{BR}=(\mathcal{N}\otimes\mathbb{I})(\rho_{AR})$. For distortion $D\geq0$, input state $\rho_A\in\mathcal{D}_n$ and distortion matrix $\Delta \in \mathbb{H}^{mn}_+$, this gives us
\begin{subequations}\label{eqn:new-qrd}
    \begin{align}
        R_q(D; \rho_A, \Delta) \quad = \quad \minimize_{\sigma_{BR} \in \mathcal{D}_{mn}} \quad & I_q(\sigma_{BR}; \rho_A) \\
        \subjto \hspace{0.5em} \quad & \tr_B(\sigma_{BR}) = \rho_A, \label{eqn:qrd-b}\\
                            \quad & \inp{\Delta}{\sigma_{BR}} \leq D, \label{eqn:qrd-c}
    \end{align}
\end{subequations}
where with some slight abuse of notation we redefine quantum mutual information as
\begin{equation}
    I_q(\sigma_{BR}; \rho_A) \coloneqq S\divx{ \sigma_{BR} }{ \tr_R(\sigma_{BR}) \otimes \rho_A }.
\end{equation}
Note that the apparent inconsistency in the quantum systems implied by the subscripts of~\eqref{eqn:qrd-b} is due to the fact that, as discussed in Section~\ref{sec:qit}, $A$ and $R$ are isomorphic systems, and that we have defined the purification in a way such that $\rho_R\coloneqq\tr_A(\rho_{AR}) = \rho_A$.

\begin{prop}\label{prop:rd-equiv}
    The quantum rate-distortion problems~\eqref{eqn:pre-qrd} and~\eqref{eqn:new-qrd} are equivalent, in the sense that the optimal values of the two problems are equal.
\end{prop}
\begin{proof}
    See Appendix~\ref{appdx:rd-equiv-state}.
\end{proof}

As we made a linear change of variables, these reformulated problem is still convex. We refer the reader to~\cite[Section 6.2]{he2023mirror} for the gradient of the reformulated quantum mutual information.

\paragraph{Distortion constraint} 
For rate-distortion algorithms, we are usually interested in characterizing the entire rate-distortion curve rather than computing the rate for a single distortion. Therefore it is standard to parameterize the distortion constraint with its corresponding dual variable $\kappa\geq0$ instead of $D$ (see e.g.,~\cite{blahut1972computation}). In this case, the quantum rate-distortion problem for distortion dual variable $\kappa\geq0$, input state $\rho_A\in\mathcal{D}_n$, and distortion matrix $\Delta \in \mathbb{H}^{mn}_+$ becomes
\begin{subequations}
    \makeatletter
    \def\@currentlabel{QRD}\label{eqn:qrd}
    \makeatother
    \renewcommand{\theequation}{QRD.\alph{equation}}
    \begin{align}
        \tilde{R}_q(\kappa; \rho_A, \Delta) \quad = \quad \minimize_{\sigma_{BR} \in \mathcal{D}_{mn}} \quad & I_q(\sigma_{BR}; \rho_A) + \kappa \inp{\Delta}{\sigma_{BR}} \label{eqn:qrd-obj} \\
        \subjto \hspace{0.5em} \quad & \tr_B(\sigma_{BR}) = \rho_A, \label{eqn:qrd-pt}
    \end{align}
\end{subequations}
Note that as the distortion constraints are linear, the convexity properties of the objective function does not change with this reformulation. Additionally, we have omitted some additive constant terms in the objective to simplify the expression. Once a solution $\sigma_{BR}^*$ to~\eqref{eqn:qrd} is found for a given $\kappa\geq0$, it is straightforward to recover the corresponding optimal rate $R_q(D; \rho_A, \Delta)=I_q(\sigma_{BR}^*; \rho_A)$ for distortion $D=\inp{\Delta}{\sigma_{BR}^*}$.

\section{Symmetry reduction} \label{sec:symm}

A practical issue with solving the quantum rate-distortion problem~\eqref{eqn:qrd} is that the problem dimension scales quartically with the quantum system dimension. We will show that we can restrict this problem to a significantly smaller subspace by taking advantage of symmetry properties of the rate-distortion problem for certain distortion matrices. We will also show that for certain highly symmetric problems, we can obtain explicit expressions for the rate-distortion function. 

\subsection{Background}

We begin by introducing some preliminaries on group and representation theory. We refer the reader to~\cite{brocker2013representations,fulton2013representation} as references for these basic facts. Throughout this section we will denote the general linear group as $GL(\mathbb{V})$, which is the group of invertible linear transforms from $\mathbb{V}$ to $\mathbb{V}$ with the group operation being composition.
\begin{defn}[Representation]
    A \emph{representation} of a group $\mathcal{G}$ is a pair $(\mathbb{V}, \pi)$ where $\mathbb{V}$ is a vector space and $\pi:\mathcal{G}\rightarrow GL(\mathbb{V})$ is a group homomorphism.
\end{defn}
\begin{defn}[Isomorphism]
    Let $(\mathbb{V}, \pi)$ and $(\mathbb{V}', \tau)$ be representations of a group $\mathcal{G}$. If there exists an invertible linear map $T:\mathbb{V}\rightarrow\mathbb{V}'$ such that
    \begin{equation*}
        \tau(g) = T \circ \pi(g) \circ T^{-1},
    \end{equation*}
    for all $g\in\mathcal{G}$, then $\pi$ and $\tau$ are isomorphic, denoted $\pi \cong \tau$.
\end{defn}
\begin{defn}[Fixed-point subspace]
    The \emph{fixed point subspace} $\mathcal{V}_\pi$ associated with a representation $(\mathbb{V}, \pi)$ of a group $\mathcal{G}$ is
    \begin{equation*}
       \mathcal{V}_{\pi|\mathbb{V}} = \{x \in \mathbb{V} : \pi(g)(x)=x,~\forall g\in\mathcal{G}\}.
    \end{equation*}
    When the vector space $\mathbb{V}$ is clear from context, we use the abbreviated notation $\mathcal{V}_\pi \coloneqq \mathcal{V}_{\pi|\mathbb{V}}$.
\end{defn}
\begin{defn}[Group average]
    For a representation $(\mathbb{V}, \pi)$ of a compact group $\mathcal{G}$, the \emph{group average} is
    \begin{equation*}
        P_\pi = \int_{g\in \mathcal{G}} \pi(g)\,d\theta(g),
    \end{equation*}
    where $\theta$ is the normalized Haar measure on $\mathcal{G}$ (see e.g.,~\cite{bredon1972introduction}). The group average is the projector onto the fixed point subspace $\mathcal{V}_\pi$.
\end{defn}
\begin{lem}\label{lem:fix-pnt}
    Consider a representation $(\mathbb{V}, \pi)$ of a group $\mathcal{G}$. If a convex optimization problem of the form~\eqref{eqn:constr-min} is invariant under $\pi$, meaning
    \begin{subequations}
        \begin{align}
            &f(\pi(g)(x)) = f(x) \quad && \forall g\in\mathcal{G},~\forall x\in\mathcal{C} \label{eqn:invar-obj} \\
            \textrm{and} \quad &\pi(g)(x)\in\mathcal{C} && \forall g\in\mathcal{G},~\forall x\in\mathcal{C}, \label{eqn:invar-constr}
        \end{align}
    \end{subequations}  
    then there is an optimal point for~\eqref{eqn:constr-min} in $\mathcal{V}_\pi$.
\end{lem}
\begin{proof}
    Let $x^*$ be a solution to the convex optimization problem and $f^*$ be the optimal value. We can show that  
    \begin{align*}
        f^* \leq f \biggl( \int_{g\in \mathcal{G}} \pi(g)(x^*)\,d\theta(g) \biggr) \leq \int_{g\in \mathcal{G}} f(\pi(g)(x^*)))\,d\theta(g) = f^*.
    \end{align*}
    The first inequality holds as $P_\pi(x^*)\in\mathcal{C}$ from~\eqref{eqn:invar-constr}, and that by definition $f^*\leq f(x)$ for all $x\in\mathcal{C}$. The second inequality uses convexity of $f$, and the equality uses~\eqref{eqn:invar-obj}. Therefore, equality must hold throughout. This implies the projection of the solution onto the fixed-point subspace, $P_\pi(x^*)$, must also be a solution, as desired.
\end{proof}

We will also introduce characters of compact groups and well known results associated with them. For these facts, we will assume that $\mathbb{V}$ is a complex vector space, as some of these results do not hold in general for real vector spaces.

\begin{defn}[Character]
    Let $(\mathbb{V}, \pi)$ be a representation of a compact group $\mathcal{G}$. The \emph{character} $\chi_\pi:\mathcal{G}\rightarrow \mathbb{C}$ of this representation is defined as
    \begin{equation*}
        \chi_\pi(g) \coloneqq \tr[ \pi(g) ]. 
    \end{equation*}
    Given characters $\chi,\phi$ and normalised Haar measure $\theta$ of $\mathcal{G}$, we define the inner product
    \begin{equation*}
        \inp{\chi}{\phi} \coloneqq \int_{g\in\mathcal{G}} \chi(g) \overline{\phi(g)}\,d\theta(g). \label{eqn:char-inp}
    \end{equation*}
\end{defn}
\begin{lem}[{{\cite[Proposition 4.10(v)]{brocker2013representations}}}]\label{lem:char-add}
    Let $(\mathbb{V}, \pi)$ and $(\mathbb{V}', \tau)$ be representations of a compact group $\mathcal{G}$. Then $\chi_{\pi\oplus\tau}=\chi_\pi + \chi_\tau$, where $\oplus$ denotes the direct sum between representations.
\end{lem}
\begin{lem}[{{\cite[Theorem 4.11(iii)]{brocker2013representations}}}]\label{lem:char-ortho}
    Characters of irreducible representations of a compact group $\mathcal{G}$ are orthonormal. That is for any irreducible representations $(\mathbb{V}, \pi)$ and $(\mathbb{V}', \tau)$ of $\mathcal{G}$
    \begin{equation*}
        \inp{\chi_\pi}{\chi_\tau} = \begin{cases}
                                1, \quad &\textrm{if } \pi\cong\tau,\\
                                0, &\textrm{if } \pi\ncong\tau.\\
                            \end{cases}
    \end{equation*}
\end{lem}

\begin{lem}\label{lem:char-dim}
    Given a complex representation $(\mathbb{C}^n, \tau)$ of a compact group $\mathcal{G}$, define a representation $(\mathbb{C}^{n\times n}, \pi)$ of $\mathcal{G}$ by
    \begin{equation}\label{eqn:rep-conj}
        \pi(g)(X) = \tau(g) X \tau(g)^\dag. 
    \end{equation}
    The dimension (over $\mathbb{C}$) of the fixed-point subspace $\mathcal{V}_\pi$ is
    \begin{align}\label{eqn:fps-dim}
        \dim_\mathbb{C} \mathcal{V}_\pi = \int_{g\in \mathcal{G}} \tr[\tau(g)]\,\overline{\tr[\tau(g)]} \,d\theta(g) = \inp{\chi_\tau}{\chi_\tau}.
    \end{align}
\end{lem}
\begin{proof}
    This uses~\cite[Theorem 4.11(i)]{brocker2013representations} combined with the properties $\vecc(AXB^\dag)=(A \otimes \overline{B})\vecc(X)$ and $\tr[A\otimes B]=\tr[A]\tr[B]$.
\end{proof}

We next consider the analogue of Lemma~\ref{lem:char-dim} for real representations.
\begin{cor}\label{cor:reals-dim}
    Let $\mathcal{G}$ be a compact group, let $(\mathbb{C}^n,\tau)$ be a complex representation of $\mathcal{G}$, and let $(\mathbb{H}^n,\pi)$ be the real representation of $\mathcal{G}$ defined by $\pi(g)(X) = \tau(g)X\tau(g)^\dag$. Then the dimension (over $\mathbb{R}$) of the fixed point subspace is given by $\dim_{\mathbb{R}}(\mathcal{V}_{\pi}) = \langle \chi_\tau,\chi_\tau\rangle$.
\end{cor}
\begin{proof}
    We first recognize that by treating $\mathbb{C}^{n\times n}$ as a real vector space, $(\mathbb{C}^{n\times n}, \pi)$ decomposes into two isomorphic real subrepresentations given by the restriction to Hermitian and skew-Hermitian matrices, respectively. The decompositions of these isomorphic representations into irreducibles are, therefore, identical. It follows that the fixed point subspaces in each have the same real dimension. Therefore given that the fixed point subspace $\mathcal{V}_{\pi|\mathbb{C}^{n\times n}}$ has a dimension over the reals of $\dim_{\mathbb{R}}(\mathcal{V}_{\pi|\mathbb{C}^{n\times n}}) = 2 \dim_{\mathbb{C}}(\mathcal{V}_{\pi|\mathbb{C}^{n\times n}})$, restricting to the Hermitian space and using Lemma~\ref{lem:char-add} leaves us with the desired real dimension for $\mathcal{V}_{\pi|\mathbb{H}^{n}}$.
\end{proof}

\subsection{Entanglement fidelity distortion}\label{subsec:general-dist}

We now restrict our attention to the quantum rate-distortion problem with the entanglement fidelity distortion matrix $\Delta=\mathbb{I} - \rho_{AR}$, where $\rho_{AR}$ is the purification of the input state $\rho_A$. We will show that for this choice of distortion matrix, there are significant symmetries in the quantum rate-distortion problem~\eqref{eqn:qrd} which can be used to drastically reduce the problem dimension.

We will first introduce the following lemma that will be useful to prove the main result.

\begin{lem}\label{lem:p-trace-unitary}
    Consider an $m$-dimensional system $A$, an $n$-dimensional $B$, and a density matrix $\rho_{AB} \in \mathcal{D}_{mn}$ defined on the joint system $AB$. Let $T_A\in\mathbb{C}^{p\times m}$ be an arbitrary linear operator on $A$, and $U_B\in\mathbb{C}^{n\times n}$ be an arbitrary unitary operator on $B$. Then
    \begin{equation*}
        \tr_B((T_A \otimes U_B)\rho_{AB}(T_A \otimes U_B)^\dag) = T_A\tr_B(\rho_{AB})T_A^\dag.
    \end{equation*}
\end{lem}
\begin{proof}
    Recall that the partial trace is defined as the adjoint of the Kronecker product with the identity matrix. Therefore, for all $\sigma_A \in \mathcal{D}_n$ the following equalities must hold
    \begin{align*}
        \inp{\sigma_A}{\tr_B((T_A \otimes U_B)\rho_{AB}(T_A \otimes U_B)^\dag))} &= \inp{\sigma_A\otimes\mathbb{I}_B}{(T_A \otimes U_B)\rho_{AB}(T_A \otimes U_B)^\dag)}\\
        &= \inp{(T_A^\dag\sigma_AT_A)\otimes\mathbb{I}_B}{\rho_{AB}}\\
        &= \inp{\sigma_A}{T_A \tr_B(\rho_{AB}) T_A^\dag},
    \end{align*}
    where the second equality uses $U_B^\dag U_B=\mathbb{I}_B$. This produces the desired result.
\end{proof}

We now find a group action that leaves the quantum rate-distortion problem invariant, which we can use to characterize the corresponding fixed-point subspace in which a solution must lie. For simplicity of notation we assume the input state $\rho_A$ is diagonal for the following lemma.

\begin{lem}\label{lem:qrd-general}
    Consider a diagonal input state $\rho_A\in\mathcal{D}_n$. Let $\mathcal{G}$ be a subgroup of the group of $n\times n$ unitary matrices, such that $g \rho_A g^\dag = \rho_A$ for all $g\in\mathcal{G}$, and $(\mathbb{H}^{n^2}, \pi)$ be a representation of $\mathcal{G}$ such that
    \begin{align}\label{eqn:qrd-representation}
        \pi(g)(X) = U X U^\dag, \quad \forall X\in\mathbb{H}^{n^2}, \quad  \textrm{where}\quad U = g \otimes \overline{g}.
    \end{align}    
    Then the quantum rate-distortion problem~\eqref{eqn:qrd} parameterized by input state $\rho_A$ and distortion matrix $\Delta=\mathbb{I} - \rho_{AR}$, where $\rho_{AR}$ is the purification of $\rho_A$, is invariant under $\pi$.
\end{lem}
\begin{proof}
    We will show that the objective function and each constraint of problem~\eqref{eqn:qrd} is invariant under $\pi$. First, we recognize that since $\rho_A$ is Hermitian and diagonal, $\rho_A=\overline{\rho_A}$, and therefore $g \rho_A g^\dag = \rho_A$ implies $\overline{g} \rho_A \overline{g}^\dag = \rho_A$. Invariance of the objective function~\eqref{eqn:qrd-obj} can be shown by first showing quantum mutual information is invariant under $\pi$ as follows
    \begin{align*}
        I_q(\pi(g)(\sigma_{BR}); \rho_A) &= S\divx{ \pi(g)(\sigma_{BR}) }{ \tr_R(\pi(g)(\sigma_{BR})) \otimes \rho_A }\\
        &= S\divx{ U \sigma_{BR} U^\dag }{ g\tr_R(\sigma_{BR})g^\dag \otimes \overline{g}\rho_A\overline{g}^\dag }\\
        &= S\divx{ U \sigma_{BR} U^\dag }{ U (\tr_R(\sigma_{BR}) \otimes \rho_A) U^\dag }\\
        &= I_q(\sigma_{BR}; \rho_A),
    \end{align*}
    where the second equality uses Lemma~\ref{lem:p-trace-unitary} and $\overline{g} \rho_A \overline{g}^\dag = \rho_A$, and the third equality uses invariance of quantum relative entropy under unitary transforms. For the linear component of the objective arising from the distortion constraint, let us first assume $\rho_A$ can be expressed as
    \begin{equation*}
        \rho_A = \sum_{I=1}^r \sum_{i\in\mathcal{I}_I} \lambda_I e_ie_i^\top,
    \end{equation*}
    i.e., $\rho_A$ has $r$ distinct eigenvalues, and $\mathcal{I}_I$ is the set of indices corresponding to eigenvalue $\lambda_I$ for all $I=1,\ldots,r$. For any $I=1,\ldots,r$ and $i\in\mathcal{I}_I$, we have
    \begin{equation*}
        \rho_A e_i = g \rho_A g^\dag e_i = \lambda_i e_i \quad \Longleftrightarrow \quad \rho_A\,g^\dag e_i = \lambda_i\,g^\dag e_i.
    \end{equation*}
    This implies that $g^\dag e_i$ must be some vector in the subspace spanned by $\{ e_i \}_{i\in\mathcal{I}_I}$, and because $g^\dag$ is unitary, that $\{ g^\dag e_i \}_{i\in\mathcal{I}_I}$ forms an orthonormal basis of this subspace. Therefore
    \begin{equation*}
        \sum_{i\in\mathcal{I}_I} g^\dag e_ie_i^\top g = \sum_{i\in\mathcal{I}_I} e_ie_i^\top,
    \end{equation*}
    i.e., $g^\dag P_I g = P_I$ where $P_I$ is the orthogonal projector onto the eigenspace corresponding to eigenvalue $\lambda_I$. Using this, we can show
    \begin{align*}
        \pi(g^\dag)(\rho_{AR}) &= \sum_{ij}^n \sqrt{\lambda_i\lambda_j} (g^\dag e_ie_i^\top g) \otimes \overline{(g^\dag e_je_j^\top g)} \\
        &= \sum_{I=1}^r\sum_{J=1}^r \sqrt{\lambda_I\lambda_J} \biggl(\sum_{i\in\mathcal{I}_I} g^\dag e_ie_i^\top g \biggr) \otimes \overline{\biggl(\sum_{j\in\mathcal{I}_J} g^\dag e_je_j^\top g \biggr)}\\
        &= \sum_{I=1}^r\sum_{J=1}^r \sqrt{\lambda_I\lambda_J} \biggl(\sum_{i\in\mathcal{I}_I} e_ie_i^\top \biggr) \otimes \biggl(\sum_{j\in\mathcal{I}_J} e_je_j^\top \biggr)\\
        &= \rho_{AR}
    \end{align*}
    Combining this with the fact that $\inp{\Delta}{\pi(g)(\sigma_{BR})} = \inp{\pi(g^\dag)(\Delta)}{\sigma_{BR}}$ and that $\pi(g^\dag)(\mathbb{I})=\mathbb{I}$ shows that the linear component of the objective is invariant under $\pi$. 

    Next, we will show invariance of the feasible set. First, note that the set of density matrices is invariant under congruence by unitary transforms. Second, consider an arbitrary $\sigma_{BR}$ which satisfies~\eqref{eqn:qrd-pt}. Then we can show that
    \begin{equation*}
        \tr_B(\pi(g)(\sigma_{BR})) = \overline{g} \rho_A \overline{g}^\dag = \rho_A,
    \end{equation*}
    where the first equality uses Lemma~\ref{lem:p-trace-unitary}. This shows that the constraint set is invariant under $\pi$, and concludes the proof.
\end{proof}

The general idea motivating the following theorem is that we want to find an appropriate subgroup $\mathcal{G}$ that satisfies the condition $g \rho_A g^\dag = \rho_A$, which is sufficiently large so that we can restrict to as small of a fixed-point subspace as possible. In the following theorem, we characterize such a subgroup and its corresponding fixed-point subspace for an arbitrary purification.

\begin{thm}\label{cor:sym-space}
    Consider the quantum rate-distortion problem~\eqref{eqn:qrd} parameterized by input state $\rho_A\in\mathcal{D}_n$ with spectral decomposition $\rho_A=\sum_{i=1}^n\lambda_iv_iv_i^\dag$ and distortion matrix $\Delta=\mathbb{I} - \rho_{AR}$, where $\rho_{AR}$ is the purification of $\rho_A$. A solution to this problem exists in the subspace
    \begin{align}
       \mathcal{V}_\pi \coloneqq \biggl\{ \sum_{i \neq j}^n \alpha_{ij} v_iv_i^\dag \otimes v_jv_j^\dag + \sum_{ij}^n \beta_{ij} v_iv_j^\dag \otimes v_iv_j^\dag : \alpha_{ij} \in \mathbb{R}~\forall i\neq j,~\beta \in \mathbb{H}^n \biggr\}.\label{eqn:qrd-fps}%
    \end{align}
\end{thm}
\begin{proof}
    We first consider the special case when $\rho_A$ is diagonal, i.e., $v_i=e_i$ for all $i=1,\ldots,n$. For any of these diagonal inputs, we can show that the unitary subgroup $\mathcal{G}$ consisting of elements
    \begin{equation}\label{eqn:proof-sym-space-a}
        \mathcal{G} = \biggl\{ \sum_{i=1}^n z_i e_ie_i^\top : z_i\in\{\pm1, \pm\sqrt{-1}\}~\forall i  \biggr\},
    \end{equation}
    satisfies $g \rho_A g^\dag=\rho_A$ for all $g\in\mathcal{G}$, and therefore Lemma~\ref{lem:qrd-general} tells us that the quantum rate-distortion problem is invariant under representation~\eqref{eqn:qrd-representation} of this subgroup. It follows from Lemma~\ref{lem:fix-pnt} that a solution must live in the fixed-point subspace of this representation. To characterise this subspace, let us first consider representation $(\mathbb{C}^{n^2}, \tau)$ of the group $\mathcal{G}$ such that $\tau(g)(x) = (g \otimes \overline{g})x$. We can show that $\tau$ can be decomposed into the following subrepresentations $(V_{ij}, \tau_{ij})$ such that
    \begin{align*}
        &V_{ij} = \{ x (e_i \otimes e_j) : x\in\mathbb{C} \} \\
        \textrm{and} \quad &\tau_{ij}(g)(x) = z_i \overline{z_j} (e_ie_i^\top \otimes e_je_j^\top), \quad \forall i,j=1,\ldots,n.
    \end{align*}
    where we follow the notation of~\eqref{eqn:proof-sym-space-a} to define $z_i$. We recognize that the subrepresentations $\tau_{ii}$ for all $i=1,\ldots,n$ are just the identity map on a one-dimensional subspace, and are therefore all isomorphic to each other. If $i\neq j$, then the $\tau_{ij}$ are mutually non-isomorphic irreducible representations. Therefore, using Corollary~\ref{cor:reals-dim} we can compute the dimension over the reals of the fixed point subspace $\mathcal{V}_\pi$ as
    \begin{equation*}
        \dim_{\mathbb{R}}\mathcal{V}_\pi = \inp{\chi_\tau}{\chi_\tau} = \sum_{ijkl}^n \inp{\chi_{ij}}{\chi_{kl}} = \sum_{ij}^n \inp{\chi_{ii}}{\chi_{jj}} + \sum_{i\neq j}^n \inp{\chi_{ij}}{\chi_{ij}} = 2n^2 - n.
    \end{equation*}
    where $\chi_{ij}$ is the character of subrepresentation $\tau_{ij}$, and the first and second equalities use Lemmas~\ref{lem:char-add} and~\ref{lem:char-ortho} respectively. Verifying that the proposed subspace~\eqref{eqn:qrd-fps} is invariant under representation~\eqref{eqn:qrd-representation} of group $\mathcal{G}$ and is of the correct dimension gives the desired result for any diagonal input $\rho_A$. This result generalizes to non-diagonal inputs by using the fact that we can use a change-of-basis to show that the quantum rate-distortion problem for any non-diagonal input state is equivalent to a quantum rate-distortion problem with a diagonal input state. See Appendix~\ref{appdx:rd-equiv-basis} for more details.
\end{proof}

We also introduce the following property of this subspace which can be shown from a straightforward calculation, and will be useful in some following proofs.

\begin{cor}\label{cor:subspace-ptrace}
    Let $\sigma_{BR} \in \mathcal{V}_\pi$ as defined in~\eqref{eqn:qrd-fps}. The partial traces of $\sigma_{BR}$ with respect to $R$ and $B$ satisfy
    \begin{equation*}
        \tr_R(\sigma_{BR}) = \sum_{i=1}^n x_i v_iv_i^\dag \quad \textrm{and} \quad \tr_B(\sigma_{BR}) = \sum_{i=1}^n y_i v_iv_i^\dag,
    \end{equation*}
    for some $x,y\in\mathbb{R}^n$ (i.e., they are simultaneously diagonalizable with the input state corresponding to the fixed point subspace).
\end{cor}

The implications of this symmetry reduction are as follows. Originally, the primal variable $\sigma_{BR}$ was defined on the vector space $\mathbb{H}^{n^2}$, which has a dimension over the reals of $n^4$. By restricting to the fixed-point subspace $\mathcal{V}_\pi$ as defined in~\eqref{eqn:qrd-fps}, this dimension over the reals reduces to $2n^2-n$. Moreover, following from Corollary~\ref{cor:subspace-ptrace}, the partial trace constraint~\eqref{eqn:qrd-pt} is reduced from representing $n^2$ equations to just $n$ equations. This will be particularly useful for the dualized algorithm we propose in Section~\ref{sec:alg}.

We also recognize that the subspace~\eqref{eqn:qrd-fps} is isomorphic to a block diagonal matrix with $n^2-n$ blocks of size $1\times 1$ and one block of size $n\times n$. Without loss of generality, we can change the basis of the data of the problem instance to be in the standard basis, i.e., $\rho_A$ is diagonal, meaning that in practice we can implement the subspace using sparse matrix representations. The block diagonal structure also allows us to perform eigendecompositions, and therefore matrix logarithms and exponentials, very efficiently.

\subsection{Uniform input state} \label{sec:uniform}

Let us continue to consider the distortion matrix $\Delta=\mathbb{I} - \rho_{AR}$. An interesting special case is when the input state is the maximally mixed state $\rho_A=\mathbb{I}/n$. It turns out that for this input state, there exist additional symmetries which further simplify the problem, and which ultimately allow us to obtain explicit expressions for the rate-distortion function.

\begin{cor}\label{cor:qrd-maxmix-fps}
    A solution to~\eqref{eqn:qrd} parameterized by the maximally mixed input state $\rho_A=\mathbb{I}/n$ and distortion matrix $\Delta=\mathbb{I} - \rho_{AR}$, where $\rho_{AR}$ is the purification of $\rho_A$, exists in the subspace
    \begin{equation}\label{eqn:maxmix-fps}
        \mathcal{V}_\pi = \{ a\mathbb{I} + b \rho_{AR} : a,b\in\mathbb{R} \}.
    \end{equation}
\end{cor}
\begin{proof}
    As $\rho_A$ is a multiple of the identity matrix, $g\rho_A g^\dag=\rho_A$ is satisfied for all $g\in U(n)$, where $U(n)$ is the entire group of $n\times n$ unitary matrices. Therefore Lemma~\ref{lem:qrd-general} tells us that the quantum rate-distortion problem is invariant under the representation~\eqref{eqn:qrd-representation} of $U(n)$. It follows from Lemma~\ref{lem:fix-pnt} that a solution must live in the fixed-point subspace of this representation. Using Corollary~\ref{cor:reals-dim}, the dimension over the reals of the fixed-point subspace $\mathcal{V}_\pi$ is given by
    \begin{align*}
        \dim_{\mathbb{R}} \mathcal{V}_\pi = \int_{U(n)} \tr[g \otimes \overline{g}]\,\overline{\tr[g \otimes \overline{g}]}  = \int_{U(n)} \tr[g]^2 \, \overline{\tr[g]^2} = 2,
    \end{align*}
    where the last equality comes from~\cite[Theorem 2.1(a)]{diaconis2001linear}. It is straightforward to verify that the proposed two-dimensional subspace~\eqref{eqn:maxmix-fps} is invariant under $\pi$, which concludes the proof.
\end{proof}

Given this result, we are able to obtain an explicit expression for the quantum rate-distortion function. In particular, we recover and generalize the result from~\cite[Theorem 10]{wilde2013quantum} to maximally mixed inputs of any dimension.
\begin{thm}\label{thm:qrd-maxmix-solution}
    The quantum rate-distortion function for the maximally mixed input state $\rho_A=\mathbb{I}/n$ and distortion matrix $\Delta=\mathbb{I} - \rho_{AR}$, where $\rho_{AR}$ is the purification of $\rho_A$, is
    \begin{equation}
        R_q(D; \rho_A, \Delta) = \begin{cases} 
            \log(n^2) - H \biggl( \biggl[1 - D, \underbrace{\frac{D}{n^2 - 1}, \ldots, \frac{D}{n^2 - 1}}_{(n^2-1) \textrm{ copies}} \biggr] \biggr), \quad &\textnormal{if } 0\leq D \leq \displaystyle 1-\frac{1}{n^2}\\
            0, &\textnormal{if } \displaystyle 1-\frac{1}{n^2} \leq D \leq 1.
        \end{cases}
    \end{equation}
    A channel that attains this rate for $0\leq D \leq 1-1/n^2$ is the depolarizing channel
    \begin{equation}
        \mathcal{N}(\rho) = \frac{D}{n^2-1}\mathbb{I} + \mleft(1 - \frac{n^2D}{n^2-1} \mright)\rho.
    \end{equation}

\end{thm}
\begin{proof}
    We will first find the solution to the modified quantum rate-distortion problem~\eqref{eqn:qrd} for the specified problem parameters, then use this solution to recover solutions to the original problem~\eqref{eqn:pre-qrd}. From Corollary~\ref{cor:qrd-maxmix-fps}, we can restrict our attention to consider quantum states of the form $\sigma_{BR}=a\mathbb{I} + b\rho_{AR}$ where $a,b\in\mathbb{R}$. Moreover, we have the relationship $b = 1 - n^2a$ as $\sigma_{BR}$ must have unit trace, and so the rate-distortion problem reduces to a one-dimensional problem. Therefore, the solution to~\eqref{eqn:qrd} can be found by finding when the gradient of a univariate objective is zero. Doing this we find that the optimal point of~\eqref{eqn:qrd} for a given $\kappa\geq0$ is
    \begin{align*}
        \sigma_{BR}^* &= \frac{1}{e^\kappa + n^2 - 1} (\mathbb{I} + (e^\kappa - 1) \rho_{AR}),
    \end{align*}
    from which we can easily compute the corresponding rate and distortion as
    \begin{align*}
        R_q &= 2\log(n) + \frac{\kappa e^\kappa}{e^\kappa + n^2 - 1} - \log(e^\kappa + n^2 - 1)\\
        D &= \frac{n^2 - 1}{e^\kappa + n^2 - 1}.
    \end{align*}
    It is straightforward to verify that $\sigma_{BR}^*(\kappa)$ is a valid solution as it is a feasible point of~\eqref{eqn:qrd} for all $\kappa\geq0$. Reparameterizing $R_q$ to be in terms of $D$ instead of $\kappa$ gives the desired rate-distortion function. Similarly, expressing $\sigma_{BR}^*$ in terms of $D$ instead of $\kappa$, we find
    \begin{equation}
        \sigma_{BR}^* = \frac{D}{n^2-1}\mathbb{I} + \mleft(1 - \frac{n^2D}{n^2-1} \mright)\rho_{AR}.
    \end{equation}
    As seen in Appendix~\ref{appdx:rd-equiv}, we can recover the Choi matrix from this solution state as $C_\mathcal{N} = n\sigma_{BR}^*$. Noting that $C_\mathcal{N}=\mathbb{I}/n$ corresponds to the quantum channel that maps all states to the maximally mixed state, and $C_\mathcal{N}=n\rho_{AR}$ corresponds to the identity channel, identifies the desired quantum channel. 
\end{proof}

\subsection{Other quantum distortion matrices} \label{sec:other-dist}

As noted in~\cite{wilde2013quantum}, there are several other common choices for the quantum distortion matrix $\Delta$ which have been studied in the literature. As with the entanglement fidelity distortion, we can use symmetry reduction to gain some additional insight for these other common distortion matrices. Throughout this section, we will assume that $\rho_A = \sum_{i=1}^n \lambda_i v_iv_i^\dag$, and denote $\{ b_i \}$ as an orthonormal basis for system $B$.

We will omit proofs for these results, as they can be shown using an almost identical method as used in Section~\ref{subsec:general-dist}. However, we note that for the quantum rate-distortion problem with input state $\rho_A\in\mathcal{D}_n$ and distortion matrix $\Delta\in\mathbb{H}^{mn}_+$ to be invariant under a representation $(\mathbb{H}^{mn}, \pi)$ of a group $\mathcal{G}$, it suffices that
\begin{align}\label{eqn:general-group}
    &\pi(g)(X) = g X g^\dag, \quad \forall X\in\mathbb{H}^{mn}, \\ 
    \quad  \textrm{where}\quad &g \in \mathcal{G} \leq \{ U_B \otimes U_R : U_B\in U(m), U_R\in U(n) \}, \nonumber
\end{align}
such that $\tr_B(g)\rho_A=\rho_A\tr_B(g)$ and $g\Delta g^\dag = \Delta$ for all $g\in\mathcal{G}$, and where we use $\mathcal{G}\leq\mathcal{G}'$ to denote that $\mathcal{G}$ is a subgroup of $\mathcal{G}'$. Note that we have one additional requirement that must be satisfied as compared to Lemma~\ref{lem:qrd-general}, as the distortion matrix $\Delta$ is no longer necessarily related to the input states.

\begin{exmp} \label{exmp:cc}
    Consider the classical-classical distortion matrix
    \begin{equation}
        \Delta_{cc} = \sum_{i=1}^m \sum_{j=1}^n \delta_{ij} b_ib_i^\dag \otimes v_jv_j^\dag,
    \end{equation}
    where $\delta\in\mathbb{R}^{m\times n}_+$. Consider the representation $(\mathbb{H}^{mn}, \pi_{cc})$ with group homomorphism~\eqref{eqn:general-group} of the group
    \begin{align}
        \mathcal{G}_{cc} = \biggl\{ \biggl( \sum_{i=1}^m w_i b_ib_i^\dag \biggr) \otimes \biggl( \sum_{j=1}^n z_j v_jv_j^\dag \biggr) : w_i, z_j \in \{ \pm1, \pm\sqrt{-1} \}, \forall i,j \biggr\}.
    \end{align}
    It can be shown that~\eqref{eqn:qrd} with the classical-classical distortion matrix is invariant under this representation. Therefore, a solution must lie in the corresponding fixed-point subspace
    \begin{equation}\label{eqn:cc-subspace}
        \mathcal{V}_{cc} = \biggl\{ \sum_{i=1}^m \sum_{j=1}^n x_{ij} b_ib_i^\dag \otimes v_jv_j^\dag : x\in\mathbb{R}^{m\times n} \biggr\}.
    \end{equation}
    As expected (and commented on in~\cite{wilde2013quantum}), the problem completely reduces down to diagonal matrices, and is therefore isomorphic to the classical rate-distortion problem (see Appendix~\ref{appdx:crd}).
\end{exmp}

\begin{exmp} \label{exmp:qc}
    Consider the quantum-classical distortion matrix studied in~\cite{datta2013quantum}
    \begin{equation}
        \Delta_{qc} = \sum_{i=1}^m b_ib_i^\dag \otimes \Delta_{i},
    \end{equation}
    where $\Delta_{i} \in \mathbb{H}^n_+$ for all $i=1,\ldots,m$. Consider the representation $(\mathbb{H}^{mn}, \pi_{qc})$ with group homomorphism~\eqref{eqn:general-group} of the group
    \begin{align}
        \mathcal{G}_{qc} = \biggl\{ \biggl( \sum_{i=1}^m w_i b_ib_i^\dag \biggr) \otimes \mathbb{I}_R : w_i \in \{ \pm1, \pm\sqrt{-1} \}, \forall i \biggr\}.
    \end{align}
    It can be shown that~\eqref{eqn:qrd} with the quantum-classical distortion matrix is invariant under this representation. Therefore, a solution must lie in the corresponding fixed-point subspace
    \begin{equation}
        \mathcal{V}_{qc} = \biggl\{ \sum_{i=1}^m b_ib_i^\dag \otimes \rho_i : \rho_i\in\mathbb{H}^{n}, \forall i=1,\ldots,m \biggr\},
    \end{equation}
    which can be interpreted as a block-diagonal structure by changing the basis of $\{b_i\}$ to be in the standard basis.
\end{exmp}

\begin{exmp} \label{exmp:cq}
    Consider the classical-quantum distortion matrix for a system similar to that studied in~\cite{hayashi2023efficient}
    \begin{equation}\label{eqn:cq-distortion}
        \Delta_{cq} = \sum_{j=1}^n \Delta_{j} \otimes v_jv_j^\dag,
    \end{equation}
    where $\Delta_{j} \in \mathbb{H}_+^m$ for all $j=1,\ldots,n$. Consider the representation $(\mathbb{H}^{mn}, \pi_{cq})$ with group homomorphism~\eqref{eqn:general-group} of the group
    \begin{align}
        \mathcal{G}_{cq} = \biggl\{ \mathbb{I}_B \otimes \biggl( \sum_{j=1}^n z_j v_jv_j^\dag \biggr) : z_j \in \{ \pm1, \pm\sqrt{-1} \}, \forall j \biggr\}.
    \end{align}    
    It can be shown that~\eqref{eqn:qrd} with the classical-quantum distortion matrix is invariant under this representation. Moreover, the corresponding fixed-point subspace is
    \begin{equation}
        \mathcal{V}_{cq} = \biggl\{ \sum_{j=1}^n \rho_j \otimes v_jv_j^\dag : \rho_j\in\mathbb{H}^{m}, \forall j=1,\ldots,n \biggr\},
    \end{equation}
    which can be interpreted as a block-diagonal structure by changing the basis of $\{v_i\}$ to be in the standard basis.
\end{exmp}


\section{Efficient algorithm} \label{sec:alg}

We propose mirror descent to compute the quantum rate-distortion function. This is justified by the following relative smoothness properties of mutual information.
\begin{lem}[{{\cite[Theorem 6.4]{he2023mirror}}}]\label{lem:qrd-smooth}
    Consider any input state $\rho_A\in\mathcal{D}_n$. Quantum mutual information $I_q(\wc; \rho_A)$ is $1$-smooth relative to $-S$ on $\mathcal{D}_{mn}$.  
\end{lem}
Therefore, following from Fact~\ref{fact:conv}, the sequence generated by mirror descent with unit step size $t_k=1$ and kernel function $-S$ to compute~\eqref{eqn:qrd} will converge sublinearly to the optimal value. In the special case where we consider the classical-quantum distortion matrix~\eqref{eqn:cq-distortion} considered in Example~\ref{exmp:cq}, and if additionally, $\Delta_j=\Delta'$ for all $j=1,\ldots,m$, then mirror descent updates with unit step size are explicitly given by
\begin{align}
    \sigma_{BR}^{k+1} = \frac{\exp(\log(\tr_R(\sigma_{BR}^k)) + \Delta') \otimes \rho_A}{\tr[\exp(\log(\tr_R(\sigma_{BR}^k)) + \Delta')]}.
\end{align}
Unfortunately, for a general distortion matrix, it is not obvious how a similar efficiently computable update rule can be obtained. Rather than solve the mirror descent subproblem~\eqref{eqn:mirror-descent} directly, it is more convenient to solve the dual problem, which using~\eqref{eqn:md-dual} we can show is
\begin{equation}\label{eqn:dual-prob}
    \maximize_{\nu} \quad g_q^k(\nu),
\end{equation}
where
\begin{equation}\label{eqn:dual}
    g_q^k(\nu) \coloneqq -\tr[\exp(\log(\tr_R(\sigma_{BR}^k)) \otimes \mathbb{I}_R - \mathbb{I}_B \otimes \nu - \kappa \Delta)] - \inp{\rho_A}{\nu},
\end{equation}
and $\nu \in \mathbb{H}^n$ is the dual variable corresponding to the partial trace constraint~\eqref{eqn:qrd-pt}. Note that we have ignored the unit trace constraint $\tr[\sigma_{BR}]=1$ as this is already implied by the partial trace constraint. This dual problem~\eqref{eqn:dual-prob} is simpler to solve compared to the original primal subproblem as it is both unconstrained and has a significantly smaller problem dimension of $n^2$ compared to the original dimension of $n^2m^2$. Once a solution $\nu^{k+1}$ to~\eqref{eqn:dual-prob} is found, the primal solution can be recovered using~\eqref{eqn:md-p-var} as
\begin{equation}
    \sigma_{BR}^{k+1} = \exp(\log(\tr_R(\sigma_{BR}^k)) \otimes \mathbb{I}_R - \mathbb{I}_B \otimes \nu^{k+1} - \kappa \Delta).
    \label{eqn:qrd-iter}
\end{equation}
The overall algorithm is summarised in Algorithm~\ref{alg:qrd}. 

\begin{algorithm}
    \caption{Quantum rate-distortion algorithm}
    \begin{algorithmic}
        \Input{Input state $\rho_A\in\mathcal{D}_n$, distortion matrix $\Delta\in \mathbb{H}^{mn}_+$ and dual variable $\kappa\geq0$}
        \State \textbf{Initialize:} Initial point $\sigma_{BR}^0\in\mathcal{D}_{mn}$
        \For{$k=0,1,\ldots$}{}
        \begin{align*}
            \nu^{k+1} &= \argmax_{\nu} \, -\tr[\exp(\log(\tr_R(\sigma_{BR}^k)) \otimes \mathbb{I}_R - \mathbb{I}_B \otimes \nu - \kappa \Delta)] - \inp{\rho_A}{\nu}\\
            \sigma_{BR}^{k+1} &= \exp(\log(\tr_R(\sigma_{BR}^k)) \otimes \mathbb{I}_R - \mathbb{I}_B \otimes \nu^{k+1} - \kappa \Delta)
        \end{align*}
        \EndFor
        \Output{Approximate solution $\sigma_{BR}^{k+1}$.}
    \end{algorithmic}
    \label{alg:qrd}
\end{algorithm}

There are a number of standard algorithms available which can be used to solve the unconstrained convex optimization problem~\eqref{eqn:dual-prob}. We will consider methods of the form
\begin{equation}\label{eqn:dual-ascent-step}
    \nu_{i+1}^{k+1} = \nu_i^{k+1} + t_i G_i \nabla g_q^k(\nu_i^{k+1}), \quad \forall i\in\mathbb{N},
\end{equation}
where $G_i=\mathbb{I}$ corresponds to gradient ascent, and $G_i=-[\nabla^2g_q^k(\nu_i^{k+1})]^{-1}$ corresponds to Newton's method. We also take $t_i>0$ to be a step size obtained from a backtracking procedure~\cite[Algorithm 9.2]{boyd2004convex} which guarantees that the objective converges monotonically.

\begin{prop}\label{prop:dual}
    The dual problem~\eqref{eqn:dual-prob} satisfies the following properties:
    \begin{enumerate}[label=(\roman*), ref=\ref{prop:dual}(\roman*)]
        \item There exists a unique solution $\nu_*\in\mathbb{H}^n$ to the dual problem, and;
        \item Solving the problem using gradient descent or damped Newton's method with backtracking will produce a sequence of iterates $\{\nu_i\}$ which converges to $\nu_*$. \label{prop:dual-ii}
    \end{enumerate}
\end{prop}
\begin{proof}
    To show part (i), we first show the dual function~\eqref{eqn:dual} is coercive and strictly convex by using Proposition~\ref{prop:dual-prop}, recognizing that the tensor product $X\mapsto\mathbb{I}\otimes X$ is injective, and using Remark~\ref{rem:degenerate} to assume $\rho_A\succ0$ is full-rank. Existence and uniqueness of a solution then follows from~\cite[Proposition 3.1.1 and 3.2.1]{bertsekas2009convex}. To prove part (ii), we use the fact that the superlevel sets of the dual function are compact, and that the dual function is strongly concave and has Lipschitz gradient and Hessian over compacts sets. These properties allow us to use standard convergence results of the two algorithms. We leave the details of the proof for part (ii) in Appendix~\ref{appdx:dual}.
\end{proof}

Note that to use either of these descent methods, we need the gradient and Hessian of $g_q^k$. The gradient can easily be found using~\cite[Corollary B.5]{he2023mirror} as
\begin{equation}
    \nabla g_q^k(\nu) = \tr_B[ \exp(\log(\tr_R(\sigma_{BR}^k)) \otimes \mathbb{I}_R - \mathbb{I}_B \otimes \nu - \kappa \Delta) ] - \rho_A.
\end{equation}
To find the Hessian, let us define the diagonalization 
\begin{equation}\label{eqn:diag-dual-hess}
    \log(\tr_R(\sigma_{BR}^k)) \otimes \mathbb{I}_R - \mathbb{I}_B \otimes \nu - \kappa \Delta = U\Lambda U^\dag.
\end{equation}
We use~\cite[Theorem 6.6.30]{horn1994topics} to obtain the directional derivative of $\nabla g_q^k$ in direction $V\in\mathbb{H}^n$ as
\begin{equation}\label{eqn:hess-dual}
    \mathsf{D}\nabla g_q^k(\nu)[V] = -\tr_B[ U (  f^{[1]}(\Lambda) \odot ( U^\dag (\mathbb{I} \otimes V) U )  ) U^\dag ],
\end{equation}
where $\odot$ represents the Hadamard or element-wise product, and $f^{[1]}(\Lambda)$ is the first divided differences matrix of $\Lambda$ corresponding to $f(x)=e^x$, i.e., the
matrix whose $(i, j)$-th entry is given by $f^{[1]}(\lambda_{i}, \lambda_{j})$ where
\begin{subequations}\label{eqn:first-divided-diff}
    \begin{align}
        f^{[1]}(\lambda, \mu) &= \frac{f(\lambda) - f(\mu)}{\lambda - \mu}, \quad \textrm{if }\lambda\neq\mu \label{eqn:div-diff-a}\\
        f^{[1]}(\lambda, \lambda) &= f'(\lambda). \label{eqn:div-diff-b}
    \end{align}
\end{subequations}
Equation~\eqref{eqn:hess-dual} defines the Hessian as a linear map, which we can use to construct the Hessian matrix by expressing it in coordinates with respect to a choice of basis for Hermitian matrices.

\begin{rem}\label{rem:hayashi-ours}
    We note that Algorithm~\ref{alg:qrd} is very similar to~\cite[Algorithm 14]{hayashi2022bregman}, which was derived using an expectation-maximization approach. In fact, if we also dualize the redundant unit trace constraint then solve for the corresponding dual variable using the KKT conditions, we obtain the alternate dual function
    \begin{equation}\label{eqn:dual-hayashi}
        \tilde{g}_q^k(\tilde{\nu}) \coloneqq -\log(\tr[\exp(\log(\tr_R(\sigma_{BR}^k)) \otimes \mathbb{I}_R - \mathbb{I}_B \otimes \tilde{\nu} - \kappa \Delta)]) - \inp{\rho_A}{\tilde{\nu}}.
    \end{equation}
    By recognizing the identity $\log(A\otimes B)=\log(A)\otimes\mathbb{I}+\mathbb{I}\otimes\log(B)$, we see that~\eqref{eqn:dual-hayashi} is identical to the function being minimized in the subproblems of~\cite[Algorithm 14]{hayashi2022bregman} up to a simple translation. The advantage of the dual function~\eqref{eqn:dual} which we minimize in Algorithm~\ref{alg:qrd} is that the trace expression in $\tilde{g}_q^k$ is composed with a logarithm, and is therefore no longer strictly convex nor coercive. We can confirm this by showing that $\tilde{g}_q^k(t\mathbb{I})$ is a constant for all $t\in\mathbb{R}$. Therefore, we can no longer use the same arguments as Proposition~\ref{prop:dual} to prove that gradient descent or damped Newton's with backtracking will converge to the solution, as without coercivity the superlevel sets of $\tilde{g}_q^k$ are no longer compact. Additionally, our experimental results in Section~\ref{subsec:exp-alg} suggest that~\eqref{eqn:dual-hayashi} is a more difficult function to minimize.
\end{rem}

\begin{rem}\label{rem:linear}
    Up to this point, we have established theoretical guarantees that mirror descent applied to the quantum rate-distortion problem will result in sublinear convergence to the solution. However, for rate-distortion problems using the entanglement fidelity distortion, empirical convergence results show that mirror descent converges to the solutions at a linear rate (see Figure~\ref{fig:rate_distortions} in Appendix~\ref{appdx:linear}). Moreover, this linear rate appears to become faster as the parameter $\kappa$ increases. One standard way to prove global linear convergence of mirror descent is by using a relative strong convexity argument (see Fact~\ref{fact:conv} and, e.g.,~\cite{he2023mirror}). However, we can show that quantum mutual information is \emph{not} relatively strongly convex on the set of density matrices (see Appendix~\ref{appdx:linear}). 
    
    Instead, it is possible to gain some insight into this linear rate by studying the relative strong convexity properties of the quantum mutual information around a neighborhood of the solution to the rate-distortion problems. In particular, using a priori knowledge about the solutions to particular instances of rate-distortions problem analyzed in Section~\ref{sec:uniform}, we can characterize the local relative strong convexity parameters of quantum mutual information around the solutions, and use them to better understand the linear convergence behavior of mirror descent as a function of $\kappa$. We refer the interested reader to Appendix~\ref{appdx:linear} for a more in depth discussion about these linear rates.
\end{rem}

\subsection{Inexact mirror descent}\label{sec:inexact}

As we use a numerical method to solve the mirror descent subproblem~\eqref{eqn:mirror-descent} for Algorithm~\ref{alg:qrd}, it is important to account for numerical errors that will inevitably arise. Existing works~\cite{solodov2000error,schmidt2011convergence,yang2022bregman} have studied the convergence properties of proximal gradient methods where the update steps are numerically computed to some finite tolerance. Of particular interest is that we do not always need to solve the subproblem to high accuracy to guarantee convergence to the optimal value, which can significantly improve computation speeds. We first present a generalized algorithm, then show how to apply it to the quantum rate-distortion problem. We remind the reader that Assumption~\ref{assump:general} is made throughout this section.

\paragraph{General algorithm} To solve equality constrained problems of the form~\eqref{eqn:constr-min} using an inexact method, we propose Algorithm~\ref{alg:inexact}. This algorithm is adapted from~\cite{yang2022bregman} which is an inexact Bregman proximal point method for arbitrary convex constraints, whereas our algorithm is an inexact mirror descent method for linear equality constraints.

The algorithm works as follows. We first warm-start the dual variable $\nu$ from the previous iteration (see Step $0$ of Algorithm~\ref{alg:inexact}) based on the assumption that consecutive subproblems will have similar solutions. A primal variable $x\in\domain \interior\varphi$ is recovered using this dual variable using~\eqref{eqn:md-p-var} (see~\eqref{eqn:inexact-primal-a} in Algorithm~\ref{alg:inexact}). However as the dual variable is not necessarily dual optimal, the primal variable is not necessarily primal feasible. Therefore, we define a pseudo-projection $\proj_\mathcal{C} : \interior\domain \varphi \rightarrow \relinterior\mathcal{C}$ which is continuous, surjective and idempotent, and compute a feasible primal variable as $\Tilde{x} = \proj_\mathcal{C}(x) \in \mathcal{C}$ (see~\eqref{eqn:inexact-primal-b} in Algorithm~\ref{alg:inexact}). Finally, an error criterion based on~\cite{yang2022bregman} is used to check if the subproblem has solved to a sufficient accuracy (see Step $2$ of Algorithm~\ref{alg:inexact}). If not, then $\nu$ is updated in an ascent direction for $g^k$ (using e.g., gradient ascent or Newton's method), and the process repeats until the desired accuracy has been reached.

\begin{algorithm*}
    \caption{Inexact mirror descent algorithm}
    \begin{algorithmic}
        \Input{Objective function $f$, Legendre reference function $\varphi$ where $\domain\varphi = \domain f$, linear equality constraint data $\mathcal{A}, b$, step size $\{ t_k \}$, error tolerance $\{ \varepsilon_k \}$.}
        \State \textbf{Initialize:} Initial primal $x^0 \in\mathcal{C}$ and dual $\nu^0\in\mathbb{V}'$ variables.
        \For{$k=0,1,\ldots$}{}
            \State \textbf{Step 0}: Initialize dual variable from previous iteration, i.e., $\nu^{k+1}_0 = \nu^k$.
            \For{$i=0,1,\ldots$}{}
                \State \textbf{Step 1}: Recover primal variable from dual variable, i.e.,
                \begin{align}
                    x^{k+1}_i &= \nabla\varphi^{-1}[\nabla\varphi(x^k) - t_k (\nabla f(x^k) + \mathcal{A}^\dag (\nu^{k+1}_i))] \label{eqn:inexact-primal-a}\\
                    \Tilde{x}^{k+1}_i &= \proj_{\mathcal{C}}( x^{k+1}_i ), \label{eqn:inexact-primal-b}
                \end{align}
                \State \hspace{3.5em} where $\proj_\mathcal{C} : \interior\domain \varphi \rightarrow \relinterior\mathcal{C}$ is continuous, surjective and idempo-
                \State \hspace{3.5em} tent.
                \State \textbf{Step 2}: \textbf{if} subproblem has solved to sufficient accuracy, i.e., satisfies
                \begin{equation}
                    D_\varphi\divx{\Tilde{x}^{k+1}_i}{x^{k+1}_i} \leq \varepsilon_k, \label{eqn:inexact-exit}
                \end{equation}
                \State \hspace{3.5em} \textbf{then} 
                \State \hspace{3.5em}\hspace{\algorithmicindent} $\nu^{k+1}\coloneqq\nu^{k+1}_i$, $x^{k+1}\coloneqq x^{k+1}_i$, and $\tilde{x}^{k+1}\coloneqq\tilde{x}^{k+1}_i$
                \State \hspace{3.5em}\hspace{\algorithmicindent} \textbf{break} 
                \State \hspace{3.5em} \textbf{else} 
                \State \hspace{3.5em}\hspace{\algorithmicindent} Update dual variable, i.e., $\nu^{k+1}_{i+1} = \nu_i^{k+1} + t^{k+1}_i \Delta\nu_i^{k+1}$, where $t^{k+1}_i$ is  
                \State \hspace{3.5em}\hspace{\algorithmicindent} an appropriate step size and $\Delta\nu_i^{k+1}$ is an ascent direction for $g^{k}$.
                \State \hspace{3.5em} \textbf{end if}
            \EndFor
        \EndFor
        \Output{Approximate solution $\Tilde{x}^{k+1}$.}
    \end{algorithmic}
    \label{alg:inexact}
\end{algorithm*}

We now present the convergence properties of Algorithm~\ref{alg:inexact}. We first show that the exit criterion will always eventually hold.

\begin{prop}
    Consider the problem~\eqref{eqn:mirror-pre} for some $y\in\relinterior\mathcal{C}$ and Legendre function $\varphi$. Consider the sequence $\{ \nu_i \}\in\mathbb{V}'$, and let $\{ x_i \}$ and $\{ \tilde{x}_i \}$ be the corresponding sequences generated by~\eqref{eqn:inexact-primal-a} and~\eqref{eqn:inexact-primal-b} respectively, where $x^k=y$. If $\{ \nu_i \}$ converges to the dual optimum $\nu_*$ of the dual subproblem~\eqref{eqn:md-dual-prob}, then $D_\varphi\divx{\Tilde{x}_i}{x_i}\rightarrow0$.
\end{prop}
\begin{proof}
    Let us denote $x_*$ as the primal optimum of~\eqref{eqn:mirror-pre}. From Fact~\ref{fact:interior}, we have $\{ x_i \}\in\interior\domain\varphi$ and $x_* \in\interior\domain\varphi$. Using Fact~\ref{fact:legendre-ii}, we can show $(\nabla \varphi)^{-1}$ is continuous, and so $x_i$ is a continuous function of $\nu_i$ from~\eqref{eqn:inexact-primal-a}. Therefore, convergence of $\{ \nu_i \}$ to $\nu_*$ implies $\{ x_i \}$ converges to a solution of $\inf_x \mathcal{L}(x, \nu_*; y)$, which from strict convexity of $\varphi$ must be the unique primal optimum $x_*$~\cite[Theorem 12.13]{nocedal1999numerical}. From continuity of the pseudo-projection operator on $\interior\domain\varphi$ and primal feasibility of $x_*$, it follows that $\Tilde{x}_i=\proj_\mathcal{C}(x_i)$ must also converge to $x_*$. The desired result then follows from continuity of $D_\varphi$ on $\interior\domain\varphi \times \interior\domain\varphi$, and the property $D_\varphi\divx{x}{x}=0$ for all $x\in\interior\domain\varphi$.
\end{proof}
\begin{rem}
    For the quantum rate-distortion problem, convergence of the iterates $\{ \nu_i \}$ to the optimal solution of the subproblem is guaranteed by Proposition~\ref{prop:dual} if we use gradient ascent or Newton's method with a backtracking line search.
\end{rem}

We now move onto proving convergence rates of the inexact mirror descent iterates. We begin with a preliminary result which takes advantage of the fact that we restrict our attention to problems with only linear equality constraints.
\begin{lem}\label{prop:inexact-opt-condition}
    The pair of primal iterates $(x^{k+1}, \Tilde{x}^{k+1})$ produced by~\eqref{eqn:inexact-primal-a} and~\eqref{eqn:inexact-primal-b} always satisfy
    \begin{equation}
        0 \in \nabla f(x^k) + N_{\mathcal{A}}(\Tilde{x}^{k+1}) + \frac{1}{t_k}(\nabla\varphi(x^{k+1}) - \nabla\varphi(x^{k})), \label{eqn:inexact-opt-condition}
    \end{equation}
    where $N_\mathcal{A}$ denotes the normal cone of the affine set $\{ x\in\mathbb{V} : \mathcal{A}(x) = b \}$. 
\end{lem}
\begin{proof}
    As $\Tilde{x}^{k+1}$ is primal feasible, it follows that $N_{\mathcal{A}}(\Tilde{x}^{k+1})$ is equal to the image of $\mathcal{A}^\dag$. Rearranging~\eqref{eqn:inexact-primal-a} then substituting this in gives
    \begin{equation*}
        \frac{1}{t_k}(\nabla\varphi(x^{k}) - \nabla\varphi(x^{k+1})) - \nabla f(x^k) = \mathcal{A}^\dag (\nu^k) \in N_{\mathcal{A}}(\Tilde{x}^{k+1}),
    \end{equation*}
    for all $\nu^k \in \mathbb{V}'$, from which~\eqref{eqn:inexact-opt-condition} follows.
\end{proof}

The following main convergence result can been seen as a modification of~\cite[Theorem 3.2]{yang2022bregman}.

\begin{thm}\label{prop:adapt-rate-sublin}
    Consider Algorithm~\ref{alg:inexact} to solve the convex optimization problem~\eqref{eqn:constr-min}. Let $\varphi$ be Legendre, $f^*$ represent the optimal value of this problem, and $x^*$ be any corresponding optimal point. 
    \begin{enumerate}[label=(\roman*), ref=\ref{prop:adapt-rate-sublin}(\roman*)]
        \item If $f$ is $L$-smooth relative to $\varphi$ and $t_k=1/L$, then the sequence $\{ (x^{k}, \Tilde{x}^{k}) \}$ satisfies 
        \begin{equation}\label{eqn:adapt-rate-sublin}
            f(x_{avg}^k) - f^* \leq \frac{L}{k} \biggl(D_\varphi\divx{x^*}{x^0} + \sum_{i=0}^{k-1} \varepsilon_k \biggr), \quad \forall k\in\mathbb{N},
        \end{equation}
        where $x_{avg}^k = \sum_{i=1}^k \Tilde{x}^{i}/k$.
        \item If $f$ is $L$-smooth relative to $\varphi$, $\mu$-strongly convex relative to $\varphi$, and $t_k=1/L$, then the sequence $\{ (x^{k}, \Tilde{x}^{k}) \}$ satisfies 
        \begin{equation}\label{eqn:inexact-lin}
        D_\varphi\divx{x^*}{x^{k}} \leq \mleft( 1 - \frac{\mu}{L} \mright)^{k} \biggl(D_\varphi\divx{x^*}{x^0} + A \biggr), \quad \forall k\in\mathbb{N},
        \end{equation}
        where
        \begin{equation}
            A = \sum_{i=0}^{k-1}\mleft( 1-\frac{\mu}{L} \mright)^{-1-i}\varepsilon_i.
        \end{equation}
    \end{enumerate}
\end{thm}

\begin{proof}
    From $L$-smoothness of $f$, we have
    \begin{align*}
        f(\Tilde{x}^{k+1}) \leq f(x^k) + \inp{\nabla f(x^k)}{\Tilde{x}^{k+1} - x^k} + LD_\varphi\divx{\Tilde{x}^{k+1}}{x^k}.
    \end{align*}
    By substracting $f(u)$ from both sides of the inequality for any $u\in\mathcal{C}$ and using $\mu$-strong convexity of $f$, we obtain
    \begin{align}\label{eqn:inexact-smooth-strong}
        f(\Tilde{x}^{k+1}) - f(u) &\leq  \inp{\nabla f(x^k)}{\Tilde{x}^{k+1} - u} + LD_\varphi\divx{\Tilde{x}^{k+1}}{x^k} - \mu D_\varphi\divx{u}{x^k}.
    \end{align}    
    Using Lemma~\ref{prop:inexact-opt-condition} and the definition of the normal cone $N_\mathcal{A}(\Tilde{x}^{k+1})$ of $\{ x\in\mathbb{V} : \mathcal{A}(x) = b \}$ at $\Tilde{x}^{k+1}$, we can show that
    \begin{align*}
        \inp{\nabla f(x^k)}{\Tilde{x}^{k+1} - u} &\leq L \inp{ \nabla\varphi(x^{k}) - \nabla\varphi(x^{k+1}) }{\Tilde{x}^{k+1} - u} \nonumber \\
        &= L ( D_\varphi\divx{\Tilde{x}^{k+1}}{x^{k+1}} + D_\varphi\divx{u}{x^k} - D_\varphi\divx{\Tilde{x}^{k+1}}{x^k} - D_\varphi\divx{u}{x^{k+1}} )\nonumber \\
        &\leq L ( D_\varphi\divx{u}{x^k} - D_\varphi\divx{\Tilde{x}^{k+1}}{x^{k}} - D_\varphi\divx{u}{x^{k+1}} + \varepsilon_k), \label{eqn:inexact-rate-proof-a}
    \end{align*}
    for any $u\in\mathcal{C}$, where the equality uses the definition of the Bregman divergence~\eqref{eqn:breg-div}, and the second inequality uses~\eqref{eqn:inexact-exit}. Combining this with~\eqref{eqn:inexact-smooth-strong} and letting $u=x^*$ gives
    \begin{equation}
        f(\Tilde{x}^{k+1}) - f^* \leq  (L-\mu)D_\varphi\divx{x^*}{x^k} - LD_\varphi\divx{x^*}{x^{k+1}} + L\varepsilon_k. \label{eqn:inexact-proof}
    \end{equation}
    To show the part (i), let $\mu=0$, i.e., we make no assumption of relative strong convexity. From convexity of $f$, we find that
    \begin{align*}
        f(x_{avg}^k) - f^* &\leq \frac{1}{k} \sum_{i=1}^k (f(\Tilde{x}^{i}) - f^*)\\
        &\leq \frac{L}{k} \biggl(D_\varphi\divx{x^*}{x^0} - D_\varphi\divx{x^*}{x^{k}} + \sum_{i=0}^{k-1} \varepsilon_k \biggr),
    \end{align*}
    as the sum telescopes, which recovers the desired result. To show part (ii), we recognize that $f^* \leq f(x)$ for all $x\in\mathcal{C}$, and therefore from~\eqref{eqn:inexact-proof} we obtain
    \begin{align*}
        D_\varphi\divx{x^*}{x^{k+1}} &\leq \mleft(1 - \frac{\mu}{L}\mright)D_\varphi\divx{x^*}{x^k} + \varepsilon_k.
    \end{align*}
    Recursively applying this inequality recovers the desired result.
\end{proof}
\begin{rem}\label{rem:inexact}
    If $\varepsilon_k=\varepsilon$ is a constant for all $k$, then~\eqref{eqn:adapt-rate-sublin} implies ergodic sublinear convergence to an $L\varepsilon$-optimal solution. If $f$ is $\mu$-strongly convex relative to $\varphi$, then~\eqref{eqn:inexact-lin} implies linear convergence to a $L\varepsilon/\mu$-Bregman ball of the solution. If $\{ \varepsilon_k \}$ is a summable (i.e., $\lim_{k\rightarrow\infty} \sum_{i=0}^k \varepsilon_k < \infty$), then~\eqref{eqn:adapt-rate-sublin} implies ergodic sublinear convergence to the optimal solution. If $f$ is $\mu$-strongly convex relative to $\varphi$ and $\{ \varepsilon_k \}$ is a sequence which converges to zero linearly, then~\eqref{eqn:inexact-lin} implies linear convergence to the solution. See~\cite[Proposition 3]{schmidt2011convergence} for a more in-depth discussion on how the linear rate varies depending on how quickly $\{ \varepsilon_k \}$ converges. Notably, the best linear convergence behavior occurs when $\{ \varepsilon_k \}$ convergences at a similar rate as that of the function values.
\end{rem}

\begin{rem}\label{lem:gap-dual}
    The criterion~\eqref{eqn:inexact-exit} is equivalent to 
    \begin{equation}
        \inp{\nabla f(x^k)}{\Tilde{x}^{k+1}} + \frac{1}{t_k} D_\varphi\divx{\Tilde{x}^{k+1}}{x^k} - g^k(\nu^{k+1}) \leq \frac{\varepsilon_k}{t_k},
    \end{equation}    
    i.e., a bound on the primal-dual gap of the mirror descent subproblem~\eqref{eqn:mirror-descent}.
\end{rem}

\paragraph{Application to quantum rate-distortion} 
To use Algorithm~\ref{alg:inexact} for the quantum rate-distortion problem, all that remains is to propose a suitable pseudo-projection~\eqref{eqn:inexact-primal-b}. We propose the following operator
\begin{subequations}
    \begin{align}
        &\proj_\mathcal{C}(\sigma_{BR}) = P \sigma_{BR} P^\dag \label{eqn:qrd-proj}\\
        \textrm{for}\quad & P = \mathbb{I}_B \otimes ( \rho_A^{1/2} \tr_B(\sigma_{BR})^{-1/2} ). \label{eqn:qrd-proj-b}
    \end{align}
\end{subequations}
It follows from Lemma~\ref{lem:p-trace-unitary}, where $X=R$, $Y=B$, $A_X=( \rho_A^{1/2} \tr_B(\sigma_{BR})^{-1/2} )$ and $U_Y=\mathbb{I}_B$, that the proposed pseudo-projection maps positive definite matrices to positive definite matrices satisfying~\eqref{eqn:qrd-pt}. Idempotence follows by recognizing that for any $\sigma_{BR}\in\mathbb{H}^{mn}_{++}$ satisfying~\eqref{eqn:qrd-pt}, i.e., $\tr_B(\sigma_{BR}) = \rho_A$, the pseudo-projection satisfies $\proj_\mathcal{C}(\sigma_{BR}) = \sigma_{BR}$ because the matrix $P$ from~\eqref{eqn:qrd-proj-b} satisfies
\begin{align*}
    \quad P &= \mathbb{I}_B \otimes ( \rho_A^{1/2} \tr_B(\sigma_{BR})^{-1/2} ) \\
    &= \mathbb{I}_B \otimes ( \rho_A^{1/2} \rho_A^{-1/2} )\\
    &= \mathbb{I}_B\otimes\mathbb{I}_A.
\end{align*}
Continuity of the pseudo-projection follows from the fact that the matrix inverse and square root are continuous on the positive definite cone.

Therefore, \eqref{eqn:qrd-proj} is a suitable pseudo-projection for the quantum rate-distortion problem. Note that the standard Euclidean projection is not suitable as it maps variables to the boundary of the domain, where the gradient of quantum mutual information is not well-defined. Finally, we summarize the inexact mirror descent algorithm applied to the quantum rate-distortion problem in Algorithm~\ref{alg:inexact-qrd} and its corresponding convergence guarantees.

\begin{algorithm*}
    \caption{Inexact quantum rate-distortion algorithm}
    \begin{algorithmic}
        \Input{Input state $\rho_A\in\mathcal{D}_n$, distortion matrix $\Delta\in \mathbb{H}^{mn}_+$, error tolerance $\{ \varepsilon_k \}$.}
        \State \textbf{Initialize:} Initial primal $\sigma_{BR}^0\in\mathcal{D}_{mn}$ and dual $\nu^0\in\mathbb{H}^n$ variables.
        \For{$k=0,1,\ldots$}{}
            \State \textbf{Step 0}: Initialize dual variable from previous iteration, i.e., $\nu^{k+1}_0 = \nu^k$.
            \For{$i=0,1,\ldots$}{}
                \State \textbf{Step 1}: Recover primal variable from dual variable, i.e.,
                \begin{align*}
                    \sigma_{BR, i}^{k+1} &= \exp(\log(\tr_R(\sigma_{BR}^k)) \otimes \mathbb{I}_R - \mathbb{I}_B \otimes \nu^{k+1}_i - \kappa \Delta)\\
                    \Tilde{\sigma}_{BR, i}^{k+1} &= P \sigma_{BR, i}^{k+1} P^\dag,
                \end{align*}
                \State \hspace{3.5em} where $P = \mathbb{I}_B \otimes ( \rho_A^{1/2} \tr_B(\sigma_{BR, i}^{k+1})^{-1/2} )$.
                \State \textbf{Step 2}: \textbf{if} subproblem has solved to sufficient accuracy, i.e., satisfies
                \begin{equation*}
                    S\divx{\Tilde{\sigma}_{BR, i}^{k+1}}{\sigma_{BR, i}^{k+1}} - \tr[\Tilde{\sigma}_{BR, i}^{k+1}] + \tr[\sigma_{BR, i}^{k+1}] \leq \varepsilon_k. 
                \end{equation*}
                \State \hspace{3.5em} \textbf{then} 
                \State \hspace{3.5em}\hspace{\algorithmicindent} $\nu^{k+1}\coloneqq\nu^{k+1}_i$, $\sigma_{BR}^{k+1}\coloneqq \sigma_{BR, i}^{k+1}$, and $\Tilde{\sigma}_{BR}^{k+1}\coloneqq\Tilde{\sigma}_{BR, i}^{k+1}$
                \State \hspace{3.5em}\hspace{\algorithmicindent} \textbf{break} 
                \State \hspace{3.5em} \textbf{else} 
                \State \hspace{3.5em}\hspace{\algorithmicindent} Update dual variable $\nu^{k+1}_{i+1}$ using~\eqref{eqn:dual-ascent-step}.
                \State \hspace{3.5em} \textbf{end if}
            \EndFor
        \EndFor
        \Output{Approximate solution $\Tilde{\sigma}_{BR}^{k+1}$.}
    \end{algorithmic}
    \label{alg:inexact-qrd}
\end{algorithm*}

\begin{cor}
    Consider the quantum rate-distortion problem~\eqref{eqn:qrd} for some input state $\rho_A\in\mathcal{D}_n$ and distortion matrix $\Delta\in\mathbb{H}^{mn}_+$. Let $\hat{I}^*$ denote its optimal value, and $\sigma_{BR}^*$ be any corresponding optimal point. The sequence $\{ (\sigma_{BR}^k, \Tilde{\sigma}_{BR}^k) \}$ generated by Algorithm~\ref{alg:inexact-qrd} satisfies
    \begin{equation}
        I_q(\sigma_{avg}^k; \rho_A) + \kappa\inp{\Delta}{\sigma_{avg}^k} - \hat{I}^* \leq \frac{1}{k} \biggl(S\divx{\sigma_{BR}^*}{\sigma_{BR}^0} + \sum_{i=0}^{k-1} \varepsilon_k \biggr), \quad \forall k\in\mathbb{N},
    \end{equation}
    where $\sigma_{avg}^k = \sum_{i=1}^k \Tilde{\sigma}_{BR}^i/k$.
\end{cor}
\begin{proof}
    This follows from Lemma~\ref{lem:qrd-smooth} and Theorem~\ref{prop:adapt-rate-sublin}.
\end{proof}


\section{Numerical experiments} \label{sec:exp}

We present numerical experiments to demonstrate the computational performance of Algorithm~\ref{alg:inexact-qrd} in solving the quantum rate-distortion problem~\eqref{eqn:qrd}. Input states $\rho_A$ are randomly generated using QETLAB~\cite{qetlab}, and the entanglement fidelity was used to measure distortion by letting $\Delta=\mathbb{I}-\rho_{AR}$, where $\rho_{AR}$ is the purification of $\rho_A$. All experiments were ran on MATLAB using an Intel i5-11700 CPU with 32GB of RAM.

For all implementations of Algorithm~\ref{alg:inexact-qrd}, unless otherwise stated, we use an error schedule of
\begin{equation}\label{eqn:err-sched}
    \varepsilon_k = \max \{\min\{ f(x^k) - f(x^{k-1}), \xi^k, \varepsilon_{k-1} \}, 10^{-15} \},
\end{equation}
where $\xi=0.9$ and $\varepsilon_{-1}=10^{-2}$. Based on the discussion from Remarks~\ref{rem:linear} and~\ref{rem:inexact}, the first term in the minimum aims to match the linear rate of the error sequence to that of the mirror descent iterations, the second term ensures that the error schedule decays linearly, the third term ensures the sequence is monotonic, and the last term accounts for machine precision. For brevity, we will use MD-N to refer to when Newton's method with backtracking is used to compute the inexact mirror descent iterates, and MD-GD to refer to when gradient descent with backtracking is used. We terminate MD-N once the change in objective value drops below $10^{-15}$, and MD-GD at $10^{-8}$ as the sublinear rate of gradient descent makes computing the iterates to higher accuracy impractical. All algorithms are also terminated if the computation time exceeds $\SI{3600}{\second}$. We use ``timeout'' to refer to situations when an algorithm is unable to perform a single iteration within this timeframe, and ``out of memory'' to refer to when there is insufficient RAM to run the algorithm. All backtracking line searches use the implementation described by~\cite[Algorithm 9.2]{boyd2004convex}, where $\alpha=0.1$ and $\beta=0.1$. For gradient descent methods, we start backtracking from $t_0=1000$, while for Newton's method we use $t_0=1$. We report upper bounds on absolute optimality gaps measured in bits, and unless otherwise stated are measured by computing a lower bound of the optimal value using MD-N and the method described in Appendix~\ref{appdx:lb}. Where relevent, methods are initialized using $\sigma_{BR}^0=\rho_A\otimes\rho_A$ and $\nu^0=-\log(\rho_A)$. 


\subsection{Comparison between exact and inexact methods}\label{subsec:exp-exact}

We first present some preliminary results comparing exact and inexact computation of the mirror descent iteration. We compare between variations of Algorithm~\ref{alg:inexact-qrd} when the iterates are computed using Newton's method with backtracking exactly, i.e., close to machine precision $\varepsilon_k=10^{-15}$ for all $k$, and inexactly using~\eqref{eqn:err-sched}. We also compare against when iterates are computed using gradient descent with backtracking inexactly using~\eqref{eqn:err-sched}, and to a constant error of $\varepsilon_k=10^{-8}$ for all $k$ as proposed in~\cite[Algorithm 4]{hayashi2022bregman}. All methods use symmetry reduction to reduce the problem dimension.

We summarize the results in Table~\ref{tab:exact-v-inexact}. We see that using Newton's method to compute the mirror descent iterates yields identical accuracies between exact and inexact implementations, despite the inexact method solving up to $6$ times faster. This difference also seems more pronounced for lower values of $\kappa$. Similarly, both inexact methods using gradient descent solve to similar accuracies, while the proposed decaying error solves approximately $3$ times faster compared to solving to a constant tolerance. 




%

\begin{table*}[t]
\footnotesize
\caption{Comparison between various strategies of choosing the tolerance to which mirror descent iterates are computed, when computing the quantum rate-distortion function with entanglement fidelity distortion at channel dimensions $n$ and distortion dual variables $\kappa$. We report the time elapsed (s) and upper bounds on the absolute optimality gap (bits) of each method. All methods use symmetry reduction to reduce the problem dimension. }
\label{tab:exact-v-inexact}
\centerline{
\begin{tabular*}{1\textwidth}{@{\extracolsep{\fill}}lrrcrcrcrc@{\extracolsep{\fill}}}
\toprule
\multicolumn{2}{l}{} & \multicolumn{2}{c}{MD-N ($\varepsilon_k=10^{-15}$)} & \multicolumn{2}{c}{MD-N ($\varepsilon_k=$ \eqref{eqn:err-sched})} & \multicolumn{2}{c}{MD-GD ($\varepsilon_k=10^{-8}$)}      & \multicolumn{2}{c}{MD-GD ($\varepsilon_k=$ \eqref{eqn:err-sched})}     \\ \cmidrule{3-4}\cmidrule{5-6}\cmidrule{7-8}\cmidrule{9-10}
$n$   &   $\kappa$               & Time & Gap & Time & Gap & Time & Gap & Time & Gap  \\ \midrule
$8$ &  $1.0$                & $.32$          & $\mathbf{\num{7e-9}}$   & $.17$          & $\mathbf{\num{7e-9}}$ & $.34$          & $\mathrm{\num{1e-7}}$ & $\mathbf{.10}$ & $\mathrm{\num{1e-7}}$    \\
  &   $3.0$                 & $.10$          & $\mathbf{\num{4e-9}}$   & $\mathrm{\mathbf{.07}}$          & $\mathbf{\num{4e-9}}$ & $.38$          & $\mathrm{\num{1e-8}}$ & $.14$ & $\mathrm{\num{8e-9}}$      \\
$32$  &   $3.0$               & $5.52$          & $\mathbf{\num{6e-8}}$   & $1.52$          & $\mathbf{\num{6e-8}}$ & $.90$          & $\mathrm{\num{2e-7}}$ & $\mathbf{.45}$ & $\mathrm{\num{2e-7}}$   \\
  &   $5.5$               & $.90$          & $\mathbf{\num{5e-9}}$   & $\mathbf{.32}$          & $\mathbf{\num{5e-9}}$ & $2.07$          & $\mathrm{\num{1e-8}}$ & $.81$ & $\mathrm{\num{6e-8}}$   \\
$128$  &   $7.0$              & $85.70$          & $\mathbf{\num{3e-8}}$   & $\mathbf{14.40}$          & $\mathbf{\num{3e-8}}$ & $34.75$          & $\mathrm{\num{6e-8}}$ & $14.97$ & $\mathrm{\num{6e-8}}$   \\
  & $8.5$                &  $31.83$          & $\mathbf{\num{6e-10}}$   & $\mathbf{8.91}$          & $\mathbf{\num{6e-10}}$ & $57.51$          & $\mathrm{\num{7e-9}}$ & $16.85$ & $\mathrm{\num{7e-8}}$   \\ \bottomrule
\end{tabular*}
}
\end{table*}


\subsection{Comparison with other algorithms}\label{subsec:exp-alg}



We compare Algorithm~\ref{alg:inexact-qrd} to existing algorithms which have been used to compute the quantum rate-distortion function. These include the backtracking primal-dual hybrid gradient (PDHG) algorithm from~\cite[Algorithm 3]{he2023mirror}, the expectation maximization (EM) algorithm from~\cite[Algorithm 14]{hayashi2022bregman}, and \textsc{CvxQuad}~\cite{fawzi2018efficient,fawzi2019semidefinite} with the SDPT3~\cite{toh1999sdpt3} solver. We show results with and without symmetry reduction for all of these methods. For PDHG, we use the same backtracking parameters as used in the experiments of~\cite{he2023mirror}. As noted in Remark~\ref{rem:hayashi-ours}, the EM algorithm is very similar to our mirror descent algorithm except that it solves a slight variation of the dual problem. Therefore to allow for a fair comparison, we use the same inexact method as Algorithm~\ref{alg:inexact} using gradient descent to compute these subproblems rather than solving to a constant inexact tolerance as was originally proposed in~\cite{hayashi2022bregman}, and which we showed was computationally slower in Section~\ref{subsec:exp-exact}. Both the PDHG and EM methods were terminated once the change in objective value dropped below $10^{-8}$. The default settings for \textsc{CvxQuad} were used, and we report the optimality gap returned by the solver.

\begin{rem}
    Just as our proposed algorithm shares similarities to EM, the inexact mirror descent method in Algorithm~\ref{alg:qrd} also shares similarities to PDHG. In particular, if we use gradient descent to compute the inexact mirror descent iterates, then, like PDHG, we are alternating between primal mirror descent iterates, and dual gradient ascent iterates. The main difference is that we may perform multiple dual steps for every primal step we take.
\end{rem}

We summarize the results without symmetry reduction in Table~\ref{tab:no-symm}, and results with symmetry reduction in Table~\ref{tab:symm}. We see that symmetry reduction results in a significant improvement in both computation time and memory requirements. We emphasize that without symmetry reduction, the quantum rate-distortion for a $9$-qubit channel ($n=512$) needs to optimize over a $70$ billion dimensional variable. Symmetry reduction reduces the optimization variable to a dimension of just $\num[group-separator={,}]{523776}$.

Overall, we see that our proposed method is able to obtain high-accuracy solutions relatively quickly. In particular, our method generally outperforms all existing benchmarks in both computation time and accuracy simultaneously. We note that it is not entirely clear whether Newton's method or gradient descent is better to solve our mirror descent subproblems. At high problem dimensions, gradient descent solves significantly faster to similar accuracy solutions. At low to medium problem dimensions however, Newton's method seems to solve faster for larger values of $\kappa$, while gradient descent solves faster for smaller $\kappa$. However, Newton's method allows us to solve the mirror descent iterates ``exactly'', which in turn allows us to efficiently compute an optimality gap using the method from Appendix~\ref{appdx:lb}.




\begin{table}[t]
    \caption{Comparison between various algorithms to compute for the quantum rate-distortion function with entanglement fidelity distortion, at channel dimensions $n$ and distortion dual variables $\kappa$. We report the time elapsed (seconds) and upper bounds on the absolute optimality gap (bits) of each method. We use ``$>3600.00$'' to refer when the algorithm has not solved to the desired tolerance within an hour, and report the optimality gap of the last completed iteration at this time. We use ``Timeout'' to refer to when the algorithm has not performed any iterations within an hour. }
    \centerline{
    \begin{subtable}{\textwidth}\centering
        \footnotesize
        \caption{Without symmetry reduction}
        \label{tab:no-symm}
        \resizebox{\columnwidth}{!}{%
        \begin{tabular}{@{}lrrrrcrcrcrc@{}}
        \toprule
        \multicolumn{2}{l}{} & \multicolumn{2}{c}{MD-N} & \multicolumn{2}{c}{MD-GD} & \multicolumn{2}{c}{PDHG~\cite{he2023mirror}}      & \multicolumn{2}{c}{EM~\cite{hayashi2022bregman}} & \multicolumn{2}{c}{\textsc{CvxQuad}~\cite{fawzi2018efficient,fawzi2019semidefinite}}     \\ \cmidrule{3-4}\cmidrule{5-6}\cmidrule{7-8}\cmidrule{9-10}\cmidrule{11-12}
        $n$   &   $\kappa$               & Time & Gap & Time & Gap & Time & Gap & Time & Gap & Time & Gap \\ \midrule
        $2$                    & $1.0$   & $.19$ & $\mathbf{\num{4e-9}}$  & $.12$ & $\mathrm{\num{2e-8}}$ & $\mathbf{.11}$ & $\mathrm{\num{1e-7}}$ & $.20$ & $\mathrm{\num{2e-5}}$ & $3.00$  & $\mathrm{\num{8e-9}}$          \\
                             & $3.0$   & $.13$ & $\mathbf{\num{1e-8}}$  & $.10$ & $\mathrm{\num{2e-8}}$ & $\mathbf{.05}$ & $\mathrm{\num{2e-6}}$ & $.22$ & $\mathrm{\num{5e-5}}$ & $2.94$  & $\mathrm{\num{2e-8}}$          \\
        $8$                    & $1.0$   & $4.88$ & $\mathbf{\num{2e-8}}$  & $2.30$ & $\mathrm{\num{1e-7}}$ & $\mathbf{.84}$ & $\mathrm{\num{6e-6}}$ & $7.56$ & $\mathrm{\num{4e-5}}$ & \multicolumn{2}{c}{Out of memory}          \\
                             & $3.0$   & $\mathbf{1.12}$ & $\mathbf{\num{1e-8}}$  & $2.59$ & $\mathrm{\num{2e-8}}$ & $1.40$ & $\mathrm{\num{9e-7}}$ & $6.06$ & $\mathrm{\num{3e-5}}$ & \multicolumn{2}{c}{Out of memory}          \\
        $32$                   & $3.0$   & $>3600.00$ & $\mathbf{\num{2e-7}}$ & $\mathbf{332.16}$ & $\mathrm{\num{4e-7}}$ & $355.38$ & $\mathrm{\num{3e-5}}$ & $>3600.00$ & $\mathrm{\num{6e-3}}$ & \multicolumn{2}{c}{Out of memory} \\
                             & $5.5$   & $1321.85$ & $\mathbf{\num{1e-7}}$ & $\mathbf{631.83}$ & $\mathrm{\num{2e-7}}$ & $1500.75$ & $\mathrm{\num{5e-5}}$ & $>3600.00$ & $\mathrm{\num{4e-3}}$ & \multicolumn{2}{c}{Out of memory} \\
        $128$                  & $7.0$   & \multicolumn{2}{c}{Timeout} & \multicolumn{2}{c}{Timeout} & $\mathbf{>3600.00}$ & $\mathbf{\num{2e-1}}$ & \multicolumn{2}{c}{Timeout} & \multicolumn{2}{c}{Out of memory} \\
                             & $8.5$   & \multicolumn{2}{c}{Timeout} & \multicolumn{2}{c}{Timeout} & $\mathbf{>3600.00}$ & $\mathbf{\num{1e-1}}$ & \multicolumn{2}{c}{Timeout} & \multicolumn{2}{c}{Out of memory} \\
        $512$                  & $9.5$   & \multicolumn{2}{c}{Out of memory}  & \multicolumn{2}{c}{Out of memory} & \multicolumn{2}{c}{Out of memory} & \multicolumn{2}{c}{Out of memory} & \multicolumn{2}{c}{Out of memory} \\
         & $11.0$  & \multicolumn{2}{c}{Out of memory}  & \multicolumn{2}{c}{Out of memory} & \multicolumn{2}{c}{Out of memory} &  \multicolumn{2}{c}{Out of memory}  & \multicolumn{2}{c}{Out of memory} \\ \bottomrule
        \end{tabular}
        }
    \end{subtable}
    }
    
    \vspace{0.5cm}    
    \centerline{
    \begin{subtable}{\textwidth}\centering
        \footnotesize
        \caption{With symmetry reduction}
        \label{tab:symm}
        \resizebox{\columnwidth}{!}{%
        \begin{tabular}{@{}lrrrrcrcrcrc@{}}
        \toprule
        \multicolumn{2}{l}{} & \multicolumn{2}{c}{MD-N} & \multicolumn{2}{c}{MD-GD} & \multicolumn{2}{c}{PDHG~\cite{he2023mirror}}      & \multicolumn{2}{c}{EM~\cite{hayashi2022bregman}} & \multicolumn{2}{c}{\textsc{CvxQuad}~\cite{fawzi2018efficient,fawzi2019semidefinite}}     \\ \cmidrule{3-4}\cmidrule{5-6}\cmidrule{7-8}\cmidrule{9-10}\cmidrule{11-12}
        $n$   &   $\kappa$               & Time & Gap & Time & Gap & Time & Gap & Time & Gap & Time & Gap \\ \midrule
        $2$                    & $1.0$   & $.06$ & $\mathrm{\num{4e-9}}$  & $\mathbf{.05}$ & $\mathrm{\num{2e-8}}$ & $\mathbf{.05}$ & $\mathrm{\num{7e-8}}$ & $.08$ & $\mathrm{\num{2e-8}}$ & $2.89$  & $\mathbf{\num{2e-9}}$          \\
                             & $3.0$   & $.05$ & $\mathrm{\num{1e-8}}$  & $\mathbf{.04}$ & $\mathrm{\num{2e-8}}$ & $.07$ & $\mathrm{\num{7e-8}}$ & $.07$ & $\mathrm{\num{5e-5}}$ & $2.61$  & $\mathbf{\num{2e-9}}$          \\
        $8$                    & $1.0$   & $.17$ & $\mathbf{\num{7e-9}}$  & $\mathbf{.10}$ & $\mathrm{\num{1e-7}}$ & $.12$ & $\mathrm{\num{6e-6}}$ & $.53$ & $\mathrm{\num{1e-6}}$ & $454.98$  & $\mathrm{\num{2e-8}}$          \\
                             & $3.0$   & $\mathbf{.07}$ & $\mathbf{\num{4e-9}}$  & $.16$ & $\mathrm{\num{8e-9}}$ & $.19$ & $\mathrm{\num{6e-7}}$ & $.41$ & $\mathrm{\num{5e-8}}$ & $686.11$  & $\mathrm{\num{6e-8}}$          \\
        $32$                   & $3.0$   & $1.52$ & $\mathbf{\num{6e-8}}$ & $\mathbf{.45}$ & $\mathrm{\num{2e-7}}$ & $1.84$ & $\mathrm{\num{2e-5}}$ & $14.83$ & $\mathrm{\num{2e-5}}$ & \multicolumn{2}{c}{Out of memory} \\
                             & $5.5$   & $\mathbf{.32}$ & $\mathbf{\num{5e-9}}$ & $.81$ & $\mathrm{\num{6e-8}}$ & $6.51$ & $\mathrm{\num{5e-5}}$ & $17.25$ & $\mathrm{\num{5e-7}}$ & \multicolumn{2}{c}{Out of memory} \\
        $128$                  & $7.0$   & $\mathbf{14.40}$ & $\mathbf{\num{3e-8}}$ & $14.97$ & $\mathrm{\num{6e-8}}$ & $47.07$ & $\mathrm{\num{1e-4}}$ & $>3600.00$ & $\mathrm{\num{6e-6}}$ & \multicolumn{2}{c}{Out of memory} \\
                             & $8.5$   & $\mathbf{8.91}$ & $\mathbf{\num{6e-10}}$ & $16.84$ & $\mathrm{\num{7e-8}}$ & $100.55$ & $\mathrm{\num{1e-4}}$ & $>3600.00$ & $\mathrm{\num{6e-7}}$ & \multicolumn{2}{c}{Out of memory} \\
        $512$                  & $9.5$   & $2174.38$ & $\mathbf{\num{7e-8}}$  & $\mathbf{140.97}$ & $\mathrm{\num{1e-7}}$ & $1650.75$ & $\mathrm{\num{3e-4}}$ & $>3600.00$ & $\mathrm{\num{9e-5}}$ & \multicolumn{2}{c}{Out of memory} \\
        & $11.0$  & $1216.96$ & $\mathbf{\num{5e-9}}$  & $\mathbf{575.86}$ & $\mathrm{\num{1e-8}}$ & $>3600.00$ & $\mathrm{\num{5e-4}}$ &  $>3600.00$  &  $\mathrm{\num{1e-4}}$  & \multicolumn{2}{c}{Out of memory} \\ \bottomrule
        \end{tabular}
        }
     \end{subtable}
     }
\end{table}

\section{Concluding remarks} \label{sec:conc}

We have presented an inexact mirror descent algorithm which is able to efficiently compute the quantum rate-distortion function with provable sublinear convergence rates. Additionally, we have shown how we can exploit symmetries that exist in common quantum rate-distortion settings to significantly reduce the dimensions of the problems we need to solve. It would be interesting to see if our techniques and analysis can be extended to other similar problems, such as the quantum rate-distortion for mixed states~\cite{khanian2023rate} and the quantum information bottleneck function~\cite{hayashi2023efficient,hayashi2023generalized}.

In this work, we have focused on a reformulation of the quantum rate-distortion problem where we dualize the distortion inequality constraint, as is typically done for these constrained channel capacity-type problems~\cite{blahut1972computation}. However, by converting the distortion inequality to an equality constraint using~\cite[Lemma 21]{hayashi2022bregman}, it is relatively straightforward to extend most of our algorithm and analysis to directly account for the distortion constraint. The only limitation is that there is not an obvious pseudo-projection~\eqref{eqn:inexact-primal-b} for this new problem. Therefore, future work to find an efficient pseudo-projection operator will allow us to use our inexact mirror descent algorithm to directly solve the unsimplified rate-distortion problem. Alternatively, we could perform exact computations of the mirror descent steps, which, although slower, circumvents the need for a pseudo-projection.

Computation of the mirror descent subproblems is an important part of inexact mirror descent which significantly impacts how effective the overall algorithm is. In Proposition~\ref{prop:dual}, we established qualitative properties of the dual function which allows us to guarantee global convergence when maximizing it using standard techniques. However, we lack a more quantitative understanding of these functions which may allow us to have a better understanding of convergence rates. A more in-depth analysis of the dual function may allow us to understand how to best solve the mirror descent subproblems. 




\appendix

\section{Classical rate-distortion} \label{appdx:crd}

In this section, we study the classical rate-distortion function, showing how symmetry reduction and mirror descent can be applied to this problem. The purpose of this is to provide a point of comparison for the concepts we use to study the quantum rate-distortion function, and to show how many of our results still apply to the classical setting.

\subsection{Notation}

We define $n$-dimensional classical states as $n$-dimensional probability distributions, i.e.,
\begin{equation*}
    \mathcal{P}_n \coloneqq \biggl\{p \in \mathbb{R}^n_+ : \sum_{i=1}^n p_i = 1 \biggr\}.
\end{equation*}
\emph{Classical channels} are linear functions that map $n$-dimensional classical states to $m$-dimensional classical states, and are mathematically represented as $m\times n$ column stochastic matrices
\begin{equation*}
    \mathcal{Q}_{m,n} \coloneqq \biggl\{ Q\in\mathbb{R}^{m\times n}_+ : \sum_{i=1}^m Q_{ij} = 1, \forall j=1,\ldots,n \biggr\}.
\end{equation*}
Probabilities defined on the Cartesian product of two sample spaces are represented by joint probability matrices of the form
\begin{equation*}
    \mathcal{P}_{n\times m} \coloneqq \biggl\{p \in \mathbb{R}^{n\times m}_+ : \sum_{i=1}^n\sum_{j=1}^m p_{ij} = 1 \biggr\}.
\end{equation*}
With some overloading of function notation, we define the Shannon entropy $H:\mathbb{R}_+^n\rightarrow\mathbb{R}$ as
\begin{equation}
    H(x)\coloneqq-\sum_{i=1}^n x_i \log(x_i),
\end{equation}
and the Kullback-Leibler (KL) divergence $H:\mathbb{R}_+^n\times\mathbb{R}_{++}^n\rightarrow\mathbb{R}$ as
\begin{equation}
    H\divx{x}{y}\coloneqq\sum_{i=1}^n x_i \log(x_i / y_i).
\end{equation}
Throughout this section we will use $\bm{1}$ to denote the all ones column vector.

\subsection{Problem definition}
For a classical channel with $n$ inputs and $m$ outputs, consider a probability distribution $p\in\mathcal{P}_n$, a distortion matrix $\delta\in\mathbb{R}^{m\times n}_+$ where $\delta_{ij}$ represents the distortion of producing the $i$-th output from the $j$-th input, and a maximum allowable total distortion $D\geq0$. The corresponding \emph{classical rate-distortion function}~\cite{shannon1948mathematical} is
\begin{subequations}\label{eqn:pre-crd}
    \begin{align}
        R_c(D; p, \delta) \quad \coloneqq \quad \minimize_{Q\in\mathcal{Q}_{m,n}} \quad & I_c(Q; p) \\
        \subjto \quad & \sum_{i=1}^m \sum_{j=1}^n p_jQ_{ij}\delta_{ij} \leq D,
    \end{align}
\end{subequations}
where
\begin{equation}
    I_c(Q; p) \coloneqq H\divx{Q\diag(p)}{ Qpp^\top },
\end{equation}
is known as the \emph{classical mutual information}. It follows from joint convexity of the KL divergence~\cite[Theorem 2.7.2]{cover1999elements} that~\eqref{eqn:pre-crd} is a convex optimization problem.

When there is the same number of inputs and outputs, i.e., $m=n$, a common distortion measure for the classical rate-distortion function is the Hamming distance~\cite{cover1999elements}, which corresponds to $\delta = \bm{1}\bm{1}^\top - \mathbb{I}$ (i.e., the square matrix with zeros on the diagonal and ones everywhere else).

\begin{rem}
    Like the quantum analog (see Remark~\ref{rem:degenerate}), without loss of generality, we will always assume that the input distribution of the classical channel is non-degenerate, i.e., $p > 0$. See Appendix~\ref{appdx:rd-equiv-degenerate} for more details.
\end{rem}

\paragraph{Change of variables}
We will make a change of variables where we optimize over the joint distribution $P\in\mathcal{P}_{m\times n}$, where $P_{ij} = p_jQ_{ij}$ is the joint probability of obtaining the $i$-th output from the $j$-th input. For distortion $D\geq0$, input distribution $p\in\mathcal{P}_n$, and distortion matrix $\delta\in\mathbb{R}_+^{m\times n}$, this gives us
\begin{subequations}\label{eqn:new-crd}
    \begin{align}
        R_c(D; p, \delta) \quad = \quad \minimize_{P\in\mathcal{P}_{m\times n}} \quad & I_c(P; p) \\
        \subjto \quad & \sum_{i=1}^m P_{ij} = p_j, \qquad j=1,\ldots,n,\\
                \quad & \sum_{i=1}^m \sum_{j=1}^n P_{ij}\delta_{ij} \leq D, \label{eqn:crd-c}
    \end{align}
\end{subequations}
where with some slight abuse of notation we redefine classical mutual information as
\begin{equation}
    I_c(P; p) \coloneqq H\divx{P}{ P\bm{1}p^\top }.
\end{equation}
\begin{prop}\label{prop:crd-equiv}
    The classical rate-distortion problems~\eqref{eqn:pre-crd} and~\eqref{eqn:new-crd} are equivalent, in the sense that the optimal values of the two problems are equal.
\end{prop}
\begin{proof}
    See Appendix~\ref{appdx:rd-equiv-state}.
\end{proof}

\paragraph{Distortion constraint}
Like the quantum case, we will reparameterize the distortion constraint with the dual variable $\kappa\geq0$ instead of $D$. The classical rate-distortion problem for distortion dual variable $\kappa\geq0$, input distribution $p\in\mathcal{P}_n$, and distortion matrix $\delta\in\mathbb{R}_+^{m\times n}$ then becomes
\begin{subequations}
    \makeatletter
    \def\@currentlabel{CRD}\label{eqn:crd}
    \makeatother
    \renewcommand{\theequation}{CRD.\alph{equation}}
    \begin{align}
        \tilde{R}_c(\kappa; p, \delta) \quad = \quad \minimize_{P\in\mathcal{P}_{m \times n}} \quad & I_c(P; p) + \kappa \sum_{i=1}^m \sum_{j=1}^n P_{ij}\delta_{ij} \label{eqn:crd-obj} \\
        \subjto \quad & \sum_{i=1}^m P_{ij} = p_j, \quad j=1,\ldots,n, \label{eqn:crd-pt}
    \end{align}
\end{subequations}

\subsection{Symmetry reduction}

To show how we can achieve a similar result as Theorem~\ref{thm:qrd-maxmix-solution} for the classical scenario, consider the classical rate-distortion problem with Hamming distortion $\delta=\bm{1}\bm{1}^\top - \mathbb{I}$ and a uniform input distribution $p=\bm{1}/n$. The analytic expression for this setting is well known (see, e.g.,~\cite[Exercise 10.5]{cover1999elements}), which is typically derived by lower bounding the mutual information using Fano's inequality, then showing achievability of this lower bound. Here, we show how this result can also be derived by using a similar symmetry reduction as we used to derive the quantum rate-distortion function.
\begin{cor}\label{cor:reals-dim-class}
    Let $\mathcal{G}$ be a compact group, let $(\mathbb{R}^n,\tau)$ be a real representation of $\mathcal{G}$, and let $(\mathbb{R}^{n\times n},\pi)$ be the real representation of $\mathcal{G}$ defined by $\pi(g)(X) = \tau(g)X\tau(g)^\dag$. Then the dimension (over $\mathbb{R}$) of the fixed point subspace is given by $\dim_{\mathbb{R}}(\mathcal{V}_{\pi}) = \langle \chi_\tau,\chi_\tau\rangle$.
\end{cor}
\begin{proof}
    This follows from a similar argument as the proof for Corollary~\ref{cor:reals-dim}. We first extend $(\mathbb{R}^n,\tau)$ to a complex representation in the natural way. We then note that $(\mathbb{C}^{n\times n}, \pi)$ decomposes into isomorphic subrepresentations given by the restriction to real and imaginary subspaces.
\end{proof}

\begin{lem}\label{lem:crd-maxmix-rep}
    Let us define the group $S_n$ as the group of permutations on $\{ 1,\ldots,n \}$, and a representation $(\mathbb{R}^{n\times n}, \pi)$ of $S_n$ as
    \begin{equation*}
        \pi(g)(X) = P_g X P_g^\top, \quad \forall X\in\mathbb{R}^{n^2\times n^2},
    \end{equation*}
    where $P_g$ is the matrix representation of permutation $g$. Problem~\eqref{eqn:crd} for uniform input distribution $p=\bm{1}/n$ and Hamming distortion $\delta=\bm{1}\bm{1}^\top - \mathbb{I}$ is invariant under $\pi$.
\end{lem}
\begin{proof}
    Classical mutual information and the constraint~\eqref{eqn:crd-pt} are invariant under column-wise and row-wise permutations. Moreover, as conjugation by permutation matrices always maps diagonal elements to diagonals, and off-diagonal elements to off-diagonals, $\delta$ is also invariant under $\pi$.
\end{proof}
\begin{cor}
    A solution to~\eqref{eqn:crd} for input distribution $p=\bm{1}/n$ and Hamming distortion $\delta=\bm{1}\bm{1}^\top - \mathbb{I}$ exists in the subspace
    \begin{equation}
        \mathcal{V}_\pi = \{ a\mathbb{I} + b\bm{1}\bm{1}^\top : a,b\in\mathbb{R} \}.
    \end{equation}
\end{cor}
\begin{proof}
    It follows from Lemmas~\ref{lem:fix-pnt} and~\ref{lem:crd-maxmix-rep} that a solution must live in the fixed-point subspace of the representation $(\mathbb{R}^{n \times n}, \pi)$ defined in Lemma~\ref{lem:crd-maxmix-rep}. It is well known that the matrix representation $(\mathbb{C}^n, \tau)$ of the permutation group $S_n$ (i.e., $\tau(g)=P_g$) can be decomposed as the direct sum of two irreducible subrepresentations $V=\{ \bm{1} \}$ and $W = \{ x\in\mathbb{C}^n : \inp{\bm{1}}{x} = 0 \}$. Therefore using Corollary~\ref{cor:reals-dim-class}, we can show the dimension over the reals of the corresponding fixed-point subspace is
    \begin{align*}
        \dim_{\mathbb{R}} \mathcal{V}_\pi = \inp{\chi_\tau}{\chi_\tau} = \inp{\chi_V + \chi_W}{\chi_V + \chi_W} = 2,
    \end{align*}
    where the second and third equalities use Lemmas~\ref{lem:char-add} and~\ref{lem:char-ortho} respectively. It is straightforward to verify that the proposed two-dimensional subspace is invariant under $\pi$, which concludes the proof.
\end{proof}
\begin{thm}\label{thm:crd-maxmix-solution}
    The classical rate-distortion function with uniform input distribution $p_j=\bm{1}/n$ and Hamming distortion $\delta=\bm{1}\bm{1}^\top - \mathbb{I}$ is given by
    \begin{equation}
        R_c(D; p, \delta) = \begin{cases} 
            \log(n) - H \biggl( \biggl[1 - D, \underbrace{\frac{D}{n - 1}, \ldots, \frac{D}{n - 1}}_{(n-1) \textrm{ copies}} \biggr] \biggr), \quad &\textnormal{if } 0\leq D \leq \displaystyle 1-\frac{1}{n}\\
            0, &\textnormal{if } \displaystyle 1-\frac{1}{n} \leq D \leq 1.
        \end{cases}
    \end{equation}
    A channel which attains this rate for $0\leq D \leq 1-1/n$ is given by
    \begin{equation}
        Q = \frac{D}{n-1}\bm{1}\bm{1}^\top + \mleft(1 - \frac{nD}{n-1} \mright) \mathbb{I}.
    \end{equation}
\end{thm}
\begin{proof}
    This follows from a similar proof to that of Theorem~\ref{thm:qrd-maxmix-solution}.
\end{proof}

\subsection{Mirror descent}

Like the quantum case, we propose mirror descent to compute the classical rate-distortion function. This is justified by the following relative smoothness properties of mutual information.
\begin{lem}[{{\cite[Theorem 6.4]{he2023mirror}}}]
    Consider any input distribution $p\in\mathcal{P}_n$. Classical mutual information $I_c(\wc; p)$ is $1$-smooth relative to $-H$ on $\mathcal{P}_{m\times n}$.
\end{lem}
Therefore, following from Fact~\ref{fact:conv}, mirror descent with a unit step size $t_k=1$ and kernel function $-H$ is guaranteed to converge sublinearly to the optimal value. The mirror descent update~\eqref{eqn:mirror-descent} for the classical rate-distortion problem can be expressed as the following analytic expression
\begin{equation}
    P_{ij}^{k+1} = p_j \frac{q^k_i \exp(-\kappa \delta_{ij})}{\sum_{i=1}^m q^k_i \exp(-\kappa \delta_{ij})} \quad \textrm{for} \quad i=1,\ldots,m,~j=1,\ldots,n. \label{eqn:crd-update}
\end{equation}
where $q^k_i = \sum_{j=1}^n P_{ij}^k$ is the marginal probability of obtaining the $i$-th output. This iteration is precisely the same as the Blahut-Arimoto algorithm~\cite{blahut1972computation}.

Like the quantum rate-distortion problem with the entanglement fidelity distortion measure, empirical results show that using the mirror descent iterates~\eqref{eqn:crd-update} to compute the classical rate-distortion with Hamming distortion results in a linear convergence rate. More details about this can be found in Appendix~\ref{appdx:linear}.

\section{Equivalencies between rate-distortion problems} \label{appdx:rd-equiv}

To prove that two optimization problems are equivalent, it suffices to show that we can use a feasible point of one optimization problem to obtain a feasible point of the other problem with the same objective value, and vice versa. 

To show equivalencies between quantum rate-distortion problems, we will first introduce the Choi representation of a quantum channel, which will be a convenient mathematical representation than~\eqref{eqn:kraus-choi} for us to use to analyze the quantum rate-distortion problem.
\begin{defn}[Choi operator]\label{defn:choi}
    Consider a linear map $\mathcal{N}:\mathbb{H}^n\rightarrow\mathbb{H}^m$ which maps Hermitian matrices on an $n$-dimensional input system $A$ to Hermitian matrices on an $m$-dimensional output system $B$. The \emph{Choi matrix representation} $C_\mathcal{N}\in\mathbb{H}^{mn}$ of $\mathcal{N}$, which acts on the bipartite system $BA$, is defined as
    \begin{equation}
        C_\mathcal{N} \coloneqq \sum_{ij}^n \mathcal{N}(v_iv_j^\dag) \otimes v_iv_j^\dag,
    \end{equation}
    where $\{ v_i \}$ is any orthonormal basis for $A$.
\end{defn}

A linear operator $\mathcal{N}$ is completely positive if and only if $C_\mathcal{N}\succeq0$~\cite[Theorem 2.22]{watrous2018theory}, and is trace preserving if and only if $\tr_B(C_\mathcal{N})=\mathbb{I}_A$~\cite[Theorem 2.26]{watrous2018theory}. Applying a linear operator on input state $\rho_A\in\mathcal{D}_n$ is equivalent to $\mathcal{N}(\rho_A)=\tr_A(C_\mathcal{N}(\mathbb{I}_B\otimes\rho_A))$~\cite[Equation (2.66)]{watrous2018theory}.

\subsection{Basis independence} \label{appdx:rd-equiv-basis}
We will first show that that the quantum rate-distortion problem is independent of the eigenvectors of the input state $\rho_A$ by showing that we can construct an equivalent optimization problem by performing a suitable unitary transform on systems $A$ and $R$. For a given spectral decomposition $\rho_A=\sum_{i=1}\lambda_i v_iv_i^\dag$ of $\rho_A$, let us define the unitary $U=\sum_{i=1}^{n} v_iu_i^\dag$ where $\{u_i\}$ is any orthonormal basis for $A$. We can define a new problem with input state $U^\dag \rho_A U$, purification $(U\otimes U)^\dag \rho_{AR} (U\otimes U)$ and distortion matrix $(\mathbb{I}\otimes U)^\dag \Delta (\mathbb{I}\otimes U)$. For any feasible channel $\mathcal{N}$ of the original problem~\eqref{eqn:pre-qrd}, we can construct feasible channel $\mathcal{N}'$ for the new problem with the same objective value by letting
\begin{equation*}
    C_{\mathcal{N}'} = (\mathbb{I}\otimes U)^\dag C_{\mathcal{N}} (\mathbb{I}\otimes U).
\end{equation*}
Equivalency in the reverse direction follows from an identical argument. We note that a similar argument can be used to show that rate-distortion problem is independent of the choice of purification, even when $R$ is chosen not to be isomorphic to $A$.

\subsection{Degenerate and non-degenerate inputs} \label{appdx:rd-equiv-degenerate}
Here, we show that we can always assume that the input states $p$ and $\rho_A$ are non-degenerate (i.e., $p>0$ and $\rho_A\succ0$), as it is possible to construct equivalent optimization problems of a lower input dimension to avoid any degeneracies. For the classical rate-distortion function, we can construct a new problem where we simply ignore the $j$-th element of $p_j$ and column of $\delta$ for all $j$ corresponding to $p_j=0$. Clearly, given any feasible channel $Q$ of the original problem~\eqref{eqn:pre-crd} we can construct a feasible channel $Q'$ to the reduced problem with the same objective value by dropping the corresponding columns of $Q$. In the reverse direction, we just need to populate the columns of $Q'$ we removed with any valid probability distribution.

For the quantum rate-distortion function, let us assume $\{ v_i \}$ for $i=1,\ldots,n'$ are the eigenvectors corresponding to non-zero eigenvalues of the (possibly rank-deficient) matrix $\rho_A$. Similar to the argument for basis independence, we can define a new $n'$-dimensional system $A'$ with orthonormal basis $\{ u_i \}$ and an isometry $V=\sum_{i=1}^{n'} v_iu_i^\dag$. We use this to define a new problem with a full-rank input state $V^\dag\rho_AV$, purification $(V\otimes V)^\dag\rho_{AR}(V\otimes V)$ and distortion matrix $(\mathbb{I}\otimes V)^\dag \Delta (\mathbb{I}\otimes V)$. For any feasible channel $\mathcal{N}$ of the original problem~\eqref{eqn:pre-qrd}, we can construct feasible channel $\mathcal{N}'$ for the reduced problem with the same objective value by letting
\begin{equation*}
    C_{\mathcal{N}'} = (\mathbb{I}\otimes V)^\dag C_{\mathcal{N}} (\mathbb{I}\otimes V).
\end{equation*}
In the reverse direction, we have
\begin{equation*}
    C_{\mathcal{N}} = (\mathbb{I}\otimes V) C_{\mathcal{N}'} (\mathbb{I}\otimes V)^\dag + \sum_{i=n'+1}^n \rho_i \otimes v_iv_i^\dag,
\end{equation*}
where $\rho_i\in\mathcal{D}_m$ is any valid density matrix for all $i=n'+1,\ldots,n$.

\subsection{Optimizing over channel and joint output state} \label{appdx:rd-equiv-state}
Now given non-degeneracy of $p$, it is fairly straightforward to show that we can map between feasible channels $Q$ and states $P$ which obtain the same objective values between~\eqref{eqn:pre-crd} and~\eqref{eqn:new-crd} using 
\begin{equation*}
    P_{ij} = p_jQ_{ij} \quad \textrm{and} \quad Q_{ij} = P_{ij} / p_j,
\end{equation*}
for all $i=1,\ldots,m$, $j=1,\ldots,n$. Note that this is well defined as $p_j>0$ for all $j=1,\ldots,n$. Similarly for the quantum setting, given that $\rho_A$ has spectral decomposition $\rho_A=\sum_{i=1}\lambda_i v_iv_i^\dag$, without loss of generality we will first represent $\sigma_{BR}$ and $\mathcal{N}$ as
\begin{equation*}
    \sigma_{BR} = \sum_{ij}^n X_{ij} \otimes v_iv_j^\dag  \quad \textrm{and} \quad C_{\mathcal{N}} = \sum_{ij}^n Y_{ij} \otimes v_iv_j^\dag,
\end{equation*}
for suitable $X_{ij},Y_{ij} \in \mathbb{H}^m$. Given this, we can map from feasible channels $\mathcal{N}$ and states $\sigma_{BR}$ which obtain the same objective values between~\eqref{eqn:pre-qrd} and~\eqref{eqn:new-qrd} using 
\begin{equation*}
    \sigma_{BR} = \sum_{ij}^n \sqrt{\lambda_i\lambda_j}\,Y_{ij} \otimes v_iv_j^\dag  \quad \textrm{and} \quad C_{\mathcal{N}} = \sum_{ij}^n \frac{1}{\sqrt{\lambda_i\lambda_j}} X_{ij} \otimes v_iv_j^\dag.
\end{equation*}
Again, this is well-defined as $\lambda_i>0$ for all $i=1,\ldots,n$ from the assumption that $\rho_A$ is full-rank. Note that the first relationship is equivalent to $\sigma_{BR}=(\mathcal{N}\otimes\mathbb{I})(\rho_{AR})$. 
\begin{equation*}
    C_{\mathcal{N}'} = (\mathbb{I}\otimes U_A)^\dag C_{\mathcal{N}} (\mathbb{I}\otimes U_A).
\end{equation*}

\section{Linear convergence rates}\label{appdx:linear}

In Figure~\ref{fig:rate_distortions}, we show the empirical convergence behaviors of exact mirror descent on the classical and quantum rate-distortion problems for the Hamming distortion $\delta = \bm{1}\bm{1}^\top - \mathbb{I}$ and entanglement fidelity distortion $\Delta = \mathbb{I} - \rho_{AR}$ respectively, for various values of $\kappa$ and randomly generated inputs. We observe linear convergence for all examples, and note that this linear rate appears to become faster as $\kappa$ increases. These observations were also shared across multiple randomly generated input states.

From Fact~\ref{fact:conv}, this suggests that classical and quantum mutual information might be strongly convex relative to Shannon and von Neumann entropy respectively. However, simple counterexamples show that this is false. For a continuously differentiable function $f$ to be strongly convex relative to a Legendre function $\varphi$ on a set $\mathcal{X}$, there must be a $\mu>0$ that satisfies
\begin{equation*}
    \inp{\nabla f(x) - \nabla f(y)}{x - y} \geq \mu (D_\varphi\divx{x}{y} + D_\varphi\divx{y}{x}),
\end{equation*}
for all $x, y \in \relinterior \mathcal{X}$ (see e.g.,~\cite[Proposition 1.1(b)]{lu2018relatively}). For classical mutual information with input distribution $p\in\mathcal{P}_n$, we can show that for $X,Y\in\mathcal{P}_{m\times n}$ with columns $X_j = p_jx$ and $Y_j = p_jy$ for any $x,y\in\mathcal{P}_m$ for all $j=1,\ldots,n$ that
\begin{align*}
    \inp{\nabla I_c(X; p) - \nabla I_c(Y; p)}{X - Y} = 0,
\end{align*}
Similarly, for the quantum case with input state $\rho_A\in\mathcal{D}_n$, we can show that for $\sigma_{BR}=\sigma_B\otimes\rho_A$ and $\omega_{BR}=\omega_B\otimes\rho_A$ for any $\sigma_B,\omega_B \in \mathcal{D}_m$, that
\begin{align*}
    \inp{\nabla I_q(\sigma_{BR}; \rho_A) - \nabla I_q(\omega_{BR}; \rho_A)}{\sigma_{BR} - \omega_{BR}} = 0,
\end{align*}
by using additivity of quantum relative entropy~\cite[Property 11.2.1]{wilde2017quantum}. Therefore, we conclude that classical and quantum mutual information are not relatively strongly convex on their respective domains. Instead, it is possible to gain insights into this linear convergence behavior through the \emph{local} convexity properties of mutual information at the solutions to the rate-distoriton problems.

\begin{figure}[t]
\centering
\input{figures/rate_distortions.tex}
\caption{Convergence behaviors of computing the classical and quantum rate-distortion functions using mirror descent for Hamming distortion and entanglement fidelity distortion and for randomly generated inputs of dimension $n=32$.}
\label{fig:rate_distortions}
\end{figure}
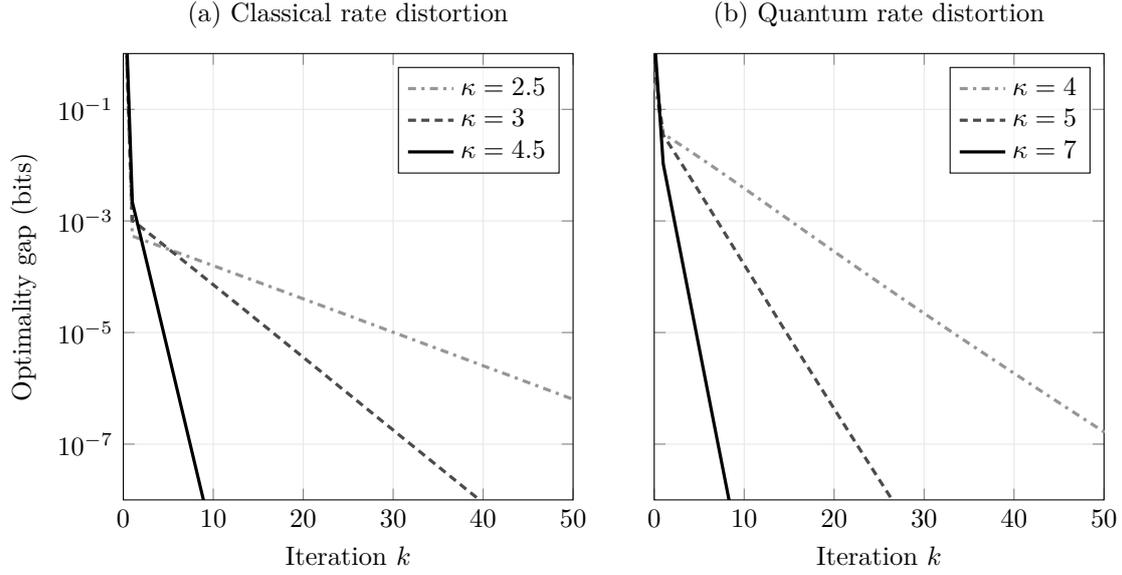

\subsection{Preliminaries}
We will first introduce the notion of local relative strong convexity.
\begin{defn}
    Let $f$ and $\varphi$ be twice continuously differentiable at some $x$, and let $\varphi$ be Legendre. Then $f$ is locally $\mu$-strongly convex relative to $\varphi$ at $x$ if there exists a $\mu>0$ such that
    \begin{equation*}
        \nabla^2 f(x) - \mu\nabla^2\varphi(x) \succeq 0.
    \end{equation*}
\end{defn}
We also introduce the following additional assumptions on the kernel function $\varphi$. The first two are standard assumptions used to prove convergence of the mirror descent iterates to the solution, rather than just convergence of the function values.
\begin{assump}\label{assump:varphi}
    The kernel function $\varphi: \mathbb{V}\rightarrow\bar{\mathbb{R}}$ satisfies the following.
    \begin{enumerate}[label=(\roman*), ref=\ref{assump:varphi}(\roman*)]
        \item The right partial level sets $\{ y\in\interior\domain\varphi : D_\varphi\divx{x}{y} \leq \alpha \}$ are bounded for all $x\in\domain\varphi$ and $\alpha\in\mathbb{R}$. \label{assump:varphi-i}
        \item If a sequence $\{ y^k \} \in \interior\domain\varphi$ converges to $x\in\domain\varphi$, then $D_\varphi\divx{x}{y^k}\rightarrow0$. \label{assump:varphi-ii}
        \item The function $\varphi$ is strongly convex. \label{assump:varphi-iii}
    \end{enumerate}
\end{assump}
We now present some results for local linear convergence.
\begin{prop}\label{prop:local-lin}
    Consider problem~\eqref{eqn:constr-min} where $x^*$ is any corresponding optimal solution, $f$ and $\varphi$ are twice continuously differentiable at $x^*$, and $\varphi$ is Legendre and satisfies Assumption~\ref{assump:varphi}. If $f$ is locally $\mu$-strongly convex relative to $\varphi$ at $x^*$, then the following results hold.
    \begin{enumerate}[label=(\roman*), ref=\ref{lem:local-lin}(\roman*)]
        \item The optimal solution $x^*$ is the unique optimal solution.
        \item Let $\{ x^k \}$ be the sequence of iterates generated by using the mirror descent iterates~\eqref{eqn:mirror-descent} to solve~\eqref{eqn:constr-min}. If $f$ is $L$-smooth relative to $\varphi$ and $t_k=1/L$ for all $k$, then there exists a $\mu' \in (0, \mu)$ and $N\in\mathbb{N}$ such that the sequence $\{ x^k \}$ satisfies
        \begin{equation}\label{eqn:md-linear}
            D_\varphi\divx{x^*}{x^k} \leq \biggl(1 - \frac{\mu'}{L}\biggr)^k D_\varphi\divx{x^*}{x^N},
        \end{equation}
        for all $k>N$.\label{lem:local-lin-b}
    \end{enumerate}
\end{prop}
\begin{proof}
    As $\mu>0$ and $\varphi$ satisfies Assumption~\ref{assump:varphi-iii}, meaning $\nabla^2 \varphi(x) \succ 0$ for all $x \in \interior \domain \varphi$, then there exists a $\mu' \in (0, \mu)$ such that 
    \begin{equation*}
        \nabla^2 f(x^*) - \mu'\nabla^2\varphi(x^*) \succ 0.
    \end{equation*}    
    As $f$ and $\varphi$ are both twice continuously differentiable, then there exists an $\varepsilon>0$ such that 
    \begin{equation*}
        \nabla^2 f(x) - \mu'\nabla^2\varphi(x) \succeq 0,
    \end{equation*}
    for all $\norm{x - x^*} < \varepsilon$. For part (i), it follows from Assumption~\ref{assump:varphi-iii} that $f$ is strictly convex within this neighborhood, in which $x^*$ must therefore be the unique solution. As the set of solutions to~\eqref{eqn:constr-min} must be convex, $x^*$ must be the unique solution on the entire feasible set.
    
    For part (ii), we first note that given $L$-relative smoothness of $f$ and Assumption~\ref{assump:varphi} that the mirror descent iterates $\{ x^k \}$ converge to the unique solution $x^*$ (this follows from a similar argument as, e.g.,~\cite[Theorem 2(ii)]{bauschke2017descent}). Therefore there exists an $N\in\mathbb{N}$ such that $\norm{x^k - x^*} < \varepsilon$  for all $k>N$. Linear convergence within this region follows from Fact~\ref{fact:conv}.
\end{proof}

\begin{rem}\label{rem:varphi-prop}
    The kernel functions $\varphi=-H$ and $\varphi=-S$ satisfy Assumption~\ref{assump:varphi-i} and Assumption~\ref{assump:varphi-ii}. Additionally, they satisfy Assumption~\ref{assump:varphi-iii} on any bounded set (notably, they are $1$-strongly convex with respect to the $1$- and trace-norm over $\mathcal{P}_n$ and $\mathcal{D}_n$ respectively~\cite[Proposition 3.12]{he2023mirror}).
\end{rem}


\subsection{Analysis of rate-distortion problems}
If we can characterize the local relative strong convexity parameters of mutual information at the solutions to the rate-distortion problems, then Proposition~\ref{prop:local-lin} allows us to predict the convergence behavior of mirror descent around a neighborhood of the solution. Knowing these local relative strong convexity parameters a priori can be difficult for arbitrary parameterizations of the rate-distortion problems. However, in Section~\ref{sec:uniform}, we derived explicit expressions for the optimal solutions of the classical rate-distortion problem for uniform input distribution and Hamming distortion (Theorem~\ref{thm:crd-maxmix-solution}), and the quantum rate-distortion problem for maximally mixed input and entanglement fidelity distortion (Theorem~\ref{thm:qrd-maxmix-solution}). Therefore, we can compute the local relative strong convexity parameters of mutual information at these solutions by solving what amounts to a generalized eigenvalue problem. If we consider~\eqref{eqn:constr-min} where $\mathbb{V}=\mathbb{R}^n$, this amounts to finding the smallest $\mu$ that satisfies
\begin{equation*}
    (U^\dag \nabla^2 f(x^*) U) v = \mu (U^\dag \nabla^2\varphi(x^*) U) v,
\end{equation*}
for some $v\in\mathbb{R}^m\setminus\{0\}$, where $m=\dim_{\mathbb{R}}(\ker\mathcal{A})$ and $U\in\mathbb{R}^{n\times m}$ is an isometry with columns forming an orthonormal basis of $\ker\mathcal{A}$ such that $UU^\dag$ is a projection onto this nullspace. 

Figure~\ref{fig:qrd_maxmix} shows these values of the local relative strong convexity parameter $\mu$ for various values of distortion dual variables $\kappa$ and problem dimensions $n$ for both of these specific classical and quantum rate-distortion problems. Notably, we see that for both the classical and quantum scenarios that $\mu$ increases monotonically from $0$ to $1$ as $\kappa$ increases from $0$ to infinity. Combining these observations with Lemma~\ref{lem:qrd-smooth} and Proposition~\ref{prop:local-lin}, we conclude that mirror descent applied to solve these particular instances of rate-distortion problems will converge linearly around a neighborhood around the solution. Moreover, this linear rate becomes arbitrarily fast as $\kappa$ tends to infinity. 

As seen in Figure~\ref{fig:rate_distortions}, these observations reflect our empirical findings for the classical rate-distortion problem with Hamming distortion and quantum rate-distortion problem with entanglement fidelity distortion, not only for uniform input distributions and maximally mixed input states, but for \emph{any} input state. We remark that it is possible to rigorously prove that the $\lim_{\kappa\rightarrow\infty}\mu(\kappa)=1$ property seen in Figure~\ref{fig:qrd_maxmix} holds for any input state. 

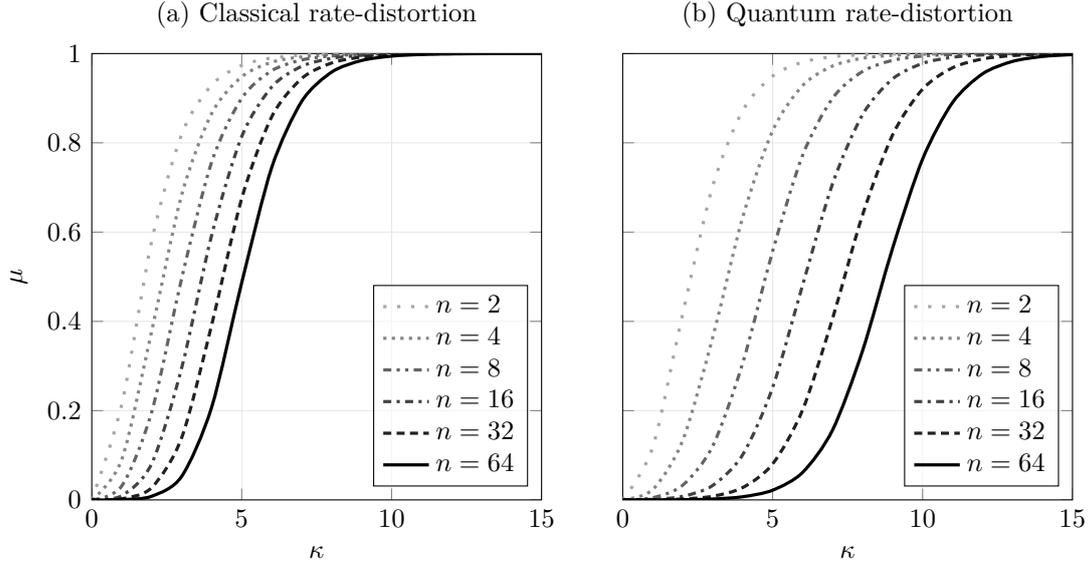
\begin{figure}[t]
\centering
\input{figures/qrd_maxmix.tex}
\vspace{-0.5em}
\caption{Local relative strong convexity parameters $\mu$ at the optimal point of the (a) classical rate-distortion problem with uniform input distribution and Hamming distortion and (b) quantum rate-distortion problem with maximally mixed input state and entanglement fidelity distortion, at different problem dimensions $n$ and distortion dual variables $\kappa$.}
\label{fig:qrd_maxmix}
\end{figure}


\section{Frank-Wolfe lower bound} \label{appdx:lb}

\subsection{Preliminaries}
We will first remind the reader about Frank-Wolfe type lower bounds for convex optimization problems over compact sets, then show how these can be applied to the classical and quantum rate-distortion problems. Consider~\eqref{eqn:constr-min} where $\mathcal{C}$ is a compact set. As $f$ is convex, any linear approximation will always lower bound the function. That is, for any $x, y \in \mathcal{C}$
\begin{equation}
   f(x) \geq f(y) + \langle \nabla f(y), x - y \rangle.
\end{equation}
By finding the infimum of the right hand side over $y$ over $\mathcal{C}$, we can obtain a value which lower bounds the function. As $\mathcal{C}$ is compact, this lower bound will always exist, and therefore we define
\begin{align}
   f_{lb}(x) &\coloneqq \min_{y\in\mathcal{C}} \bigl\{ f(x) + \langle \nabla f(x), y - x \rangle \bigr\} \nonumber\\
   &= f(x) - \inp{\nabla f(x)}{x} + \min_{y\in\mathcal{C}} \bigl\{ \inp{\nabla f(x)}{y} \bigr\}. \label{eqn:lowerbound}
\end{align}
We can use this to define the estimated optimality gap $e(x)$
\begin{equation}
    e(x) \coloneqq f(x)-f_{lb}(x)=\max_{y\in\mathcal{C}} \bigl\{ \langle \nabla f(x), x - y \rangle \bigr\},\label{eqn:opt-gap}
\end{equation}
which upper bounds the true optimality gap, i.e., $e(x) \geq f(x)-f^*$. We also show that this error bound converges to zero as $x$ converges to an optimal point.

\begin{prop}
    Consider the constrained convex optimization problem~\eqref{eqn:constr-min} where $\mathcal{C}$ is a compact set. At an optimal point $x^*$, the estimated optimality gap equals zero $e(x^*)=0$. Moreover, if $f$ is continuously differentiable, for any sequence $\{ x^k \}$ that converges to $x^*$, the sequence $\{ e(x^k) \}$ will converge to zero.
\end{prop}
\begin{proof}
    The first order optimality condition for a constrained minimization problem~\eqref{eqn:constr-min} is
    \begin{equation}
       -\nabla f(x^*) \in N_\mathcal{C}(x^*), \label{eqn:opt-constr-cond}
    \end{equation}    
    where $N_\mathcal{C}$ is the normal cone of $\mathcal{C}$ (see e.g.~\cite{bertsekas2009convex}). It follows from the definition of the normal cone and the optimality condition~\eqref{eqn:opt-constr-cond} that for all $x\in\mathcal{C}$
    \begin{equation}
       \inp{\nabla f(x^*)}{x^* - x} \leq 0.
    \end{equation}
    Noting that the optimality gap $e(x^*)$ maximizes the left hand side over $x\in\mathcal{C}$ and is non-negative shows that $e(x^*)=0$ as desired.

    To show convergence of the optimality gap to zero, we recognize that $e(x)$ is the sum of a continuous function and the support function of a compact set. As support functions on bounded sets are continuous~\cite[Proposition 8.6]{arutyunov2016convex}, from continuous differentiability of $f$ and the fact that compositions of continuous functions are continuous, we conclude that $e(x)$ is a continuous function, from which the desired result follows.
\end{proof}

Over certain sets, there are well-known ways to compute the minimization required in~\eqref{eqn:lowerbound}, see, e.g.,~\cite{jaggi2013revisiting}. 

\subsection{Application to rate-distortion problems}
We now show how the Frank-Wolfe lower bound can be applied to both the classical and quantum rate-distortion problems.

\begin{rem}
    Using the mirror descent interpretation of Blahut-Arimoto algorithms~\cite{he2023mirror}, the lower bounds derived for Blahut-Arimoto algorithms applied to classical~\cite{blahut1972computation} and quantum~\cite{li2019blahut,ramakrishnan2020computing} channel capacities are identical to the Frank-Wolfe lower bound~\eqref{eqn:opt-gap}.
\end{rem}

\paragraph{Classical rate-distortion}
For the classical rate-distortion problem~\eqref{eqn:crd}, let us consider the set $\mathcal{P}' = \{ P\in\mathcal{P}_{m\times n} : \sum_{i=1}^m P_{ij} = p_j, j=1,\ldots,n\}$ where $p\in\mathcal{P}_n$. Over this set, we find
\begin{equation}
    \min_{y\in\mathcal{P}'} \{ \inp{\nabla f(x)}{y} \} = \sum_{j=1}^n p_j \min_i \{ \partial_{ij} f(x) \}. \label{eqn:col-stoc-ub}
\end{equation}
where $\partial_{ij} f$ is the partial derivative of $f$ with respect to the $(i,j)$-th coordinate, by noting the minimization is separable into minimizing each column of $P$ over a scaled probability simplex. 

\paragraph{Quantum rate-distortion}
For the quantum rate-distortion problem~\eqref{eqn:qrd}, consider the set $\mathcal{D}' = \{ \sigma_{BR} \in\mathcal{D}_{mn} : \tr_B(\sigma_{BR}) = \sigma \}$ where $\sigma\in\mathcal{D}_n$ with spectral decomposition $\sigma=\sum_{i=1}^n\lambda_iv_iv_i^\dag$. Unfortunately, it is not obvious how to compute the lower bound~\eqref{eqn:lowerbound} efficiently over this set. However, if the gradient is of the form
\begin{equation}\label{eqn:diag-grad}
    \nabla f(x) = \sum_{i=1}^m \sum_{j=1}^n g_{ij} b_ib_i^\dag \otimes v_jv_j^\dag,
\end{equation}
where $g_{ij}\in\mathbb{R}$ for all $i,j$ and $\{ b_i \}$ is any orthonormal bases for $B$, then a symmetry reduction argument shows that the left-hand side of~\eqref{eqn:simp-bound} is invariant under the same representation as in Example~\ref{exmp:cc}, and therefore the optimal value is an element of the subspace~\eqref{eqn:cc-subspace}. Therefore, the minimization is separable in the same way as~\eqref{eqn:col-stoc-ub}, and
\begin{equation}\label{eqn:simp-bound}
    \min_{y\in\mathcal{D}'} \{ \inp{\nabla f(x)}{y} \} = \sum_{j=1}^n \lambda_j \min_i \{ g_{ij} \}.
\end{equation}
Although~\eqref{eqn:diag-grad} this seems like a very limiting assumption, we can show this holds for the quantum rate-distortion problem setting we are most interested in, i.e., where distortion is measured using the entanglement fidelity, and if iterates are generated using Algorithm~\ref{alg:qrd}. To see this, consider~\eqref{eqn:qrd} parameterized by input state $\rho_A\in\mathcal{D}_n$ and distortion matrix $\Delta=\mathbb{I} - \rho_{AR}$. By substituting~\eqref{eqn:qrd-iter} into the gradient of quantum mutual information~\cite[Equation (61)]{he2023mirror}, we can express the gradient of the objective function as
\begin{equation}\label{eqn:qmi-grad-diag}
    \nabla f(\sigma_{BR}^k) = (\log(\tr_R(\sigma_{BR}^{k-1}))  - \log(\tr_R(\sigma_{BR}^k))) \otimes \mathbb{I} - \mathbb{I}\otimes \nu^k.
\end{equation}
From Theorem~\ref{cor:sym-space} and Corollary~\ref{cor:subspace-ptrace}, we know $\tr_R(\sigma_{BR}^k)=\sum_{i=1}^n \alpha_i v_iv_i^\dag$ and $\nu^k=\sum_{j=1}^n \beta_j v_jv_j^\dag$ for some $\alpha,\beta\in\mathbb{R}^n$ for all $k$, and therefore the gradient is of the desired form~\eqref{eqn:diag-grad}.

This property of the gradients no longer holds in general if we use the inexact method in Algorithm~\ref{alg:inexact-qrd}. However, we can still use this lower bounding technique whenever we solve a step exactly, which we can choose to do at the last iteration before terminating the algorithm to obtain an optimality gap for the algorithm's final output. 

\section{Background proofs}

\subsection{Proof of Proposition~\ref{prop:dual-prop}}\label{appdx:proof-dual}
We first present a simple result about how strict convexity and coercivity are preserved by injective linear maps.
\begin{lem}\label{lem:injective}
    Consider a function $f:\mathbb{V}\rightarrow\mathbb{R}$ and an affine operator   $\mathcal{A}:\mathbb{V}\rightarrow\mathbb{V}'$.
    \begin{enumerate}[label=(\roman*), ref=\ref{lem:injective}(\roman*)]
        \item If $f$ is strictly convex, then $h(x) = f(\mathcal{A}^\dag(x))$ is strictly convex if and only if $\mathcal{A}^\dag$ is injective.\label{lem:injective-a}
        \item If $f$ is coercive, then $h(x) =  f(\mathcal{A}^\dag(x))$ is coercive if and only if $\mathcal{A}^\dag$ is injective. \label{lem:injective-b}
    \end{enumerate}
\end{lem}
\begin{proof}
    For part (i), consider any $x,y\in\mathbb{V}'$ where $x\neq y$. If $\mathcal{A}^\dag$ is injective, then $\mathcal{A}^\dag(x)\neq\mathcal{A}^\dag(y)$, and therefore from strict convexity of $f$ we have for $\lambda\in[0, 1]$
    \begin{align}\label{eqn:strict-proof}
        f(\mathcal{A}^\dag(\lambda x + (1-\lambda)y)) = f(\lambda \mathcal{A}^\dag(x) + (1-\lambda)\mathcal{A}^\dag(y)) < \lambda f(\mathcal{A}^\dag(x)) + (1-\lambda)f(\mathcal{A}^\dag(y)),
    \end{align}
    which proves strict convexity of $h$. If $\mathcal{A}^\dag$ is not injective, then there exists $x,y\in\mathbb{V}$ where $x\neq y$ such that $\mathcal{A}^\dag(x)=\mathcal{A}^\dag(y)$. For this example, the inequality~\eqref{eqn:strict-proof} must hold with equality throughout, showing that $h$ is not strictly convex.

    For part (ii), consider a sequence $x^k$ satisfying $\norm{x^k}\rightarrow\infty$. Let $a=\min\{ \norm{\mathcal{A}^\dag(x)} : \norm{x}=1 \}$. Clearly $a>0$ if and only if $\mathcal{A}^\dag$ is injective (otherwise $\mathcal{A}^\dag$ has nontrivial kernel), in which case $a\norm{x^k}\leq\norm{\mathcal{A}^\dag(x^k)}$ for all $k\in\mathbb{N}$, and so $\norm{\mathcal{A}^\dag(x^k)}\rightarrow\infty$. Therefore, it follows from coercivity of $f$ that $h$ is also coercive. If $\mathcal{A}^\dag$ is not injective, then there exists a direction $d\in\ker\mathcal{A}^\dag$ that we can use to construct the sequence $x^k = x^0 + kd$. This satisfies $\norm{x^k}\rightarrow\infty$, but $h(x^k) = h(x^0)$ for all $k\in\mathbb{N}$, and so $h$ is not coercive.
\end{proof}

We now present the main proof.
\begin{proof}[Proof of Proposition~\ref{prop:dual-prop}]
    Part (i) follows from Fact~\ref{fact:legendre-i} and Lemma~\ref{lem:injective-a}. For part (ii), we first use~\cite[Fact 2.11]{bauschke1997legendre} to show that because $\varphi^*$ is closed, convex and proper, that the function $x\mapsto \varphi^*(x)-\inp{s}{x}$ is coercive if and only if $s\in\interior\domain\varphi$. Using Lemma~\ref{lem:injective-b}, it follows that
    \begin{equation*}
        \nu \mapsto \frac{1}{t} \varphi^*[ \nabla\varphi(y) - t(\nabla f(y) + \mathcal{A}^\dag(\nu)) ] - \biggl\langle \frac{s}{t}, \nabla\varphi(y) - t(\nabla f(y) + \mathcal{A}^\dag(\nu))  \biggr\rangle,
    \end{equation*}
    is coercive if and only if $s\in\interior\domain\varphi$ and $\mathcal{A}^\dag$ is injective. Letting $b=\mathcal{A}(s)$ and discarding constant terms gives us the negative dual function $-g(\wc; y)$, which concludes the proof.
\end{proof}

\subsection{Proof of Proposition~\ref{prop:dual-ii}} \label{appdx:dual}

We will first show that the dual function has a negative definite Hessian.
\begin{prop}\label{prop:dual-neg-def}
    The dual function $g^k_q$, as defined in~\eqref{eqn:dual}, has negative definite Hessian. 
\end{prop}
\begin{proof}
    From~\eqref{eqn:hess-dual} and using the relationship $\mathsf{D}^2 g_q^k(\nu)[V,V] = \inp{V}{\mathsf{D}\nabla g_q^k(\nu)[V]}$, we have
    \begin{equation*}
        \mathsf{D}^2 g_q^k(\nu)[V,V] = -\sum_{ij}^{mn} f^{[1]}(\Lambda)_{ij} (U^\dag (\mathbb{I}_B\otimes V) U)_{ij}^2,
    \end{equation*}
    where $U,\Lambda$ are given by~\eqref{eqn:diag-dual-hess} and $f^{[1]}(\Lambda)$ is the first divided differences matrix, as defined in~\eqref{eqn:first-divided-diff}, of $\Lambda$ corresponding to $f(x)=e^x$. As the exponential function is strictly positive and monotone, all elements of the divided differences matrix $f^{[1]}(\Lambda)$ are strictly positive. Moreover, as the tensor product is injective, $U^\dag (\mathbb{I}\otimes \nu) U = 0$ if and only if $\nu=0$. Therefore, we conclude that
    \begin{equation*}
        \mathsf{D}^2 g_q^k(\nu)[V,V] < 0,
    \end{equation*}
    for all $\nu$ and non-zero $V$, which proves the desired result.
\end{proof}

We will also briefly remind the reader about these fairly standard results which establish Lipschitz-continuity and strong convexity over compact sets.
\begin{lem}\label{lem:compact-Lipschitz}
    Consider a continuously differentiable function $f:\mathcal{U} \rightarrow \mathbb{V}$ where $\mathcal{U}$ is an open subset of an inner product space $\mathbb{V}'$, and a convex compact set $\mathcal{X} \subset \mathcal{U}$. Then $f$ is Lipschitz on $\mathcal{X}$, meaning that there exists an $L\geq0$ such that $\lVert f(x) - f(y) \rVert \leq L \lVert x - y \rVert$ for all $x,y \in \mathcal{X}$. 
\end{lem}
\begin{proof}
    This is a straightfoward consequence of the mean value theorem~\cite[Theorem X.4.5]{bhatia2013matrix} and convexity of $\mathcal{X}$.
\end{proof}

\begin{lem}\label{lem:compact-strong}
    Consider a twice continuously differentiable function $f:\mathcal{U} \rightarrow \mathbb{R}$ where $\mathcal{U}$ is an open subset of an inner product space $\mathbb{V}$, and a convex compact set $\mathcal{X} \subset \mathcal{U}$. If $f$ has a positive definite Hessian, then $f$ is strongly convex on $\mathcal{X}$, meaning that there exists a $\mu>0$ such that $\mathsf{D}^2f(x)[v,v]\geq\mu$ for all $x\in\mathcal{X}$ and $v\in\mathbb{V}$ satisfying $\norm{v}=1$.
\end{lem}
\begin{proof}
    Given twice continuous differentiability of $f$, we know $\mathsf{D}^2f(x)[v,v]$ is a continuous function in both $x$ and $v$. The desired result then follows by taking the infumum of $\mathsf{D}^2f(x)[v,v]$ over the product of the compact sets $\mathcal{X}$ and $\{v\in\mathbb{V}: \norm{v}=1\}$.
\end{proof}

We now present the main proof. 

\begin{proof}[Proof of Proposition~\ref{prop:dual-ii}]
    Given that gradient descent and Newton's method with backtracking both result in a monotonic sequence of function values, for some initialization $\nu_0\in\mathbb{H}^n$, we can constrain our analysis to the superlevel set $\mathcal{S}=\{ \nu\in\mathbb{H}^n : g_q^k(\nu) \geq g_q^k(\nu_0) \}$ which is compact from coercivity of $-g_q^k(\nu)$. Over this compact set, from Lemmas~\ref{lem:compact-Lipschitz} and~\ref{lem:compact-strong} we establish that $g_q^k(\nu)$ is strongly concave, and has Lipschitz gradient and Hessian. These properties are sufficient to establish global convergence of the function values to the optimal value by using gradient descent and Newton's method with backtracking~\cite[Sections 9.3.1 and 9.5.3]{boyd2004convex}. Convergence of the iterates to the unique solution follows from strong convexity and Lipshitz gradient (see e.g., the discussion in~\cite[Section 9.1.2]{boyd2004convex}).
\end{proof}

\urlstyle{rm}
\printbibliography

\end{document}

%% file: figures/rate_distortions.tex
\pgfplotstableread[col sep=comma]{figures/data/rate_distortions.csv}\data
\begin{tikzpicture}
    \definecolor{clr1}{RGB}{0,0,0}
    \definecolor{clr2}{RGB}{75,75,75}
    \definecolor{clr3}{RGB}{150,150,150}
    
    \begin{semilogyaxis}[
        name=mainplot,
        width = 0.5\columnwidth,
        height = 0.5\columnwidth,
        xmin = 0.0, xmax = 50,
        ymin = 1e-8, ymax=1,
        ylabel=Optimality gap (bits),
        xlabel=Iteration $k$,
        label style = {font=\small},
        ticklabel style = {font=\small},
        title style = {font=\small},
        legend pos  = north east,
        legend cell align={left},
        legend style={fill=white, fill opacity=0.6, draw opacity=1,text opacity=1, font=\small},
        no markers,
        every axis plot/.append style={very thick},
        grid=major,
        major grid style={line width=.1pt,draw=gray!20},
        yminorticks=false,
        title = {(a) Classical rate distortion},
        ]
        
        \addplot[clr3, dashdotted] table[x=i , y=c25] {\data}; 
        \addplot[clr2, densely dashed] table[x=i , y=c3] {\data};
        \addplot[clr1] table[x=i , y=c45] {\data};

        \legend{
            $\kappa=2.5$,
            $\kappa=3$,
            $\kappa=4.5$
        }

    \end{semilogyaxis}
\end{tikzpicture}\hspace{0.5cm}\begin{tikzpicture}
    \definecolor{clr1}{RGB}{0,0,0}
    \definecolor{clr2}{RGB}{75,75,75}
    \definecolor{clr3}{RGB}{150,150,150}

    \begin{semilogyaxis}[
        name=secondplot,
        at={(mainplot.north east)},
        yshift=0.05cm,
        anchor=north west,   
        width = 0.5\columnwidth,
        height = 0.5\columnwidth,
        xmin = 0.0, xmax = 50,
        ymin = 1e-8, ymax=1,
        xlabel=Iteration $k$,
        yticklabels={,,},
        label style = {font=\small},
        ticklabel style = {font=\small},
        title style = {font=\small},
        legend pos  = north east,
        legend cell align={left},
        legend style={fill=white, fill opacity=0.6, draw opacity=1,text opacity=1, font=\small},
        no markers,
        every axis plot/.append style={very thick},
        grid=major,
        major grid style={line width=.1pt,draw=gray!20},
        yminorticks=false,
        title = {(b) Quantum rate distortion}
        ]
        
        \addplot[clr3, dashdotted] table[x=i , y=q4e] {\data}; 
        \addplot[clr2, densely dashed] table[x=i , y=q5e] {\data};
        \addplot[clr1] table[x=i , y=q7e] {\data};


        \legend{
            $\kappa=4$,
            $\kappa=5$,
            $\kappa=7$
        }
        
    \end{semilogyaxis}
\end{tikzpicture}

%% file: figures/qrd_maxmix.tex
\pgfplotstableread[col sep=comma]{figures/data/qrd_maxmix.csv}\data
\begin{tikzpicture}
    \definecolor{clr1}{RGB}{0,0,0}
    \definecolor{clr2}{RGB}{33,33,33}
    \definecolor{clr3}{RGB}{66,66,66}
    \definecolor{clr4}{RGB}{100,100,100}
    \definecolor{clr5}{RGB}{133,133,133}
    \definecolor{clr6}{RGB}{166,166,166}
    
    \begin{axis}[
        width = 0.5\columnwidth,
        height = 0.5\columnwidth,
        xmin = 0.0, xmax = 15,
        ymin = 0.0, ymax=1,
        xlabel=$\kappa$,
        ylabel={$\mu$},
        label style = {font=\small},
        ticklabel style = {font=\small},
        legend pos  = south east,
        legend cell align={left},
        legend style={fill=white, fill opacity=0.6, draw opacity=1,text opacity=1, font=\small},
        no markers,
        every axis plot/.append style={very thick},
        grid=major,
        major grid style={line width=.1pt,draw=gray!20},
        yminorticks=false,
        title={(a) Classical rate-distortion},
        title style = {font=\small},
        ]

        \addplot[smooth, clr6, loosely dotted] table[x=lambda , y=c2] {\data}; 
        \addplot[smooth, clr5, dotted] table[x=lambda , y=c4] {\data}; 
        \addplot[smooth, clr4, dashdotdotted] table[x=lambda , y=c8] {\data};
        \addplot[smooth, clr3, dashdotted] table[x=lambda , y=c16] {\data};
        \addplot[smooth, clr2, densely dashed] table[x=lambda , y=c32] {\data};
        \addplot[smooth, clr1] table[x=lambda , y=c64] {\data};
        
        \legend{
            $n = 2$,
            $n = 4$,
            $n = 8$,
            $n = 16$,
            $n = 32$,
            $n = 64$
        }
    \end{axis}
\end{tikzpicture}\hspace{0.5cm}\begin{tikzpicture}
    \definecolor{clr1}{RGB}{0,0,0}
    \definecolor{clr2}{RGB}{33,33,33}
    \definecolor{clr3}{RGB}{66,66,66}
    \definecolor{clr4}{RGB}{100,100,100}
    \definecolor{clr5}{RGB}{133,133,133}
    \definecolor{clr6}{RGB}{166,166,166}
    
    \begin{axis}[
        width = 0.5\columnwidth,
        height = 0.5\columnwidth,
        xmin = 0.0, xmax = 15,
        ymin = 0.0, ymax=1,
        xlabel=$\kappa$,
        yticklabels={,,},
        label style = {font=\small},
        ticklabel style = {font=\small},
        legend pos  = south east,
        legend cell align={left},
        legend style={fill=white, fill opacity=0.6, draw opacity=1,text opacity=1, font=\small},
        no markers,
        every axis plot/.append style={very thick},
        grid=major,
        major grid style={line width=.1pt,draw=gray!20},
        yminorticks=false,
        title={(b) Quantum rate-distortion},
        title style = {font=\small},
        ]

        \addplot[smooth, clr6, loosely dotted] table[x=lambda , y=q2] {\data}; 
        \addplot[smooth, clr5, dotted] table[x=lambda , y=q4] {\data}; 
        \addplot[smooth, clr4, dashdotdotted] table[x=lambda , y=q8] {\data};
        \addplot[smooth, clr3, dashdotted] table[x=lambda , y=q16] {\data};
        \addplot[smooth, clr2, densely dashed] table[x=lambda , y=q32] {\data};
        \addplot[smooth, clr1] table[x=lambda , y=q64] {\data};
        
        \legend{
            $n = 2$,
            $n = 4$,
            $n = 8$,
            $n = 16$,
            $n = 32$,
            $n = 64$
        }
    \end{axis}
\end{tikzpicture}

%% file: bibliofile.bib
@article{arimoto1972algorithm,
    author = {Arimoto, Suguru},
    doi = {10.1109/tit.1972.1054753},
    journal = {IEEE Transactions on Information Theory},
    month = {1},
    number = {1},
    pages = {14--20},
    publisher = {IEEE},
    title = {An algorithm for computing the capacity of arbitrary discrete memoryless channels},
    volume = {18},
    year = {1972}
}

@book{arutyunov2016convex,
    author = {Arutyunov, Aram V and Obukhovskii, Valeri},
    doi = {10.1515/9783110460308},
    publisher = {De Gruyter},
    title = {Convex and set-valued analysis},
    year = {2017}
}

@article{barnum2000quantum,
    author = {Barnum, Howard},
    doi = {10.1103/physreva.62.042309},
    journal = {Physical Review A},
    month = {9},
    number = {4},
    pages = {042309},
    publisher = {APS},
    title = {Quantum rate-distortion coding},
    volume = {62},
    year = {2000}
}

@article{bauschke1997legendre,
    author = {Bauschke, Heinz H and Borwein, Jonathan M and others},
    journal = {Journal of Convex Analysis},
    pages = {27--67},
    title = {Legendre functions and the method of random {B}regman projections},
    volume = {4},
    year = {1997},
}

@article{bauschke2017descent,
    author = {Bauschke, Heinz H and Bolte, J{\'e}r{\^o}me and Teboulle, Marc},
    journal = {Mathematics of Operations Research},
    month = {5},
    number = {2},
    pages = {330--348},
    publisher = {Informs},
    title = {A descent lemma beyond Lipschitz gradient continuity: First-order methods revisited and applications},
    volume = {42},
    year = {2017},
    doi = {10.1287/moor.2016.0817}
}

@article{beck2003mirror,
    author = {Beck, Amir and Teboulle, Marc},
    doi = {10.1016/s0167-6377(02)00231-6},
    journal = {Operations Research Letters},
    month = {5},
    number = {3},
    pages = {167--175},
    publisher = {Elsevier},
    title = {Mirror descent and nonlinear projected subgradient methods for convex optimization},
    volume = {31},
    year = {2003}
}

@book{beck2017first,
    author = {Beck, Amir},
    doi = {10.1137/1.9781611974997},
    publisher = {SIAM},
    title = {First-order methods in optimization},
    year = {2017}
}

@book{bertsekas2009convex,
    author = {Bertsekas, Dimitri},
    publisher = {Athena Scientific},
    title = {Convex optimization theory},
    volume = {1},
    year = {2009},
    isbn = {1-886529-31-0}
}

@book{bhatia2013matrix,
    author = {Bhatia, Rajendra},
    doi = {10.1007/978-1-4612-0653-8},
    publisher = {Springer Science \& Business Media},
    title = {Matrix analysis},
    volume = {169},
    year = {2013}
}

@article{blahut1972computation,
    author = {Blahut, Richard},
    doi = {10.1109/tit.1972.1054855},
    journal = {IEEE Transactions on Information Theory},
    month = {7},
    number = {4},
    pages = {460--473},
    publisher = {IEEE},
    title = {Computation of channel capacity and rate-distortion functions},
    volume = {18},
    year = {1972}
}

@book{boyd2004convex,
    author = {Boyd, Stephen and Vandenberghe, Lieven},
    doi = {10.1017/cbo9780511804441},
    publisher = {Cambridge University Press},
    title = {Convex optimization},
    year = {2004}
}

@book{bredon1972introduction,
    author = {Bredon, Glen E},
    booktitle = {Pure and Applied Mathematics},
    doi = {10.1016/s0079-8169(08)x6007-6},
    publisher = {Academic Press},
    title = {Introduction to compact transformation groups},
    year = {1972}
}

@article{bregman1967relaxation,
    author = {Bregman, Lev M},
    doi = {10.1016/0041-5553(67)90040-7},
    journal = {USSR Computational Mathematics and Mathematical Physics},
    number = {3},
    pages = {200--217},
    publisher = {Elsevier},
    title = {The relaxation method of finding the common point of convex sets and its application to the solution of problems in convex programming},
    volume = {7},
    year = {1967}
}

@book{brocker2013representations,
    author = {Br{\"o}cker, Theodor and Tom Dieck, Tammo},
    doi = {10.1007/978-3-662-12918-0},
    publisher = {Springer Science \& Business Media},
    title = {Representations of compact Lie groups},
    volume = {98},
    year = {2013}
}

@article{chandrasekaran2017relative,
    author = {Chandrasekaran, Venkat and Shah, Parikshit},
    doi = {10.1007/s10107-016-0998-2},
    journal = {Mathematical Programming},
    month = {1},
    number = {1},
    pages = {1--32},
    publisher = {Springer},
    title = {Relative entropy optimization and its applications},
    volume = {161},
    year = {2017}
}

@article{coey2023performance,
    author = {Coey, Chris and Kapelevich, Lea and Vielma, Juan Pablo},
    doi = {10.1007/s12532-022-00226-0},
    journal = {Mathematical Programming Computation},
    month = {3},
    number = {1},
    pages = {53--101},
    publisher = {Springer},
    title = {Performance enhancements for a generic conic interior point algorithm},
    volume = {15},
    year = {2023}
}

@book{cover1999elements,
    author = {Thomas M. Cover and Joy A. Thomas},
    doi = {10.1002/047174882X},
    publisher = {John Wiley \& Sons},
    title = {Elements of information theory},
    year = {1999}
}

@article{datta2012quantum,
    author={Datta, Nilanjana and Hsieh, Min-Hsiu and Wilde, Mark M.},
    journal={IEEE Transactions on Information Theory}, 
    title={Quantum rate distortion, reverse {S}hannon theorems, and source-channel separation}, 
    year={2013},
    volume={59},
    number={1},
    pages={615-630},
    doi={10.1109/tit.2012.2215575}
}

@article{datta2013quantum,
    author = {Datta, Nilanjana and Hsieh, Min-Hsiu and Wilde, Mark M and Winter, Andreas},
    doi = {10.1063/1.4798396},
    journal = {Journal of Mathematical Physics},
    month = {4},
    number = {4},
    publisher = {AIP Publishing},
    title = {Quantum-to-classical rate distortion coding},
    volume = {54},
    year = {2013}
}

@article{diaconis2001linear,
    author = {Diaconis, Persi and Evans, Steven},
    doi = {10.1090/S0002-9947-01-02800-8},
    journal = {Transactions of the American Mathematical Society},
    month = {3},
    number = {7},
    pages = {2615--2633},
    publisher = {American Mathematical Society (AMS)},
    title = {Linear functionals of eigenvalues of random matrices},
    volume = {353},
    year = {2001}
}

@article{fawzi2018efficient,
    author = {Fawzi, Hamza and Fawzi, Omar},
    doi = {10.1088/1751-8121/aab285},
    journal = {Journal of Physics A: Mathematical and Theoretical},
    month = {3},
    number = {15},
    pages = {154003},
    publisher = {IOP Publishing},
    title = {Efficient optimization of the quantum relative entropy},
    volume = {51},
    year = {2018}
}

@article{fawzi2019semidefinite,
    author = {Fawzi, Hamza and Saunderson, James and Parrilo, Pablo A},
    doi = {10.1007/s10208-018-9385-0},
    journal = {Foundations of Computational Mathematics},
    month = {4},
    number = {2},
    pages = {259--296},
    publisher = {Springer},
    title = {Semidefinite approximations of the matrix logarithm},
    volume = {19},
    year = {2019}
}

@book{fulton2013representation,
    author = {Fulton, William and Harris, Joe},
    publisher = {Springer Science \& Business Media},
    title = {Representation theory: A first course},
    volume = {129},
    year = {2013},
    doi = {10.1007/978-1-4612-0979-9}
}

@article{hayashi2022bregman,
    author = {Hayashi, Masahito},
    doi = {10.1109/tit.2023.3239955},
    journal = {IEEE Transactions on Information Theory},
    month = {6},
    number = {6},
    pages = {3460--3492},
    publisher = {IEEE},
    title = {{B}regman divergence based em algorithm and its application to classical and quantum rate distortion theory},
    volume = {69},
    year = {2023}
}

@article{hayashi2023efficient,
    author = {Hayashi, Masahito and Yang, Yuxiang},
    doi = {10.22331/q-2023-03-02-936},
    journal = {Quantum},
    month = {3},
    pages = {936},
    publisher = {Verein zur F{\"o}rderung des Open Access Publizierens in den Quantenwissenschaften},
    title = {Efficient algorithms for quantum information bottleneck},
    volume = {7},
    year = {2023}
}

@misc{hayashi2023generalized,
    author = {Hayashi, Masahito and Liu, Geng},
    title = {Generalized quantum Arimoto-Blahut algorithm and its application to quantum information bottleneck},
    year = {2023},
    archivePrefix = {arXiv},
    eprint = {2311.11188},
}

@misc{hayashi2023minimization,
    author = {Hayashi, Masahito},
    title = {Iterative minimization algorithm on mixture family},
    year = {2023},
    archivePrefix = {arXiv},
    eprint = {2302.06905},
}

@misc{he2023mirror,
    author = {He, Kerry and Saunderson, James and Fawzi, Hamza},
    title = {A {B}regman proximal perspective on classical and quantum {B}lahut-{A}rimoto algorithms},
    year = {2023},
    archivePrefix = {arXiv},
    eprint = {2306.04492},
}

@book{horn1994topics, 
    place={Cambridge}, 
    title={Topics in matrix analysis}, 
    publisher={Cambridge University Press}, 
    author={Horn, Roger A. and Johnson, Charles R.}, 
    year={1991},
    doi = {10.1017/cbo9780511840371},
}

@inproceedings{jaggi2013revisiting,
    author = {Jaggi, Martin},
    title = {Revisiting Frank-Wolfe: Projection-free sparse convex optimization},
    year = {2013},
    publisher = {JMLR.org},
    booktitle = {Proceedings of the 30th International Conference on International Conference on Machine Learning - Volume 28},
    pages = {427--435},
    url = {https://dl.acm.org/doi/10.5555/3042817.3042867}
}

@article{karimi2020primal,
    author = {Karimi, Mehdi and Tun{\c{c}}el, Levent},
    journal = {Mathematics of Operations Research},
    month = {2},
    number = {2},
    pages = {591--621},
    publisher = {INFORMS},
    title = {Primal--dual interior-point methods for domain-driven formulations},
    volume = {45},
    year = {2020},
    doi = {10.1287/moor.2019.1003}
}

@misc{karimi2023efficient,
    author = {Karimi, Mehdi and Tuncel, Levent},
    title = {Efficient implementation of interior-point methods for quantum relative entropy},
    year = {2023},
    archivePrefix = {arXiv},
    eprint = {2312.07438},
}

@misc{khanian2021rate,
    author = {Khanian, Zahra Baghali and Winter, Andreas},
    title = {A rate-distortion perspective on quantum state redistribution},
    year = {2021},
    archivePrefix = {arXiv},
    eprint = {2112.11952}
}

@inproceedings{khanian2023rate,
    author = {Khanian, Zahra Baghali and Kuroiwa, Kohdai and Leung, Debbie},
    booktitle = {2023 IEEE International Symposium on Information Theory},
    doi = {10.1109/isit54713.2023.10206960},
    month = {6},
    pages = {749--754},
    title = {Rate-distortion theory for mixed states},
    year = {2023}
}

@inproceedings{li2019blahut,
    author = {Li, Haobo and Cai, Ning},
    booktitle = {2019 IEEE International Symposium on Information Theory},
    journal = {International Symposium on Information Theory},
    month = {7},
    pages = {255--259},
    title = {A {B}lahut-{A}rimoto type algorithm for computing classical-quantum channel capacity},
    year = {2019},
    doi = {10.1109/isit.2019.8849608}
}

@article{lu2018relatively,
    author = {Lu, Haihao and Freund, Robert M and Nesterov, Yurii},
    doi = {10.1137/16M1099546},
    journal = {SIAM Journal on Optimization},
    month = {1},
    number = {1},
    pages = {333--354},
    publisher = {SIAM},
    title = {Relatively smooth convex optimization by first-order methods, and applications},
    volume = {28},
    year = {2018}
}

@article{maddison2021dual,
    author = {Maddison, Chris J and Paulin, Daniel and Teh, Yee Whye and Doucet, Arnaud},
    doi = {10.1137/19M130858X},
    journal = {SIAM Journal on Optimization},
    month = {1},
    number = {1},
    pages = {991--1016},
    publisher = {SIAM},
    title = {Dual space preconditioning for gradient descent},
    volume = {31},
    year = {2021}
}

@book{nemirovskij1983problem,
    author = {Nemirovskij, Arkadij Semenovi{\v{c}} and Yudin, David Borisovich},
    publisher = {Wiley},
    title = {Problem complexity and method efficiency in optimization},
    year = {1983},
    isbn = {0471103454}
}

@book{nielsen2010quantum, 
    place={Cambridge}, 
    title={Quantum computation and quantum information: 10th anniversary edition}, 
    publisher={Cambridge University Press}, 
    author={Nielsen, Michael A. and Chuang, Isaac L.}, 
    year={2010},
    doi = {10.1017/cbo9780511976667},
}

@book{nocedal1999numerical,
    author = {Nocedal, Jorge and Wright, Stephen J},
    doi = {10.1007/b98874},
    publisher = {Springer},
    title = {Numerical optimization},
    year = {1999}
}

@misc{qetlab,
    author = {Nathaniel Johnston},
    doi = {10.5281/zenodo.44637},
    howpublished = {\url{http://qetlab.com}},
    month = {1},
    title = {{QETLAB}: A {MATLAB} toolbox for quantum entanglement, version 0.9},
    year = {2016}
}

@article{ramakrishnan2020computing,
    author = {Ramakrishnan, Navneeth and Iten, Raban and Scholz, Volkher B and Berta, Mario},
    doi = {10.1109/tit.2020.3034471},
    journal = {IEEE Transactions on Information Theory},
    number = {2},
    pages = {946--960},
    publisher = {IEEE},
    title = {Computing quantum channel capacities},
    volume = {67},
    year = {2020}
}

@book{rockafellar1970convex,
    author = {Rockafellar, R Tyrrell},
    booktitle = {Lecture Notes in Control and Information Sciences},
    doi = {10.1007/bfb0110040},
    publisher = {Princeton University Press},
    title = {Convex analysis},
    volume = {18},
    year = {1970}
}

@inproceedings{schmidt2011convergence,
    author = {Schmidt, Mark and Roux, Nicolas and Bach, Francis},
    booktitle = {Proceedings of the 24th International Conference on Neural Information Processing Systems},
    journal = {Advances in Neural Information Processing Systems},
    month = {12},
    pages = {1458--1466},
    title = {Convergence rates of inexact proximal-gradient methods for convex optimization},
    volume = {24},
    year = {2011},
    url = {https://dl.acm.org/doi/10.5555/2986459.2986622}
}

@article{shannon1948mathematical,
    author = {Shannon, Claude Elwood},
    doi = {10.1002/j.1538-7305.1948.tb01338.x},
    journal = {The Bell System Technical Journal},
    month = {7},
    number = {3},
    pages = {379--423},
    publisher = {Nokia Bell Labs},
    title = {A mathematical theory of communication},
    volume = {27},
    year = {1948}
}

@article{solodov2000error,
    author = {Solodov, Mikhail V and Svaiter, Benar Fux},
    doi = {10.1007/s101070050022},
    journal = {Mathematical Programming},
    month = {8},
    number = {2},
    pages = {371--389},
    publisher = {Springer},
    title = {Error bounds for proximal point subproblems and associated inexact proximal point algorithms},
    volume = {88},
    year = {2000}
}

@article{teboulle2018simplified,
    author = {Teboulle, Marc},
    doi = {10.1007/s10107-018-1284-2},
    journal = {Mathematical Programming},
    month = {5},
    number = {1},
    pages = {67--96},
    publisher = {Springer},
    title = {A simplified view of first order methods for optimization},
    volume = {170},
    year = {2018}
}

@article{toh1999sdpt3,
    author = {Toh, Kim-Chuan and Todd, Michael J and T{\"u}t{\"u}nc{\"u}, Reha H},
    doi = {10.1080/10556789908805762},
    journal = {Optimization Methods and Software},
    number = {1-4},
    pages = {545--581},
    publisher = {Taylor \& Francis},
    title = {{SDPT3} — A {MATLAB} software package for semidefinite programming, version 1.3},
    volume = {11},
    year = {1999}
}

@techreport{tseng2008accelerated,
    author = {Tseng, Paul},
    title = {On accelerated proximal gradient methods for convex-concave optimization},
    year = {2008},
    url = {https://pages.cs.wisc.edu/~brecht/cs726docs/Tseng.APG.pdf}
}

@book{watrous2018theory,
    author = {Watrous, John},
    doi = {10.1017/9781316848142},
    publisher = {Cambridge University Press},
    title = {The theory of quantum information},
    year = {2018}
}

@article{wilde2013quantum,
    author = {Wilde, Mark M and Datta, Nilanjana and Hsieh, Min-Hsiu and Winter, Andreas},
    doi = {10.1109/tit.2013.2271772},
    journal = {IEEE Transactions on Information Theory},
    month = {10},
    number = {10},
    pages = {6755--6773},
    publisher = {IEEE},
    title = {Quantum rate-distortion coding with auxiliary resources},
    volume = {59},
    year = {2013}
}

@book{wilde2017quantum, 
    place={Cambridge}, 
    edition={2}, 
    title={Quantum information theory}, 
    publisher={Cambridge University Press}, 
    author={Wilde, Mark M.}, 
    year={2017},
    doi = {10.1017/9781316809976},
}

@article{yang2022bregman,
    author = {Yang, Lei and Toh, Kim-Chuan},
    journal = {SIAM Journal on Optimization},
    month = {7},
    number = {3},
    pages = {1523--1554},
    publisher = {SIAM},
    title = {Bregman proximal point algorithm revisited: A new inexact version and its inertial variant},
    volume = {32},
    year = {2022},
    doi = {10.1137/20M1360748}
}
